\documentclass[11pt]{article}

\usepackage[utf8]{inputenc}

\usepackage[margin=1in]{geometry}
\usepackage{booktabs} 
\usepackage[ruled]{algorithm2e} 

\SetAlFnt{\small}
\SetAlCapFnt{\small}
\SetAlCapNameFnt{\small}
\SetAlCapHSkip{0pt}
\IncMargin{-\parindent}

\usepackage{pifont}
\usepackage{nicefrac}
\usepackage[numbers]{natbib}

\usepackage{color}
\usepackage{wrapfig}
\usepackage{graphicx}
\usepackage{amsmath,amssymb,amsthm}
\usepackage{bm}
\usepackage{rotating}
\usepackage{multicol}
\usepackage{titlesec}
\usepackage{latexsym}
\usepackage{enumitem}
\usepackage{hyperref}[backref=page]
\usepackage{float}
\usepackage{caption}
\hypersetup{
     colorlinks   = true,
     linkcolor    = red, 
     urlcolor     = blue, 
	 citecolor    = blue 
}

\newcounter{newct}

\newcommand{\bv}{\begin{array}}

\newcommand{\x}{{ \bf x}}

\newcommand{\appProp}[3]{\vspace{1mm}\noindent{\bf Proposition~\ref{#1}.} {\bf (#2). }{\em #3 \vspace{1mm}}}

\newcommand{\appLem}[3]{\vspace{1mm}\noindent{\bf Lemma~\ref{#1}.} {\bf (#2). } {\em #3}}
\newcommand{\appThm}[3]{\vspace{1mm}\noindent{\bf Theorem~\ref{#1}.} {\bf (#2). }{\em #3 \vspace{1mm}}}

\newcommand{\appPropnoname}[2]{\vspace{1mm}\noindent{\bf Proposition~\ref{#1}.} {\em #2}}

\newtheorem{thm}{Theorem}
\newtheorem{dfn}{Definition}

\newtheorem{lem}{Lemma}
\newtheorem{ex}{Example}

\newtheorem{prop}{Proposition}
\newtheorem{coro}{Corrollary}

\newtheorem{claim}{Claim}

\newcommand{\vbb}{{\vec{b}}}

\newcommand{\piuni}{{\pi}_{\text{uni}}}
\newcommand{\ma}{\mathcal A}

\newcommand{\mm}{\mathcal M}

\newcommand{\ml}{\mathcal L}
\newcommand{\mV}{\mathcal V}

\newcommand{\calW}{\mathcal W}
\newcommand{\calG}{\mathcal G}
\newcommand{\calS}{\mathcal S}

\newcommand{\calI}{\mathcal I}
\newcommand{\ra}{\rightarrow}

\newcommand{\cor}{r} 

\newcommand{\rank}{\text{Rank}}
\newcommand{\sign}{\text{Sign}}

\newcommand{\pc}[2]{\text{PS}_{#1}(#2)}

\newcommand{\puts}{W}
\newcommand{\others}{\text{others}}

\newcommand{\kt}{\text{KT}}

\newcommand{\slt}[5]{\widetilde{\text{Tie}}^{\max}_{#1}(#2,#4,#5)}

\newcommand{\ilt}[5]{\widetilde{\text{Tie}}^{\min}_{#1}(#2,#4,#5)}

\newcommand{\expect}{{\mathbb E}}

\newcommand{\hist}{\text{Hist}}

\newcommand{\ind}{\text{Id}}

\newcommand\lirong[1]{{\color{red} \footnote{\color{red}Lirong: #1}} }

\newcommand{\Omit}[1]{}

\newcommand{\mc}{\text{MC}}

\newcommand{\md}[2]{{\dim_{#1,n}^{\max}}(#2)}
\newcommand{\mo}{\mathcal O}

\newcommand{\me}{\mathcal E}

\newcommand{\ba}{{\mathbf A}}
\newcommand{\be}{{\mathbf E}}
\newcommand{\bs}{{\mathbf S}}
\newcommand{\pba}[1]{{\mathbf A}^{#1}}
\newcommand{\pbe}[1]{{\mathbf E}^{#1}}
\newcommand{\pbs}[1]{{\mathbf S}^{#1}}

\newcommand{\vXp}{{\vec X}_{\vec\pi}}
\newcommand{\pvX}[1]{{\vec X}_{#1}}

\newcommand{\bd}{{\mathbf D}}
\newcommand{\bF}{{\mathbf F}}

\newcommand{\conv}{\text{CH}}
\newcommand{\cone}{\text{Cone}}

\newcommand{\GCD}[1]{\text{GCD}(#1)}
\newcommand{\wmg}{\text{WMG}}
\newcommand{\umg}{\text{UMG}}

\newcommand{\pair}{\text{Pair}}

\newcommand{\score}{\text{Score}}
\newcommand{\plu}{\text{Plu}}
\newcommand{\copeland}{\text{Cd}_\alpha}
\newcommand{\maximin}{\text{MM}}
\newcommand{\rp}{\text{RP}}
\newcommand{\schulze}{\text{Sch}}
\newcommand{\stv}{\text{STV}}
\newcommand{\coombs}{\text{Coombs}}

\newcommand{\eo}{\text{EO}}
\newcommand{\ties}{\text{Ties}}

\newcommand{\poly}{{\mathcal H}}
\newcommand{\polyn}{{\mathcal H}_{n}}
\newcommand{\polynint}{{\mathcal H}_n^{\mathbb Z}}
\newcommand{\polyz}{{\mathcal H}_{\leqslant 0}}

\newcommand{\ppoly}[1]{{\mathcal H}^{#1}}
\newcommand{\ppolyz}[1]{{\mathcal H}_{\leqslant 0}^{#1}}

\newcommand{\upoly}{{\mathcal C}}

\newcommand{\upolynint}{{\mathcal C}_n^{\mathbb Z}}
\newcommand{\upolyz}{{\mathcal C}_{\leqslant 0}}

\newcommand{\cpoly}[1]{{\mathcal H}_{#1}}

\newcommand{\cpolynint}[1]{{\mathcal H}_{#1,n}^{\mathbb Z}}
\newcommand{\cpolyz}[1]{{\mathcal H}_{#1,\leqslant 0}}

\newcommand{\invert}[1]{\left(#1\right)^\top}
\newcommand{\restrict}[2]{#1|_{#2}}

\newcommand{\pale}{\mo_\ma}
\newcommand{\palo}{\mo_\ma'}
\newcommand{\pal}[1]{\mo_\ma^{#1}}
\newcommand{\uge}{{\mathcal G}_\ma}
\newcommand{\ugo}{{\mathcal G}_\ma'}
\newcommand{\ug}[1]{{\mathcal G}_\ma^{#1}}

\begin{document}
\setcounter{page}{0}
\title{
How Likely Are Large Elections Tied?} 

\author{ Lirong Xia, 
Rensselaer Polytechnic Institute, Troy, NY 12180, USA,\\
xialirong@gmail.com }
\date{}

\maketitle
\begin{abstract} 
Understanding the likelihood for an election to be tied is a classical topic in many disciplines including social choice, game theory, political science, and public choice. The problem is important not only as a fundamental problem in probability theory and statistics, but also because  it plays a  critical role in many other important issues such as indecisiveness of voting, strategic voting, privacy of voting, voting power, voter turnout, etc. Despite a large body of literature and the common belief that ties are rare, little is known about how rare ties are in large elections except for a few simple positional scoring rules under the i.i.d.~uniform distribution over the votes, known as the {\em Impartial Culture (IC)} in social choice. In particular, little progress was made after Marchant explicitly posed the likelihood of $k$-way ties under IC as an open question in 2001~\citep{Marchant2001:The-probability}.

We give an asymptotic answer to this open question for a wide range of commonly studied voting rules under a more general and realistic model, called the {\em smoothed social choice framework}~\citep{Xia2020:The-Smoothed}, which was inspired by the celebrated smoothed complexity analysis~\citet{Spielman2009:Smoothed}. We prove dichotomy theorems on the smoothed likelihood of ties under   positional scoring rules, edge-order-based rules, and some multi-round score-based elimination rules, which include commonly studied voting rules such as plurality, Borda, veto, maximin, Copeland, ranked pairs, Schulze, STV, and Coombs  as special cases. We also complement the theoretical results by experiments on synthetic data  and real-world rank  data on Preflib~\citep{Mattei13:Preflib}.  Our main technical tool is an improved  characterization of the smoothed likelihood for a Poisson multinomial variable to be in a polyhedron,  by exploring the interplay between the V-representation and the matrix representation of polyhedra and might be of independent interest. 

\end{abstract}

\newpage
\setcounter{page}{1}
\section{Introduction}

Suppose a presidential election between two alternatives (candidates) $a$ and $b$ will be held soon, and there are  $n$ agents (voters). Each agent independently votes for an alternative with probability $0.5$, and the alternative with more votes wins. How likely will the election end up with a tie? What if there are more than two alternatives, agents rank the alternatives, and a rank-based voting rule is used to choose the winner? What if the distribution is not independent and identically distributed (i.i.d.)~and is controlled by an adversary?


Understanding the likelihood of tied elections is an important and classical topic in many disciplines including social choice, game theory, political science, and public choice, not only because it is a fundamental problem in probability theory and statistics, but also because it plays a critical role in many important issues. For example, ties are undesirable in the context of indecisiveness of voting~\cite{Gillett1977:Collective}, strategic voting~\cite{Gibbard73:Manipulation,Satterthwaite75:Strategy}, privacy of voting~\cite{Liu2020:How}, etc. On the other hand, ties are desirable in the context of voting power~\cite{Banzhaf-III1968:One-Man}, voter turnout~\citep{Downs1957:An-Economic,Riker1968:A-Theory}, etc.

While the likelihood of ties for two alternatives is well-understood~\citep{Banzhaf-III1968:One-Man,Beck75:Note}, we are not aware of a rigorous mathematical analysis for elections with three or more alternatives except for a few simple voting rules.
Previous studies were mostly done along three dimensions: (1) the voting rule used in the election, (2) the indecisiveness of the outcome, measured by the number of tied alternatives $k$, and (3) the statistical model for generating votes.   See Section~\ref{sec:related-work} for more discussions.

Despite these efforts, the following question  largely remains open. 

\vspace{2mm}
{\em \hfill How likely are large elections tied under realistic models?\hfill}
\vspace{2mm}

Specifically, Marchant~\cite{Marchant2001:The-probability} posed the likelihood of ties beyond certain positional scoring rules under the i.i.d.~uniform distribution, known as the {\em Impartial Culture (IC)} in social choice, as an open question in 2001, but we are not aware of any progress afterwards. While IC has been a popular choice in social choice theory, it has also been widely criticized of being unrealistic~\citep{Lehtinen2007:Unrealistic}.

In fact, the question is already highly challenging  under IC as illustrated in Example~\ref{ex:intro} below. Consider the probability of $3$-way ties under the Borda rule for $3$ alternatives and $n$ agents. Borda is a {positional scoring rule}, which scores every alternative according to its rank. Under Borda, each agent uses a linear order over the alternatives to represent his/her preferences, and the $i$-th ranked alternative gets $m-i$ points. The winners are the alternatives with maximum total points. 

\begin{ex}
\label{ex:intro}
Let $\ma = \{1,2,3\}$ denote the set of alternatives. For each linear order $R$ over $\ma$, let $X_R$ denote the random variable that represents the number of agents whose votes are $R$, when their votes are generated uniformly at random (i.e., IC). For example, $X_{123}$ represents the multiplicity of $1\succ 2\succ 3$. Then, the probability of $3$-way ties under Borda  w.r.t.~IC can be represented by the probability for the following system of linear equations to hold, where (\ref{eq:firstcond}) states that alternatives $1$ and $2$ are tied, (\ref{eq:secondcond}) states that $2$ and $3$ are tied, and (\ref{eq:thirdcond}) states that $1$ and $3$ are tied.
\vspace{-1mm}
\begin{align}
&2X_{123}+2 X_{132} + X_{213}+ X_{312} = 2X_{213}+2 X_{231} + X_{123}+ X_{321}\label{eq:firstcond}\\
& 2X_{213}+2 X_{231} + X_{123}+ X_{321} = 2X_{312}+2 X_{321} + X_{132}+ X_{231} \label{eq:secondcond}\\
&2X_{123}+2 X_{132} + X_{213}+ X_{312} = 2X_{312}+2 X_{321} + X_{132}+ X_{231} \label{eq:thirdcond}
\end{align}
\end{ex} 

The difficulty in accurately bounding the likelihood of ties comes from two types of statistical correlations.  The first type consists of correlations among components of $X$. That is, for any pairs of linear orders $R$ and $W$, $X_R$ and $X_W$ are statistically dependent. The second type consists of correlations among equations, and more generally, inequalities as we will see in the general problem studied in this paper. For example, while it is straightforward to see that (\ref{eq:firstcond}) and (\ref{eq:secondcond})  implies (\ref{eq:thirdcond}) in Example~\ref{ex:intro}, it is unclear how much correlation exists between (\ref{eq:firstcond}) and (\ref{eq:secondcond}). 
Existing asymptotic tools  such as multivariate Central Limit Theorems and Berry-Esseen-type Theorems~\cite{Bentkus2005:A-Lyapunov-type,Valiant2011:Estimating,Daskalakis2016:A-Size-Free,Diakonikolas2016:The-fourier,Raic2019:A-multivariate} (a.k.a.~Lyapunov-type bounds) contain  an $O(n^{-0.5})$ or higher error bound, which are too coarse and do not match the lower bound that will be proved in this paper.  The problem becomes more challenging for  inequalities, other voting rules, other number of alternatives, other $k$'s, and non-i.i.d.~distributions over votes.

\paragraph{\bf The Model.} We address the likelihood of ties under the {\em smoothed social choice} framework~\cite{Xia2020:The-Smoothed}, which is inspired by the celebrated {\em smoothed complexity analysis}~\cite{Spielman2009:Smoothed}. We believe that the framework is   more general and realistic than the  extensively studied i.i.d.~models, especially IC. In the framework, agents' ``ground truth'' preferences can be arbitrarily correlated and are chosen by an adversary, and then independent noises are added to form their votes. Mathematically,  the adversary chooses a distribution $\pi_j$ for each agent $j$ from a set $\Pi$ of distributions  over all linear orders over the alternatives, under which the probability of various events of interest are studied, for example Condorcet's paradox and satisfaction of axioms~\cite{Xia2020:The-Smoothed}.


\paragraph{\bf Our Contributions.}  In this paper, we adopt the statistical model in~\cite{Xia2020:The-Smoothed} to formulate and study the smoothed likelihood of ties.  Given an (irresolute) voting rule $\cor$,  $2\le k\le m$, and $n\in\mathbb N$ agents, the {\em max-adversary} aims to maximize the likelihood of $k$-way ties, denoted by $\slt{\Pi}{\cor}{m}{k}{n}$, by choosing  $\vec \pi = (\pi_1,\ldots,\pi_n)\in \Pi^n$. Formally,  
\begin{equation}
\label{equ:max-sTies}
\slt{\Pi}{\cor}{m}{k}{n} \triangleq \sup\nolimits_{\vec \pi\in\Pi^n}\Pr\nolimits_{P\sim\vec \pi}\left(|\cor(P)|=k\right)
\end{equation}
Similarly, the {\em min-adversary} aims to minimize the likelihood of $k$-way ties defined as follows:
\begin{equation}
\label{equ:min-sTies}\ilt{\Pi}{\cor}{m}{k}{n} \triangleq \inf\nolimits_{\vec \pi\in\Pi^n}\Pr\nolimits_{P\sim\vec \pi}\left(|\cor(P)|=k\right)
\end{equation}

We call $\slt{\Pi}{\cor}{m}{k}{n}$ (respectively, $\ilt{\Pi}{\cor}{m}{k}{n}$)  the {\em max} (respectively, {\em min}) {\em smoothed likelihood of ties}. When $\Pi$ consists of a single distribution $\pi$, $\slt{\Pi}{\cor}{m}{k}{n}$ and $\ilt{\Pi}{\cor}{m}{k}{n}$ coincide with each other and become the likelihood of ties under i.i.d.~distribution $\pi$. In particular, when $\Pi=\{\pi_{\text{uni}}\}$, where $\pi_{\text{uni}}$ is  the uniform distribution,  $\slt{\Pi}{\cor}{m}{k}{n}$ and $\ilt{\Pi}{\cor}{m}{k}{n}$  become the classical analysis of ties under IC. As discussed in~\citep{Xia2020:The-Smoothed}, the smoothed social choice framework allows agents' ground truth preferences to be arbitrarily correlated, while the noises are independent, which is a standard assumption in many literatures such as behavior science, economics, statistics, and smoothed complexity analysis.

Our main technical results are asymptotic characterizations of the smoothed likelihood of ties for a fixed number of at least three alternatives ($m\ge 3$) in large elections ($n\ra\infty$). Informally, our main results can be summarized as follows.


\vspace{2mm}\noindent{\bf Main results: smoothed likelihood of ties, informally put. }{\em Under mild assumptions on $\Pi$, for many commonly studied voting rules $r$, for any fixed $m\ge 3$, any  $2\le k\le m$,  and any $n\in \mathbb N$, $\slt{\Pi}{\cor}{m}{k}{n}$ (respectively, $\ilt{\Pi}{\cor_{\vec s}}{m}{k}{n}$) is either $0$, $\exp(-\Theta(n))$, or $\Theta(\text{poly}^{-1}(n))$.
\vspace{2mm}}

More precisely, we prove
 Theorem~\ref{thm:score} for integer positional scoring rules (including plurality, Borda, veto), Theorem~\ref{thm:eorule} for edge-order-based rules (including maximin, Schulze, ranked pairs, and Copeland, see Definition~\ref{dfn:eorule}), and Theorem~\ref{thm:STV-Coombs} for  STV and Coombs.  The formal statements of the theorems also characterize the condition for each ($0$, exponential, or polynomial) case as well as asymptotically tight bounds in the polynomial cases. 

When ties are undesirable, e.g., in the context of indecisiveness of voting, strategic voting, or privacy, a low max smoothed likelihood is good news. When ties are desirable, e.g., in the context of voting power and voter turnout, a high min smoothed likelihood is good news. Our theorems therefore completely characterize conditions for good news in different contexts.

Straightforward applications of our theorems answer the open question by Marchant~\cite{Marchant2001:The-probability} for many commonly studied voting rules as summarized in Table~\ref{tab:corollaries} below. 

\begin{table}[htp]
\centering
\caption{\small Probability of $k$-way ties ($2\le k\le m$) under some commonly studied voting rules w.r.t.~IC. For Copeland$_\alpha$, $l_\alpha$ is the minimum positive integer s.t.~$\alpha l_\alpha \in\mathbb Z$.\label{tab:corollaries}}
\resizebox{\textwidth}{!}
{
\begin{tabular}{ |@{}c@{}|}
\begin{tabular}{|@{}c@{}|@{\ }c@{\ }|@{}c@{}|}
\hline \begin{tabular}{@{}l@{}}\bf\boldmath Int.~Pos.~scoring (Coro.~\ref{coro:scoringtie})\\\bf\boldmath STV and Coombs (Prop.~\ref{prop:stv-Coombs}) \end{tabular} & \begin{tabular}{@{}l@{}}\bf\boldmath  maximin (Prop.~\ref{prop:maximin})\\ 
 \bf\boldmath  Schulze (Prop.~\ref{prop:schulze}) \end{tabular}& \bf\boldmath  ranked pairs (Prop.~\ref{prop:rp})\\
\hline 
$\left\{\begin{array}{@{}l l@{}}0 &\text{\begin{tabular}{@{}l}if no profile of $n$ votes\\ contains a  $k$-way tie\end{tabular}} \\
\Theta(n^{-\frac{k-1}{2}}) &\text{otherwise}
\end{array}\right.$&  \begin{tabular}{@{}c@{}}$\Theta(n^{-\frac{k-1}{2}})$ \end{tabular} &\begin{tabular}{@{}c@{}|@{\ }c@{\ }}
Lower & Upper\\
\hline
$\left\{\text{\begin{tabular}{@{}l@{}}
$\Omega(n^{-\frac{k-1}{2}})$ \\
$\Omega(n^{-\frac{\lceil \log k\rceil}{2}})$ if 
$m\ge k+5 \lceil \log k\rceil$
\end{tabular}}\right.$
&  $n^{-\Omega(\frac{\log k}{\log\log k})}$
\end{tabular}
\end{tabular}
\\
\hline \hline
\bf \boldmath Copeland$_\alpha$ ($0\le \alpha\le 1$) (Prop.~\ref{prop:copeland}) \\
\hline
$\left\{\begin{array}{ll}
0 &\text{if } 2\nmid n, 2\mid k, \text{and }k\ge m-1\\
\Theta(n^{-\frac{k}{4}}) &\text{if } 2\mid n,  2\mid  k,  { and }\left\{\begin{array}{l}(1)\ k=m, \text{ or} \\ 
(2)\ k=m-1 \text{ and } \alpha\ge \frac 12, \text{ or}\\
(3)\ k=m-1 \text{ and } k\le l_\alpha(l_\alpha+1)
\end{array}\right.\\
\Theta\left(n^{-\frac{l_\alpha(l_\alpha+1)}{4}}\right) &\text{if } 2\mid n,  2\mid  k,  k=m-1, \alpha<\frac12,\text{ and }k> l_\alpha(l_\alpha+1)\\
\Theta(1) &\text{otherwise (i.e., if }2\nmid k \text{ or }k\le m-2\text{)}
\end{array}\right.$\\
\hline
\end{tabular}
}
\end{table}

Roughly speaking, Table~\ref{tab:corollaries} reveals the following  ranking over the voting rules w.r.t.~their likelihood of $k$-way ties under IC, for every $2\le k\le m$.
$$\{\text{Int. Pos. Scoring, STV, Coombs}\}\le \{\text{maximin, Schulze}\}\le \text{ranked pairs}\le \text{Copeland}$$


A  closely related question  is the likelihood of {\em any-way} ties, sometimes referred to as {\em indecisiveness}~\citep{Marchant2001:The-probability},  under IC, i.e., the election admits a $k$-way tie for any $2\le k\le m$. It is not hard to see from Table~\ref{tab:corollaries} that such likelihood is dominated by the probability of $2$-way ties and is either $0$ or $\Theta(\frac{1}{\sqrt n})$ for all rules in the table except ranked pairs and Copeland, which are covered by Proposition~\ref{prop:rp}  and Proposition~\ref{prop:copeland}, respectively. To the best of our knowledge, these results are new, except for plurality and Borda. Experiments on synthetic data generated from IC confirm these observations, while experiments on Preflib data~\cite{Mattei13:Preflib} reveal a difference order, where ties are rare under Borda (1.6\%) and Copeland (2.6\%), and are quite common under veto (31.3\%) due to situations where $m>n$.

\paragraph{\bf Technical Innovations.} 
The proofs of the smoothed likelihood of ties in this paper follow the same high-level idea. We first model the existence of a  $k$-way tie by systems of linear inequalities that are similar to the ones in Example~\ref{ex:intro}. In this way, the likelihood of  ties becomes the likelihood for the histogram of the randomly generated profile, which is a {\em Poisson multivariate variable (PMV)}, to be in the polyhedron $\poly$  represented by the linear inequalities. 
Then, we prove a dichotomous characterization (Theorem~\ref{thm:maintechh})  for a PMV to be in $\poly$, and finally apply Theorem~\ref{thm:maintechh} (more precisely, its extension Theorem~\ref{thm:union-poly} to unions of multiple polyhedra) to characterize the smoothed likelihood of ties.  

More precisely, given $n,q\in\mathbb N$ and a vector $\vec \pi=(\pi_1,\ldots,\pi_n)$ of $n$ distributions over  $\{1,\ldots,q\}$, an $(n,q)$-PMV is denoted by $\vXp$, which represents the histogram of $n$ independent random variables whose distributions are $\{\pi_1,\ldots,\pi_n\}$, respectively.

\appThm{thm:maintechh}{The  PMV-in-polyhedron theorem, informally put} { Let $\poly$ denote a polyhedron and $\Pi$ denote a set of distributions that satisfy some mild conditions,  for any $n\in\mathbb N$, 

\hfill\begin{tabular}{l}
$\sup_{\vec\pi\in\Pi^n}\Pr (\vXp \in \poly )$ is $0$, $\exp(-\Theta(n))$, or $\Theta(\text{poly}(n))$, and\\
$\inf_{\vec\pi\in\Pi^n}\Pr (\vXp \in \poly )$ is $0$, $\exp(-\Theta(n))$, or $\Theta(\text{poly}(n))$
\end{tabular}\hfill
}

The bounds are asymptotically tight and the formal statement of the theorem also characterizes  conditions for the $0$, exponential, and polynomial cases, respectively.  As commented after Example~\ref{ex:intro}, we do not see a way to prove Theorem~\ref{thm:maintechh} by straightforward applications of existing asymptotic tools.  We also believe that Theorem~\ref{thm:maintechh} is a useful tool to study the smoothed likelihood of many events of interest in  social choice as commented in~\cite{Xia2020:The-Smoothed}. 

\subsection{Related Work and Discussions} 
\label{sec:related-work}
\noindent{\bf Ties in elections.}  The importance of estimating likelihood of ties has been widely acknowledged, for example, as~Mulligan and Hunter commented:
{\em ``Perhaps it is common knowledge that civic elections are not often decided by one vote...a precise calculation of the frequency of a pivotal vote can contribute to our understanding of how many, if any, votes might be rationally and instrumentally cast''}~\cite{Mulligan2003:The-Empirical}. 
In practice, however, ties are not as rare as commonly believed, even in high-stakes elections, and have led to pitfalls and consequent modifications of electoral systems and constitutional laws. For example, in the 1800 US presidential election, Jefferson and Burr tied in the electoral college votes. By the Constitution, the House of Representatives  should vote until a candidate wins the majority. However, in the subsequent 35 rounds of deadlocked voting, none of the two candidates got the majority. Eventually, Jefferson won the 36th revote to become the president. This {\em ``had demonstrated a fundamental flaw with the Constitution. As a result, the Twelfth Amendment to the Constitution was introduced and ratified''}~\cite{Campbell2015:Political}. 

\paragraph{\bf Three-or-more-way ties.}   Nowadays, legislators are well-aware of the possibility of two-way ties in elections and have specified tie-breaking mechanisms to handle them. However, three-or-more-way ties have not received their deserved attention and are sometimes overlooked. 
For example, the 2019 Code of Alabama Section 17-12-23 states: {\em ``In all elections where there is a tie between the two highest candidates for the same office, for all county or precinct offices, it shall be decided by lot by the sheriff of the county in the presence of the candidates''}. The Code does not specify what action should be taken when three or more candidates are tied for the first place. Results in this paper  characterize how rare this happens, so that the legislators can make an informed decision about whether the loophole is a significant concern in practice, and in case it is, how to fix it.

\paragraph{\bf Smoothed analysis.}  There is a large body of literature on the applications of smoothed analysis to mathematical programming, machine learning, numerical analysis, discrete math, combinatorial optimization, etc., see~\cite{Spielman2009:Smoothed} for a survey. Smoothed analysis has also been applied to various problems in economics, for example  price of anarchy~\cite{Chung2008:The-Price} and market equilibrium~\cite{Huang2007:On-the-Approximation}. 
 In a recent position paper, Baumeister, Hogrebe, and Rothe~\cite{Baumeister2020:Towards}  proposed a Mallows-based model to conduct smoothed analysis on computational aspects of social choice and commented that the model can be used to analyze voting paradoxes and ties, but the paper does not contain technical results.  \cite{Xia2020:The-Smoothed} independently proposed to conduct smoothed analysis for paradoxes and impossibility theorems in social choice,  characterized the smoothed likelihood of Condorcet's voting paradox and the ANR impossibility theorem, and proposed a new tie-breaking mechanism. We only use the probabilistic model in~\cite{Xia2020:The-Smoothed} and topic-wise, our paper is different from \cite{Xia2020:The-Smoothed}, because we formulate and study the smoothed likelihood of ties under commonly studied voting rules, which was not studied in~\cite{Xia2020:The-Smoothed}. 

\paragraph{\bf Technical novelty.}  We believe that the main technical tool of this paper (Theorem~\ref{thm:maintechh}) is a significant and non-trivial extension of Lemma~1 in~\cite{Xia2020:The-Smoothed} to arbitrary polyhedron represented by an integer matrix, every $n$, and the min-adversary. More  discussions can be found in the remark after Theorem~\ref{thm:maintechh}. 
 We believe that Theorem~\ref{thm:maintechh} is a useful tool to analyze smoothed likelihood of many other problems of interest in social choice. For example, all results in~\cite{Xia2020:The-Smoothed} can be immediately strengthened by  Theorem~\ref{thm:maintechh}.


\paragraph{\bf Previous work on likelihood of ties.} The following table summarizes previous works that are closest to this paper, whose main contributions are characterizations of likelihood of ties.

\begin{table}[htp]
\centering
\begin{tabular}{|c|c|c|c|c|}
\hline \bf Paper & $\bm m$ & $\bm   k$ &\bf  Voting rule&\bf  Distribution\\
\hline \citep{Beck75:Note}& $2$ & $2$& majority& two groups, i.i.d.~within each group\\
\hline \begin{tabular}{@{}c@{}}\citep{Margolis1977:Probability}\\ \citep{Chamberlain1981:A-note}\end{tabular} & $2$ & $2$& majority& i.i.d.~w.r.t.~an uncertain distribution\\
\hline \citep{Gillett1977:Collective} & any $m\in\mathbb N$ & $2\le k\le m$& plurality& uniformly i.i.d.~(IC)\\
\hline \citep{Gillett1980:The-Comparative} & $3$ & $2\le k\le m$&  Borda& uniformly i.i.d.~(IC)\\
\hline \citep{Marchant2001:The-probability} & any $m\in\mathbb N$  & $k=m$&  certain scoring rules& uniformly i.i.d.~over all score vectors\\
\hline
\end{tabular}
\end{table}

More precisely, Beck~\citep{Beck75:Note} studied the probability of ties under the majority rule (over two alternatives) with two groups of agents, whose votes are i.i.d.~within each group. Margolis~\citep{Margolis1977:Probability} and Chamberlain and Rothschild~\citep{Chamberlain1981:A-note} focused on the majority rule for two alternatives, where agents' preferences are i.i.d.~according to a randomly generated distribution.  
Gillett~\citep{Gillett1977:Collective} studied probability of all-way ties  (i.e., $k=m$)  under plurality w.r.t.~IC for arbitrary numbers of alternatives and agents. Gillett~\citep{Gillett1980:The-Comparative} obtained a closed-form formula for Borda indecisiveness (two or more alternatives being tied) for $m=3$ w.r.t.~IC.  Marchant~\citep{Marchant2001:The-probability} considered a class of scoring rules where each agent can choose a score vector from a given {\em scoring vectors set (SVS)}, and characterized the asymptotic 
probability of $m$-way ties under a class of scoring rule to be $\Theta(n^{\frac{1-m}{2}})$ w.r.t.~the i.i.d.~uniform distribution over all SVS. This result can be applied to Borda and approval voting but cannot be applied to plurality. Marchant~\citep{Marchant2001:The-probability} also noted  that Domb~\citep{Domb1960:On-the-theory} obtained equivalent formulas for $m=3$ under Borda. As discussed earlier, the smoothed social choice framework used in our paper is more general. In particular, corollaries of  our theorems  answer the open questions proposed by Marchant~\citep{Marchant2001:The-probability} and reveal a ranking over these rules according to the likelihood of ties under IC  as summarized in Table~\ref{tab:corollaries}.

\paragraph{Previous work related to likelihood of ties.}  \citep{LeBreton2016:Correlation} studied the setting where the agents are partitioned into multiple groups, and within each group, agents' votes are generated from the {\em impartial anonymous culture (IAC)} model.  The smoothed  social choice framework and IAC are not directly comparable. The former is more general  in the sense that agents' ``ground truth" preferences are arbitrarily correlated. The latter is more general in the sense that agents' ``noises'' are not independent. There is also a line of empirical and mixed empirical-theoretical work on the likelihood of ties under the US electoral college system~\cite{Gelman1998:Estimating,Gelman2012:What}. Studying the smoothed likelihood of ties under these settings are left for future work. 


The probability of tied elections is closely related to the probability for a single voter to be pivotal, sometimes called {\em voting power}, which plays an important role in the paradox of voting~\cite{Downs1957:An-Economic} and in definitions of power indices in cooperative game theory. There is a large body of work on voting games, where the probability for a voter to be pivotal, which is equivalent to the likelihood of ties among other voters, plays a central role in the analysis of voters' strategic behavior. Examples include seminal works~\citep{Austen96:Information,Feddersen96:Swing,Myerson2000:Large}, and more recent work~\citep{Nunez2019:Truth-revealing}.
 It is not hard to see that the voting power for two alternatives or for multiple alternatives under the plurality rule almost equals to the probability of tied elections with one less vote under certain tie-breaking rules, as pointed out in~\citep{Good1975:Estimating}. For three or more alternatives the two problems are closely related but technically different, which we leave for future work. 

The likelihood of ties is also related to the manipulability of voting rules~\cite{Conitzer16:Barrier}---if an election is not tied, then no single agent can change the outcome, therefore no agent alone has incentive to cast a manipulative vote. Our results are related to but different from the typical-case analysis of manipulability in the literature~\cite{Procaccia07:Junta,Xia08:Generalized,Mossel13:Smooth,Xia15:Generalized} and the quantitative Gibbard-Satterthwaite theorem~\cite{Friedgut2011:A-quantitative,Mossel2015:A-quantitative},  where votes are assumed to be i.i.d. Likelihood of ties are also related to but different from the margin of victory~\cite{Magrino11:Computing,Xia12:Computing} and more broadly, bribery and control in elections~\cite{Faliszewski2016:Control}.

\paragraph{Adding noise to study ties.} We are not aware of a previous work that characterized the smoothed likelihood of ties as we do in this paper. The idea of adding noise to study ties is not new. In the definition of {\em resolvability} in \citep{Tideman87:Independence}, an additional vote is added to break ties. In \citep{Freeman2015:General}, irresolute voting rules were defined by taking the union of winners under profiles around a given profile. 
Another {\em resolvability} studied in the literature (see e.g., Formulation\#1 in Section 4.2 of~\cite{Schulze11:New}\footnote{Wikipedia~\cite{wikiresolvability} contributes this definition to Douglas R.~Woodall but we were not able to find a formal reference.}) requires that the probability of ties goes to $0$ under the voting rule w.r.t.~IC, which is closely related to the literature in the indecisiveness of voting.  Our setting and results are more general because IC is  a special case of the smoothed social choice framework, and our results also characterize the rate of convergence.

\paragraph{Computational aspects of tie-breaking.}  There is a large body of recent work on computational aspects of tie-breaking. \citep{Conitzer09:Preference} proposed the {\em parallel universe tie-breaking (PUT)} for multi-stage voting rules and characterized the complexity of the STV rule. \citep{Brill12:Price} characterize the complexity of PUT under ranked pairs, whose smoothed likelihood of ties is studied in this paper. \citep{Mattei2014:How-hard} characterized complexity of PUT  under other multi-stage voting rules such  as Baldwin and Coombs. \citep{Obraztsova11:Ties,Obraztsova11:Complexity,Aziz2013:Ties,Obraztsova2013:On-manipulation} investigated the effect of different tie-breaking mechanisms to the complexity of manipulation. \citep{Freeman2015:General} propose a general way of defining ties under generalized scoring rules. \citep{Wang2019:Practical} proposed practical AI algorithms for computing PUT under STV and ranked pairs.



\section{Preliminaries}
\label{sec:prelim}
\noindent{\bf Basic Setting.} 
For any  $q\in\mathbb N$, we let $[q]=\{1,\ldots,q\}$. Let $\ma=[m]$ denote the set of $m\ge 3$ {\em alternatives}. Let $\ml(\ma)$ denote the set of all linear orders over $\ma$. Let $n\in\mathbb N$ denote the number of agents. Each agent uses a linear order to represent his or her preferences, called a {\em vote}. The vector of $n$ agents' votes, denoted by $P$, is called a {\em (preference) profile}, sometimes called an $n$-profile. A {\em fractional} profile is a preference profile $P$ together with a possibly non-integer and possibly negative weight vector $\vec \omega_P=(\omega_R:R\in P)\in{\mathbb R}^{n}$ for the votes in $P$. It follows that a non-fractional profile is a fractional profile with uniform weight, namely $\vec \omega_P = \vec 1$.  Sometimes we slightly abuse the notation by omitting the weight vector  when it is clear from the context or when $\vec\omega_P=\vec 1$. 

For any (fractional) profile $P$, let $\hist(P)\in {\mathbb Z}_{\ge 0}^{m!}$ denote the anonymized profile of $P$, also called the {\em histogram} of $P$, which contains the total weight of every  linear order in $\ml(\ma)$ according to $P$.  An {\em (irresolute) voting rule} $\cor$ is a
mapping from a profile to a non-empty set of winners in $\ma$. Below we recall the definitions of several commonly studied voting rules. 

\paragraph{\bf Integer positional scoring rules.}  An {\em (integer) positional scoring rule}  is characterized by an integer scoring vector $\vec s=(s_1,\ldots,s_m)\in{\mathbb Z}^m$ with $s_1\ge s_2\ge \cdots\ge s_m$ and $s_1>s_m$. For any alternative $a$ and any linear order $R\in\ml(\ma)$, we let $\vec s(R,a)=s_i$, where $i$ is the rank of $a$ in $R$. Given a profile $P$ with weights $\vec \omega_P$,  the  positional scoring rule $\cor_{\vec s}$ chooses all alternatives $a$ with maximum $\sum_{R\in P}\omega_R\cdot s(R,a)$. For example, {\em plurality} uses the scoring vector $(1,0,\ldots,0)$, {\em Borda} uses the scoring vector $(m-1,m-2,\ldots,0)$, and {\em veto} uses the scoring vector $(1,\ldots,1,0)$. 

\paragraph{\bf Weighted Majority Graphs.}  For any (fractional) profile $P$ and any pair of alternatives $a,b$, let $ P[a\succ b]$ denote the total weight of votes in $P$ where $a$ is preferred to $b$. Let $\wmg(P)$ denote the {\em weighted majority graph} of $P$, whose vertices are $\ma$ and whose weight on edge $a\ra b$ is $w_P(a,b) = P[a\succ b] - P[b\succ a]$. Sometimes a distribution $\pi$ over $\ml(\ma)$ is viewed as a fractional profile, where for each $R\in\ml(\ma)$ the weight on $R$ is $\pi(R)$. In this case we let $\wmg(\pi)$ denote the weighted majority graph of the {fractional} profile represented by $\pi$. 

A voting rule is said to be {\em weighted-majority-graph-based (WMG-based)} if its winners only depend on the WMG of the input profile. In this paper we consider the following commonly studied WMG-based rules.
\begin{itemize}
\item {\bf Copeland.} The Copeland rule is parameterized by a number $0\le \alpha\le 1$, and is therefore denoted by Copeland$_\alpha$, or $\copeland$ for short. For any fractional profile $P$, an alternative $a$ gets $1$ point for each other alternative it beats in their head-to-head competition, and gets $\alpha$ points for each tie. Copeland$_\alpha$ chooses all alternatives with the highest total score as the winners. 
\item  {\bf Maximin.} For each alternative $a$, its min-score  is defined to be $\min_{b\in\ma}w_P(a,b)$. Maximin, denoted by $\maximin$, chooses all alternatives with the max min-score as the winners.
\item {\bf Ranked pairs.} Given a profile $P$, an alternative $a$ is a winner under ranked pairs (denoted by $\rp$) if there exists a way to fix edges in $\wmg(P)$ one by one in a non-increasing order w.r.t.~their weights (and sometimes break ties), unless it creates a cycle with previously fixed edges, so that after all edges are considered, $a$ has no incoming edge. This is known as the {parallel-universes tie-breaking (PUT)}~\citep{Conitzer09:Preference}.
\item {\bf Schulze.} For any directed path in the WMG, its strength is defined to be the minimum weight on any single edge along the path. For any pair of alternatives $a,b$, let $s[a,b]$ be the highest weight among all paths from $a$ to $b$. Then, we write $a\succeq b$ if and only if $s[a,b]\ge s[b,a]$, and~\citet{Schulze11:New} proved that the strict version of this binary relation, denoted by $\succ$, is transitive. The Schulze rule, denoted by $\schulze$, chooses all alternatives $a$ such that for all other alternatives $b$, we have $a\succeq b$. 
\end{itemize}

\noindent{\bf Multi-round score-based elimination (MRSE) rules.} Another large class of voting rules studied in this paper select the winner(s) in $m-1$ rounds. In each round, an integer positional scoring rule is used to rank the remaining alternatives, and a {\em loser} (an alternative with the minimum total score) is removed from the election. Like in ranked pairs, PUT is used to select winners---an alternative $a$ is a winner if there is a way to break ties among  losers so that $a$ is the remaining alternative after $m-1$ rounds. For example, the STV rule uses the plurality rule in each round and the Coombs rule uses the veto rule in each round.

We now recall the  statistical model used in the smoothed social choice framework~\cite{Xia2020:The-Smoothed}.


\begin{dfn}[\bf Single-Agent Preference Model~\cite{Xia2020:The-Smoothed}]
\label{dfn:pref-model}
A {\em single-agent preference model} is denoted by $\mm=(\Theta,\ml(\ma),\Pi)$, where $\Theta$ is the parameter space, $\ml(\ma)$ is the sample space, and $\Pi$ consists of distributions indexed by $\Theta$. $\mm$ is {\em strictly positive} if there exists $\epsilon>0$ such that the probability of any linear order under any distribution in $\Pi$ is at least $\epsilon$.  
$\mm$ is {\em closed} if $\Pi$ (which is a subset of the probability simplex in $\mathbb R^{m!}$) is a closed set in $\mathbb R^{m!}$. 
\end{dfn}


For example, given $0<\underline{\varphi}\le \overline{\varphi}\le 1$, in the {\em single-agent Mallows model}~\cite[Example~2 in the appendix]{Xia2020:The-Smoothed}, denoted by $\mm_{[\underline{\varphi},\overline{\varphi}]}$, we have $\Theta= \ml(\ma)\times [\underline{\varphi},\overline{\varphi}]$. For any $\varphi\in [\underline{\varphi},\overline{\varphi}]$ and any $R\in\ml(\ma)$, $\pi_{(R,\varphi)}\in\Pi$ is the Mallows distribution with central ranking $R$ and dispersion parameter $\varphi$. That is, for any $W\in \ml(\ma)$, $\pi_{(R,\varphi)}(W) = \varphi^{\kt(R,W)}/Z_{\varphi}$, where $\kt(R,W)$ is the {\em Kendall Tau distance} between $R$ and $W$, namely the number of pairwise disagreements between $R$ and $W$, and $Z_\varphi$ is the normalization constant.   It follows that $\mm_{[\underline{\varphi},\overline{\varphi}]}$  is  strictly positive, closed, and $\conv(\Pi)$ contains the uniform distribution over all rankings, denoted by $\piuni$.

\begin{dfn}[\bf Smoothed likelihood of ties]\label{dfn:smoothed-ties}
Given a voting rule $r$, a single-agent preference model $\mm=(\Theta,\ml(\ma),\Pi)$,  $2\le k\le m$,  and $n\in\mathbb N$, the   {\em max (respectively, min) smoothed likelihood of ($k$-way) ties} is defined as $\slt{\Pi}{\cor}{m}{k}{n}$ in~(\ref{equ:max-sTies}) (\em respectively, $\ilt{\Pi}{\cor}{m}{k}{n}$ in~(\ref{equ:min-sTies})).
\end{dfn}

\section{PMV-in-Polyhedron Problem and Main Technical Theorems}
\label{sec:maintechthm}
We first formally define PMV and the PMV-in-polyhedron problem studied in this paper.
\begin{dfn}[\bf Poisson multivariate variables (PMVs)]
Given any $q,n\in \mathbb N$ and any vector $\vec \pi$ of $n $ distributions over $[q]$,  we let $\vXp$ denote the {\em $(n,q)$-PMV} that corresponds to $\vec \pi$. That is, let $Y_1,\ldots,Y_n$ denote $n$ independent random variables over $[q]$ such that for any $j\le n$, $Y_j$ is distributed as $\pi_j$. For any $1\le i\le q$, the $i$-th component of $\vXp$ is the number of $Y_j$'s that take value $i$. 
\end{dfn}

\begin{dfn}[\bf The PMV-in-polyhedron problem]
\label{dfn:PMV-in-Poly}
Given $q\in\mathbb N$, a polyhedron $\poly\subseteq \mathbb R^q$, and a set $\Pi$ of distributions over $[q]$, we are interested in  
$$\text{\bf the upper bound }\sup\nolimits_{\vec\pi\in\Pi^n}\Pr(\vXp \in \poly)\text{, and {\bf the  lower bound}}\inf\nolimits_{\vec\pi\in\Pi^n}\Pr(\vXp \in \poly)$$
\end{dfn}
In words, the former (respectively, latter) is the maximum (respectively, minimum) probability for the $(n,q)$-PMV  to be in $\poly$, where the distribution for each of the $n$ variables  is chosen from $\Pi$. 

To present the  theorem, we   introduce some notation   followed by an example. Given $q\in\mathbb N, L\in\mathbb N$, an $L\times q$ integer matrix $\ba$, a  $q$-dimensional row vector $\vec b$, and an  $n\in\mathbb N$,  we define $\poly,\polyz$, $\polyn$, and $\polynint$  as follows.
$$\begin{array}{ll}
\poly \triangleq \left\{\vec x\in {\mathbb R}^q: \ba\cdot \invert{\vec x}\le \invert{\vec b}\right\}, &\polyz \triangleq \left\{\vec x\in {\mathbb R}^q: \ba\cdot \invert{\vec x}\le \invert{\vec 0}\right\},\\
\polyn \triangleq\left \{\vec x\in\poly\cap {\mathbb R}_{\ge 0}^q:\vec x\cdot \vec 1=n\right\}, &\polynint \triangleq \polyn  \cap {\mathbb Z}_{\ge 0}^q.
\end{array}$$
That is, $\poly$ is the polyhedron represented by $\ba$ and ${\vec b}$; $\polyz$ is the {\em characteristic cone} of $\poly$, $\polyn$ consists of non-negative vectors in $\poly$ whose  $L_1$ norm is $n$, and $\polynint$ consists of non-negative integer vectors in $\polyn$. By definition, $\polynint\subseteq \polyn\subseteq \poly$. Let $\dim(\polyz)$ denote the {\em dimension} of $\polyz$, i.e., the dimension of the minimal linear subspace of $\mathbb R^q$ that contains $\polyz$. 

Throughout the paper, we assume that the set of distribution $\Pi$ is {\em strictly positive} and {\em closed}, which are mild assumptions as discussed in~\cite{Xia2020:The-Smoothed}, formally defined as follows.
\begin{dfn}[\bf \boldmath Strictly positive and closed $\Pi$]
\label{dfn:strict-closed-Pi}
Given any $q\in\mathbb N$. A probability distribution $\pi$ over $[q]$ is {\em strictly positive (by $\epsilon$)} for some $\epsilon>0$, if for all $i\in[q]$, $\pi(i)\ge \epsilon$. We say that a set $\Pi$ of distributions over $[q]$ is {\em strictly positive (by $\epsilon$)}  for some $\epsilon>0$, if every $\pi\in\Pi$ is  strictly positive by $\epsilon$.   $\Pi$ is {\em closed} if it is a closed set in ${\mathbb R}^q$.  
\end{dfn}
 Let $\conv(\Pi)$ denote the convex hull of $\Pi$ and let $\cone(\Pi)$ denote the convex cone generated by $\Pi$. 

\begin{ex}
\label{ex:maintechdef}  Figure~\ref{fig:mainthmex} illustrates two examples with $q = 2$ and $\Pi=\{\pi_1,\pi_2\}$, where $\pi_1=(\frac13, \frac 23)$ and $\pi_2=(\frac12, \frac 12)$.   
In both examples, $\conv(\Pi)$ is the line segment between $\pi_1$ and $\pi_2$, $\cone(\Pi)$  is the shaded area, $\poly$ is the red area, $\polyz$ is the blue area, the intersection of $\poly$ and $\polyz$ is the purple area, and $\polyn$ is the green line segment. 
A key difference between Figure~\ref{fig:mainthmex} (a) and (b) is whether $\conv(\Pi)\cap \polyz=\emptyset$ (which is true in Figure~\ref{fig:mainthmex} (a) but not in (b)). Also it is possible that $\polyz\nsubseteq  \poly$, as can be seen in Figure~\ref{fig:mainthmex} (b).
\end{ex}
\begin{figure}[htp]
\centering
\begin{tabular}{cc}
\includegraphics[width = 0.49\textwidth]{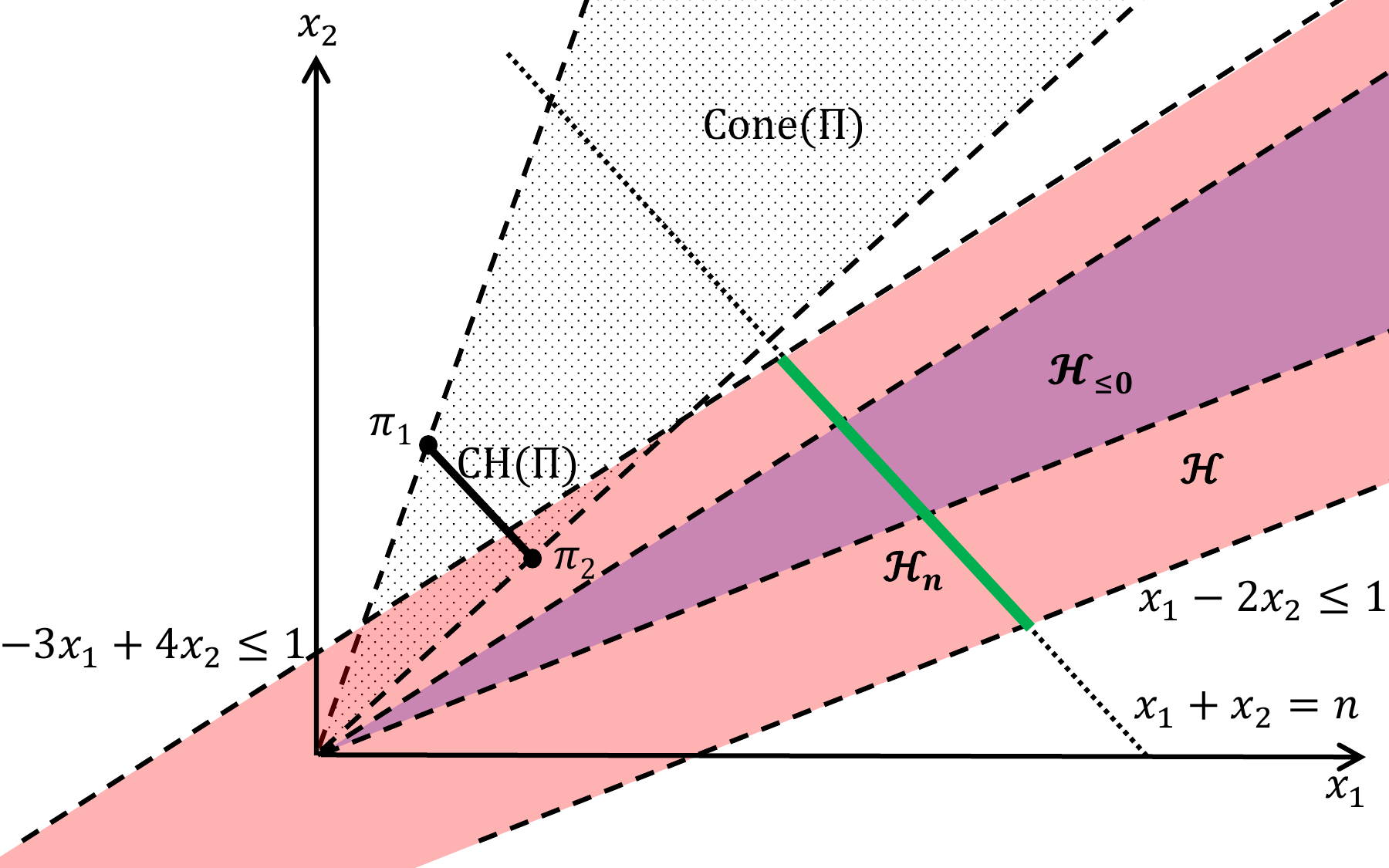} 
&\includegraphics[width = 0.49\textwidth]{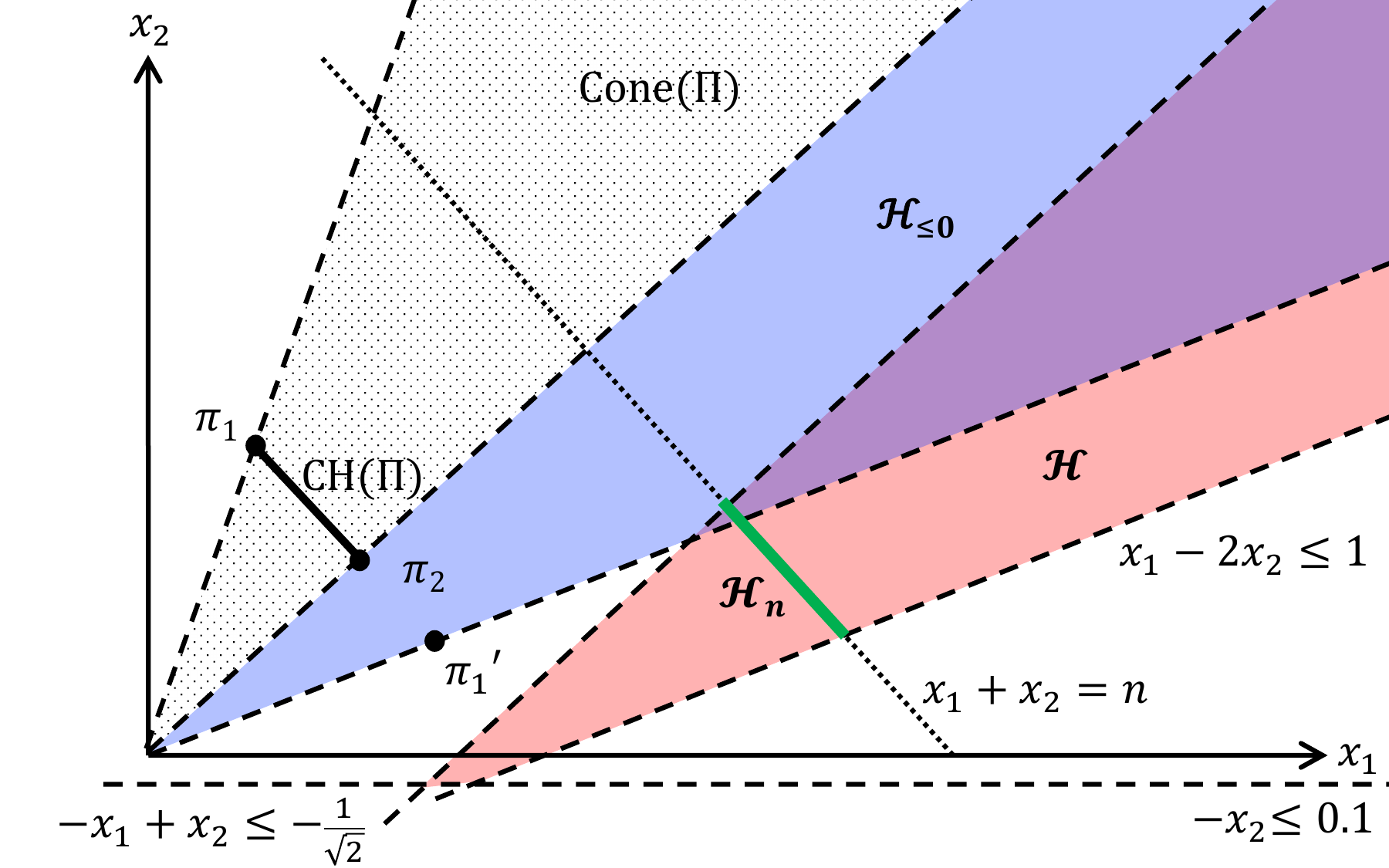} \\
\small (a) $\ba = \left[\begin{array}{rr}-3&4\\1&-2\end{array}\right]$ and $\vec b = \left[\begin{array}{r}1\\1\end{array}\right]$. 
&\small (b)  $\ba = \left[\begin{array}{rr}-1&1\\1&-2\\ 0& -1\end{array}\right]$ and $\vec b = \left[\begin{array}{r}-\frac 1{\sqrt 2}\\1\\0.1\end{array}\right]$.
\end{tabular}
\caption{\small Two examples of $\poly$, $\polyz$, $\polyn$, $\conv(\Pi)$, and $\cone(\Pi)$. \label{fig:mainthmex}}
\end{figure}
\paragraph{\bf A high-level attempt at the PMV-in-Polyhedron problem.} Before formally presenting the theorem, let us take a high-level attempt to develop intuition. Take the upper bound $\sup\nolimits_{\vec\pi\in\Pi^n}\Pr(\vXp \in \poly)$ for example, there are three cases. 

$\bullet$ {\bf\boldmath The $0$ case.} Clearly, if $\poly$ does not contain any non-negative integer whose $L_1$ norm is $n$, which is equivalent to $\polynint=\emptyset$, then   $\sup\nolimits_{\vec\pi\in\Pi^n}\Pr(\vXp \in \poly)=0$. 

$\bullet$  {\bf\boldmath The exponential case (Figure~\ref{fig:mainthmex} (a)).} Suppose $\polynint\ne\emptyset$. For any $\vec\pi=(\pi_1,\ldots,\pi_n)\in\Pi^n$ chosen by the max-adversary, we have $\vec\pi\cdot\vec 1= \sum_{j=1}^n\pi_j\in \cone(\Pi)$. According to various multivariate central limit theorems, $\vXp$ is ``centered'' around an $\Theta(\sqrt n)$ neighborhood of $\vec\pi\cdot\vec 1$ with high probability. Therefore, if $\vec\pi\cdot\vec 1$ is $\Theta(n)$ away from $\poly$, then $\sup\nolimits_{\vec\pi\in\Pi^n}\Pr(\vXp \in \poly)$ is exponentially small. This happens when $\conv(\Pi)\cap \polyz=\emptyset$  as shown in Figure~\ref{fig:mainthmex} (a).

$\bullet$  {\bf\boldmath The polynomial case (Figure~\ref{fig:mainthmex} (b)).} Otherwise we have $\conv(\Pi)\cap \polyz\ne \emptyset$ as shown in Figure~\ref{fig:mainthmex} (b). In this case, the max-adversary can choose $\vec \pi\in\Pi^n$ such that $\vec\pi\cdot\vec 1$ is either in $\cone(\Pi)$  or  close to it, which means that $\sup\nolimits_{\vec\pi\in\Pi^n}\Pr(\vXp \in \poly)$ should be larger than that in the exponential case.  However, it is not immediately clear that the probability is polynomial, because $\vec\pi\cdot\vec 1$ being close to $\cone(\Pi)$ and $\polynint\ne\emptyset$ does not immediately imply that $\vXp$ is close to $\polynint$. Even if we assume that $\vXp$ is close to some (integer) vectors in $\polynint$, it is unclear how ``dense'' such vectors  are in the $\Theta(\sqrt n)$ neighborhood of $\vec\pi\cdot\vec 1$, which $\vXp$ falls into with high probability.  In fact, accurately bounding the probability in the polynomial case  is the most challenging part of the problem, because existing asymptotic tools  fail to work due to their $O(n^{-0.5})$ error terms. 

The main technical theorem below confirms  the intuition developed above  when $\Pi$ is  closed and strictly positive (see Definition~\ref{dfn:strict-closed-Pi}), and the answer to the polynomial case is $\Theta\left(n^{\frac{\dim(\polyz)-q}{2}}\right)$, which is often much smaller than $n^{-0.5}$.
\begin{thm}[\bf Smoothed Likelihood of PMV-in-polyhedron]\label{thm:maintechh} Given any $q\in\mathbb N$, any closed and strictly positive $\Pi$ over $[q]$, and any polyhedron $\poly$ characterized by an integer matrix $\ba$, for any $n\in\mathbb N$, 
\begin{align*}
&\sup_{\vec\pi\in\Pi^n}\Pr\left(\vXp \in \poly\right)=\left\{\begin{array}{ll}0 &\text{if } \polynint=\emptyset\\
\exp(-\Theta(n)) &\text{if } \polynint \ne \emptyset \text{ and }\polyz\cap\conv(\Pi)=\emptyset\\
\Theta\left(n^{\frac{\dim(\polyz)-q}{2}}\right) &\text{otherwise (i.e. } \polynint\ne \emptyset \text{ and }\polyz\cap\conv(\Pi)\ne \emptyset\text{)}
\end{array}\right.,\\
&\inf_{\vec\pi\in\Pi^n}\Pr\left(\vXp \in \poly\right)=\left\{\begin{array}{ll}0 &\text{if } \polynint=\emptyset\\
\exp(-\Theta(n)) &\text{if } \polynint \ne \emptyset \text{ and }\\\
\Theta\left(n^{\frac{\dim(\polyz)-q}{2}}\right) &\text{otherwise (i.e. } \polynint\ne \emptyset \text{ and }\conv(\Pi)\subseteq\polyz
\text{)}\end{array}\right.
\end{align*}
\end{thm}

\paragraph{\bf Remarks on the  power of Theorem~\ref{thm:maintechh}.}  We believe that the main power of  Theorem~\ref{thm:union-poly} is that it provides a systematic way of reducing  probabilistic analysis  (asymptotically tight upper and lower bounds for the PMV-in-Polyhedron problem)  to  worst-case  non-probabilistic analysis, which are often easy to verify. In particular, when $\conv(\Pi)$ can be represented by the convex hull of a finite number of vectors, whether $\polyz\cap\conv(\Pi)=\emptyset$ and/or 
$\conv(\Pi)\not\subseteq \polyz$ can be verified by linear programming. 
Take the $\sup$ part of Theorem~\ref{thm:maintechh} in the setting of Example~\ref{ex:maintechdef} for instance.  Suppose $\polynint\ne\emptyset$.
\begin{itemize}
\item In Figure~\ref{fig:mainthmex} (a), it is easy to see that $\polyz\cap \conv(\Pi)=\emptyset$. Therefore,  
$$\forall \vec \pi \in \Pi^n, \Pr\left(\vXp \in \poly\right) \le  \exp(-\Theta(n))$$ 
\item In Figure~\ref{fig:mainthmex} (b), we have $\polyz\cap \conv(\Pi)=\{\pi_2\}\ne\emptyset$, $\conv(\Pi)\not\subseteq \polyz$, and  $\dim(\polyz) = 2$. Therefore, for any sufficiently large $n$ (for which it is not hard to prove that $\polynint\ne\emptyset$), we have
$$\forall \vec \pi \in \Pi^n,  \exp(-\Theta(n)) \le\Pr\left(\vXp \in \poly\right) \le\Theta(n^{-\frac{2-2}{2}}) = \Theta(1)$$
Notice that both bounds are asymptotically tight. For example, the lower bound can be achieved by  $\vec \pi = \{\pi_1\}^n$ and the upper bound can be achieved by   $\vec \pi = \{\pi_2\}^n$.
\end{itemize}
As an example of the inf part of Theorem~\ref{thm:maintechh}, suppose $\pi_1$ is replaced by $\pi_1'= (\frac 23,\frac 13)$ in Figure~\ref{fig:mainthmex} (b). Then, $\conv(\Pi)\subseteq \polyz$, which means that  $\forall \vec \pi \in \Pi^n,\Pr\left(\vXp \in \poly\right) \ge \Theta(1)$ and the lower bound is asymptotically tight.

\paragraph{\bf Remarks on the generality and limitations of Theorem~\ref{thm:maintechh}.} We believe that Theorem~\ref{thm:maintechh} is quite general, because first, it provides a dichotomy (more precisely, trichotomy) for the PMV-in-Polyhedron problem. Second, the upper and lower bounds are asymptotically tight. And third, the theorem works for arbitrary $\poly$ characterized by an integer matrix $\ma$ and arbitrary $\vec b$, and any closed and strictly positive   $\Pi$. As a notable special case, when $\Pi$ contains a single distribution $\pi$, the sup and inf parts of the theorem coincide, and the theorem characterizes the PMV-in-Polyhedron problem for i.i.d.~PMVs.

The main limitations are, first,  the constants in the asymptotic bounds depend on $q$, $\Pi$, and $\poly$, which are assumed to be fixed; and second, $\Pi$ must be strictly positive. Nevertheless, we believe that the two limitations are mild at least in the social choice context, because as can be seen in the next section as well as in~\cite{Xia2020:The-Smoothed}, applications of Theorem~\ref{thm:maintechh} (or more precisely, its extension to unions of multiple polyhedra in Theorem~\ref{thm:union-poly} in Section~\ref{sec:maintech-ext})  answer open questions in social choice under a more general and realistic model than IC. Moreover, as commented in~\cite{Xia2020:The-Smoothed}, many classical models, such as Mallows model and random utility models, are strictly positive. 

\paragraph{\bf Remarks on the comparison with~\citep[Lemma~1]{Xia2020:The-Smoothed}.} We first recall an equivalent and simplified version of~\citep[Lemma~1]{Xia2020:The-Smoothed} as Lemma$\ast$  below for easy reference.

\paragraph{\bf \boldmath Lemma$\ast$ (\cite[Lemma~1]{Xia2020:The-Smoothed}). }{\em Let $\poly = \{\vec x\in {\mathbb R}^q: {\mathbf E}\cdot \invert{\vec x} = \invert{\vec 0}\text{ and } {\mathbf S}\cdot \invert{\vec x} < \invert{\vec 0}\}$, where $\mathbf E$ and $\mathbf S$ are integer matrices and ${\mathbf E}\cdot \invert{\vec1} = \invert{\vec 0}$ and ${\mathbf S}\cdot \invert{\vec1} = \invert{\vec 0}$. Then, $\sup_{\vec\pi\in\Pi^n} \Pr\left(\vXp \in \poly\right)$ is $0$, $\exp(-\Omega(n))$, or $O(n^{- \rank(\mathbf E)/2})$, and the poly bound is asymptotically tight for infinitely many $n\in \mathbb N$.
}

\vspace{2mm} We believe that our Theorem~\ref{thm:maintechh} is a non-trivial and significant improvement of Lemma~$\ast$ in the following three aspects.

First, Theorem~\ref{thm:maintechh} works for any polyhedron $\poly = \left\{\vec x: \ba\cdot \invert{\vec x}\le \invert{\vec b}\right\}$ with arbitrary integer matrix $\ba$, while Lemma$\ast$ requires $\ba\cdot \invert{\vec 1}=\invert{\vec 0}$ and also essentially requires that elements in $\vec b$ to be either $0$ or $-1$, which correspond  to the $\mathbf E$ part and the $\mathbf S$ part in Lemma$\ast$, respectively. 

 Second,  Theorem~\ref{thm:maintechh} provides asymptotically tight bounds, while Lemma$\ast$ only claims that the bounds are asymptotically tight for infinitely many $n$'s. 

 Third,  Theorem~\ref{thm:maintechh} characterizes smoothed likelihood for the min-adversary, while  Lemma$\ast$ only works for the max-adversary. While the proof of the min-adversary part  of Theorem~\ref{thm:maintechh} is similar to its max-adversary part, it is due to the improved techniques and lemmas (Lemma~\ref{lem:point-wise-concentration-PMV} and~\ref{lem:lower-any-pi} in the appendix). Without them we do not see an easy way to generalize Lemma$\ast$ to the min-adversary.

The  proof can be found in Appendix~\ref{app:maintech}, where a proof sketch is presented in Appendix~\ref{app:maintech-proof-sketch} and the full proof is presented in Appendix~\ref{app:maintech-full-proof}. 

\subsection{An Extension of Theorem~\ref{thm:maintechh} to Unions of Polyhedra}
\label{sec:maintech-ext}
In this subsection, we present an extension of Theorem~\ref{thm:maintechh} to the union of $I\in\mathbb N$ polyhedra, denoted by  $\upoly = \bigcup_{i\le I}\cpoly{i}$, where $\cpoly{i} = \{\vec x\in {\mathbb R}^q: \ba_i \cdot\invert{\vec x}\le \invert{\vec b_i}\}$ and $\ba_{i}$ is an integer matrix of $q$ columns. We define the {\em PMV-in-$\upoly$} problem similarly as the PMV-in-Polyhedron problem (Definition~\ref{dfn:PMV-in-Poly}), except that $\poly$ is replaced by $\upoly$. 
\begin{dfn}[\bf\boldmath The PMV-in-$\upoly$ problem]
\label{dfn:PMV-in-C}
Given $q,I\in\mathbb N$,  $\upoly = \bigcup_{i\le I}\cpoly{i}$, where $\forall i\le I$,  $\cpoly{i}\subseteq \mathbb R^q$ is a polyhedron, and a set $\Pi$ of distributions over $[q]$, we are interested in 
$$\text{\bf the upper bound }\sup\nolimits_{\vec\pi\in\Pi^n}\Pr(\vXp \in \upoly)\text{, and {\bf the  lower bound}}\inf\nolimits_{\vec\pi\in\Pi^n}\Pr(\vXp \in \upoly)$$ 
\end{dfn}
The key observation is the following straightforward inequality for every  PMV $\vXp$:
\begin{equation}
\label{eq:upoly-high-level}
\max\nolimits_{i\le I}\Pr\left(\vXp\in\cpoly{i}\right)\le \Pr\left(\vXp\in\upoly\right)\le \sum\nolimits_{i\le I}\Pr\left(\vXp\in\cpoly{i}\right)
\end{equation}

See Figure~\ref{fig:upoly-high-level} for an illustration of $I=3$. Notice that the right hand side of (\ref{eq:upoly-high-level}) is no more than $I\cdot \max_{i\le I}\Pr (\vXp\in\cpoly{i} )$, which is $ \Theta(\max\nolimits_{i\le I}\Pr (\vXp\in\cpoly{i} ))$ because $I$ is a constant.

\begin{figure}[htp]
\centering
  \includegraphics[width = .8\linewidth]{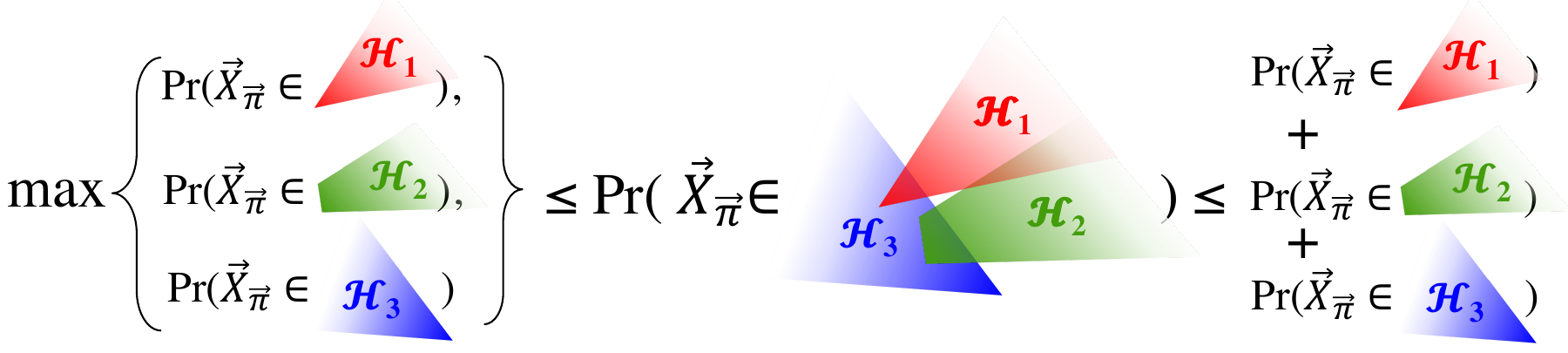}
\caption{\small Illustration of inequality (\ref{eq:upoly-high-level}), where $\upoly = \cpoly{1}\cup \cpoly{2}\cup  \cpoly{3}$. \label{fig:upoly-high-level}
}
\end{figure}

\ \\
\noindent{\bf The high-level idea behind the extension}  is based on a  weighted complete bipartite {\em activation graph} defined as follows, which represents the relationship between  $\conv(\Pi)$ and polyhedra in $\upoly$ in light of Theorem~\ref{thm:maintechh}. Let $\cpolyz{i}$ denote the characteristic cone of $\cpoly{i}$.

\begin{dfn}[\bf \boldmath Activation graph $\calG_{\Pi,\upoly,n}$] For any set of distributions $\Pi$ over $[q]$, any $\upoly = \bigcup_{i=1}^I \cpoly{i}$, and any $n\in \mathbb N$, $\cpoly{i}$ is said to be {\em active (at $n$)} if $\cpolynint{i}\ne\emptyset$; otherwise  $\cpoly{i}$ is said to be {\em inactive (at $n$)}. Moreover, we define the {\em activation graph} $\calG_{\Pi,\upoly,n}$ as follows.

$\bullet$ {\bf Vertices.} The vertices are $\conv(\Pi)$ and  $\{\cpoly{i}: 1\le i\le I\}$.  

$\bullet$  {\bf Edges and weights.} There is an edge between each $\pi\in \conv(\Pi)$ and each $\cpoly{i}$, whose weight is 
$$w_n(\pi,\cpoly{i})\triangleq\left\{\begin{array}{ll}-\infty&\text{if } \cpoly{i}\text{ is inactive at }n\\
-\frac{n}{\log n} &\text{otherwise, if }  \pi\notin \cpolyz{i} \\
\dim(\cpolyz{i})& \text{otherwise}
 \end{array}\right.$$
\end{dfn}

For example, in Figure~\ref{fig:upoly-graph},  $\cpoly{1}$ is inactive at $n$ and both  $\cpoly{2}$ and $\cpoly{3}$ are active,  $\pi\in \cpolyz{1}\cap \cpolyz{2}$ and $\pi\notin \cpolyz{3}$. Notice that the weight on $(\pi,\cpoly{2})$ is $\dim(\cpolyz{2})$ instead of $\dim(\cpoly{2})$.

\begin{figure}[htp]
\centering
  \includegraphics[width = 0.7\linewidth]{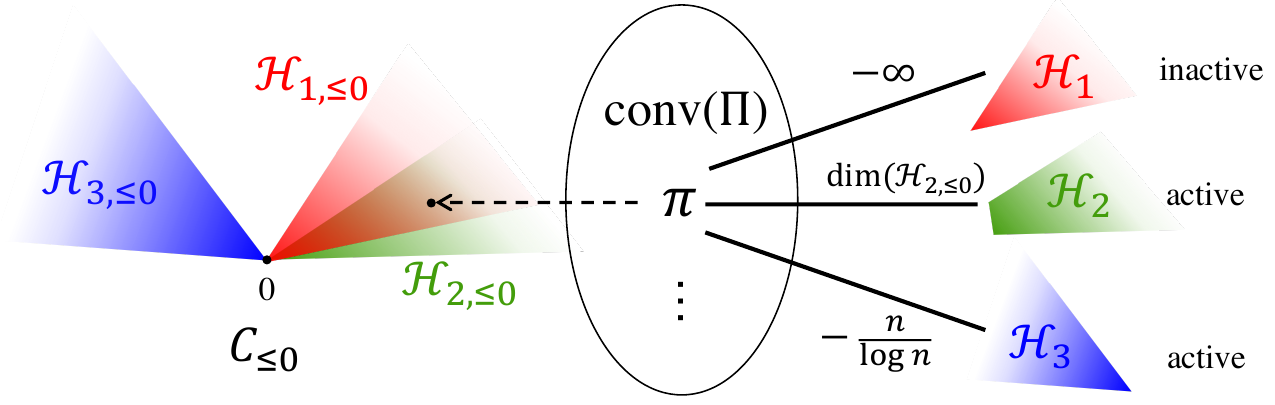}
\caption{\small Illustration of an activation graph $\calG_{\Pi,\upoly,n}$. \label{fig:upoly-graph}
}
\end{figure}

Intuitively, the $0$, exponential, and polynomial cases of Theorem~\ref{thm:maintechh} (applied to $\Pi = \{\pi\}$) corresponds to  the $-\infty$ edge, the $-\frac{n}{\log n}$ edge, and the $\dim(\cpolyz{i})$ edge, respectively. That is, for any $\vec\pi\in\Pi^n$ with $\sum_{j=1}^n \pi_j = n\cdot \pi$,  $\Pr(\vXp\in \cpoly{i})$ is roughly $n$ raise to the power of the weight between $\pi$ and $\cpoly{i}$ in the activation graph, i.e., $n^{w_n(\pi,\cpoly{i})}$. In particular, $n^{-\infty} =0$ and $n^{-\frac{n}{\log n}-q} = \exp(-\Theta(n))$. \footnote{This is the reason behind using $-\frac{n}{\log n}$.  Theorem~\ref{thm:union-poly} still  holds if  $-\frac{n}{\log n}$ is replace by any finite negative number.} 

Therefore, according to (\ref{eq:upoly-high-level}), $\Pr(\vXp\in\upoly)$ is primarily determined by the largest ${w_n (\pi, \cpoly{i})}$, i.e., the maximum weight of all edges connected to $\pi$ in the activation graph.  This is formally defined as follows.

\begin{dfn}[\bf Active dimension] Given $\upoly$, $n$, and $\pi\in \mathbb R^{q}$, we define {\em maximum active dimension of  $\upoly$ at $\pi$ and $n$} ({\em active dimension at $\pi$} for short, when $\upoly$ and $n$ are clear from the context), denoted by $\md{\upoly}{\pi}$, as follows. 
$$\md{\upoly}{\pi}\triangleq\max\nolimits_{i\le I} w_n (\pi, \cpoly{i})$$
\end{dfn}
Consequently, a max- (respectively, min-) adversary aims to choose $\vec\pi=(\pi_1,\ldots,\pi_n)\in \Pi^n$ to maximize (respectively, minimize) $\md{\upoly}{\frac 1n\sum_{j=1}^n \pi_j}$, which are characterized by  $\alpha_n$ (respectively, $\beta_n$) defined as follows.
\begin{align*}
&\alpha_{n} \triangleq \max\nolimits_{\pi\in\conv(\Pi)} \md{\upoly}{\pi}\\
&\beta_{n} \triangleq \min\nolimits_{\pi\in\conv(\Pi)} \md{\upoly}{\pi}
\end{align*}
We note that $\alpha_{n}$ and $\beta_n$ depend on $\Pi$ and $\upoly$, which are often clear from the context. Also, by definition, $\alpha_n =-\infty$ is equivalent to $\beta  = -\infty$, which is equivalent to $\upolynint = \emptyset$.  We are now ready to use $\alpha_n$ and $\beta_n$ to present the extension of Theorem~\ref{thm:maintechh} to the PMV-in-$\upoly$ problem.

\begin{thm}[\bf \boldmath Smoothed Likelihood of PMV-in-$\upoly$]
\label{thm:union-poly}
Given any $q, I\in\mathbb N$, any closed and strictly positive $\Pi$ over $[q]$, and any $\upoly = \bigcup_{i\in I}\cpoly{i}$  characterized by integer matrices, for any $n\in\mathbb N$, 
\begin{align*}
&\sup_{\vec\pi\in\Pi^n}\Pr\left(\vXp \in \upoly\right)=\left\{\begin{array}{ll}0 &\text{if } \alpha_n = -\infty\\
\exp(-\Theta(n)) &\text{if } -\infty<\alpha_n<0\\
\Theta\left(n^{\frac{\alpha_n-q}{2}}\right) &\text{otherwise (i.e. } \alpha_n>0\text{)}
\end{array}\right.,\\
&\inf_{\vec\pi\in\Pi^n}\Pr\left(\vXp \in \upoly\right)=\left\{\begin{array}{ll}0 &\text{if } \beta_n = -\infty\\
\exp(-\Theta(n)) &\text{if } -\infty<\beta_n<0\\
\Theta\left(n^{\frac{\beta_n-q}{2}}\right) &\text{otherwise (i.e. } \beta_n>0\text{)}\end{array}\right..
\end{align*}
\end{thm}
Roughly speaking, the max- (respectively, min-) smoothed likelihood for an $(n,q)$-PMV to be in $\upoly$ is approximately $n^{\frac{\alpha_n-q}{2}}$ (respectively, $n^{\frac{\beta_n-q}{2}}$).  The proof is done by combining the applications of Theorem~\ref{thm:maintechh} to $\Pi$ and every $\cpoly{i}$, and can be found in Appendix~\ref{sec:maintech-union-proof}.

\paragraph{\bf Remarks on the applications of Theorem~\ref{thm:union-poly}.} We believe that Theorem~\ref{thm:union-poly} is a useful and  general tool to study the smoothed likelihood of many events and properties in social choice, as shown in~\cite{Xia2020:The-Smoothed} as well as in the rest of this paper. Like Theorem~\ref{thm:maintechh}, the power of Theorem~\ref{thm:union-poly} is that it provides a systematic way of reducing probabilistic analysis   to worst-case and non-probabilistic analysis, i.e., the characterizations of $\alpha_n$, and $\beta_n$. 
Nevertheless, characterizing $\alpha_n$ and $\beta_n$ can still be challenging, which is equivalent to characterizing active $\cpoly{i}$, $\cpolyz{i}$, and $\dim(\cpolyz{i})$, as we will see in the next section.

\section{Smoothed Likelihood of Ties} 
\label{sec:smoothedties}
In this section, we apply Theorem~\ref{thm:union-poly} to provide dichotomous characterizations of the smoothed likelihood of ties (Definition~\ref{dfn:smoothed-ties}) under some commonly studied voting rules. 



\subsection{Integer Positional Scoring Rules}
\label{sec:posrules}
We first apply   Theorem~\ref{thm:union-poly} to polyhedra that are similar to those in Example~\ref{ex:intro} and obtain the following theorem for integer positional scoring rules.  
\begin{thm}[\bf Smoothed likelihood of ties: integer positional scoring rules]\label{thm:score} For any fixed  $m\ge 3$, let $\mm= (\Theta,\ml(\ma),\Pi)$ be a strictly positive and closed single-agent preference model and let $\vec s$ be an integer scoring vector. For any $2\le k\le m$ and any $n\in\mathbb N$,  
\begin{align*}
\slt{\Pi}{r_{\vec s}}{m}{k}{n}&=\left\{\begin{array}{ll}0 &\text{if } \forall P\in \ml(\ma)^n, |\cor_{\vec s}(P)|\ne k\\
\exp(-\Theta(n)) &\text{otherwise, if } \forall\pi\in \conv(\Pi), |\cor_{\vec s}(\pi)|< k\\
\Theta(n^{-\frac{k-1}{2}}) &\text{otherwise}
\end{array}\right.,\\
\ilt{\Pi}{r_{\vec s}}{m}{k}{n}&=\left\{\begin{array}{ll}0 &\text{if } \forall P\in \ml(\ma)^n, |\cor_{\vec s}(P)|\ne k\\
\exp(-\Theta(n)) &\text{otherwise, if }\exists\pi\in \conv(\Pi)\text{ s.t. } |\cor_{\vec s}(\pi)|< k\\
\Theta(n^{-\frac{k-1}{2}}) &\text{otherwise}
\end{array}\right..
\end{align*}
\end{thm} 
Take the max  smoothed likelihood of ties in Theorem~\ref{thm:union-poly} for example. Like Theorem~\ref{thm:maintechh}, the condition for the $0$ case is trivial. Assuming that the $0$ case does not happen, the exponential case happens if  
no distribution (viewed as a fraction profile) in the convex hull of $\Pi$ has {\em at least} $k$ winners under $\cor_{\vec s}$. Otherwise, the polynomial case happens. That is, there exists an $n$-profile $P$ with exactly $k$ winners under $\cor_{\vec s}$, and there exists $\pi\in\conv(\Pi)$ that has at least $k$ winners. Notice that the existence of such $P$  (which depends on $n$ but not $\Pi$) does not imply the existence of such $\pi$ (which depends on $\Pi$ but not $n$), nor vice versa. All proofs in this section are delegated to Appendix~\ref{app:smoothedties}.


We immediately have the follow corollary of Theorem~\ref{thm:score} when the uniform distribution $\piuni$ is in $ \conv(\Pi)$, because $\cor_{\vec s}(\piuni) = \ma$ and $|\ma| = m\ge k$, which means that the exponential case never happends.  Notice that the corollary does not require $\piuni\in \Pi$.
\begin{coro}[\bf Max smoothed likelihood of ties: positional scoring rules]
\label{coro:scoringtie}
 For any fixed  $m\ge 3$, let $\mm= (\Theta,\ml(\ma),\Pi)$ be a strictly positive and closed single-agent preference model  with $\piuni\in \conv(\Pi)$. For any integer scoring vector $\vec s$ and any $k\le m$, for any $n\in\mathbb N$, 

$\hfill\slt{\Pi}{r_{\vec s}}{m}{k}{n}=\left\{\begin{array}{ll}0 &\text{if } \forall P\in \ml(\ma)^n, |\cor_{\vec s}(P)|\ne k\\
\Theta(n^{-\frac{k-1}{2}}) &\text{otherwise}
\end{array}\right..\hfill$
\end{coro}
The $0$ case can indeed happen, for example, when $r$ is the plurality rule, $k=m$, and $m\nmid n$. 
As commented in~\citep{Xia2020:The-Smoothed}, $\piuni\in \conv(\Pi)$ is a natural assumption that holds for many  single-agent preference models. In particular, Corollary~\ref{coro:scoringtie} works for IC, which corresponds to $\Pi = \{\piuni\}$.

\subsection{Edge-Order-Based Rules}
\label{sec:eorules}
The characterization for edge-order-based rules is more complicated due to the hardness in characterizing active $\cpoly{i}$,  $\cpolyz{i}$, and $\dim(\cpolyz{i})$ in the polyhedra representation of $k$-way ties. We first introduce necessary notation to formally define edge-order-based rules, whose winners only depend on the order over all edges in WMG w.r.t.~their weights, called {\em palindromic orders}.
\begin{dfn} [\bf Palindromic orders]
A total preorder $O$ over $\me= \{(a,b)\in \ma\times \ma:a\ne b\}$ is {\em palindromic}, if for any pair of edges $(a,b),(c,d)$ in $\me$, $(a,b)\rhd_O (c,d)$ if and only if $(d,c)\rhd_O (b,a)$, where $e_1\rhd_O e_2$ means that $e_1$ is ranked strictly above $e_2$ in $O$, and $e_1\equiv_O e_2$ means that $e_1$ and $e_2$ are tied in $O$. Let $\pale$ denote the set of all palindromic orders over $\me$. 

For any weighted directed graph $G$ over $\ma$ with weights $\{w(a,b):a\ne b\}$ such that $w(a,b)=-w(b,a)$, let $\eo(G)\in\pale$ denote the palindromic order w.r.t.~the decreasing order of weights in $G$. For any profile $P$, let $\eo(P)=\eo(\wmg(P))$. 
\end{dfn}

In this paper we often use the  {\em tier representation} of palindromic orders, which partition edges  into equivalent classes (tiers). 

\begin{dfn}[\bf Tier representation and refinement of palindromic orders]
Any palindromic order $O\in \pale$ can be partitioned into {\em tiers}:  
$$O=T_1\rhd \cdots\rhd T_t \rhd T_0\rhd T_{t+1}\rhd\cdots \rhd T_{2t},$$
where for each $1\le i\le t$, edges in $T_i$ are tied, edges in $T_{2t+1-i}$ are tied, and edges in $T_{2t+1-i}$ are obtained by flipping edges in $T_i$. $T_0$ is called the {\em middle tier}, which consists of all edges $e$ with $e\equiv_O\bar e$, where $\bar e$ represents flipped $e$. 
 Only $T_0$ is allowed to be empty.  Let $\ties(O)=\sum_{i=1}^t(|T_i|-1)+|T_0|/2$.  
 
 $O_1\in \pale$  {\em refines}  $O_2\in \pale$, if for all pair of elements $(e_1,e_2)$, $e_1\rhd_{O_2} e_2$ implies $e_1\rhd_{O_1} e_2$. 
 \end{dfn}


 
\begin{ex} Figure~\ref{fig:ex-palindromic} illustrates an example of a profile $P$, its $\wmg$, and its corresponding palindromic order.  In Figure~\ref{fig:ex-palindromic} (b) only edges with positive weights are shown.
\begin{figure}[htp]
\centering
\begin{tabular}{@{}c c c@{}}
\begin{minipage}{0.23\linewidth}
\centering
$ \{1\succ 2\succ 3, 1\succ 3\succ 2\}$
\end{minipage} & \begin{minipage}{0.2\linewidth}
\centering
\includegraphics[width = \textwidth]{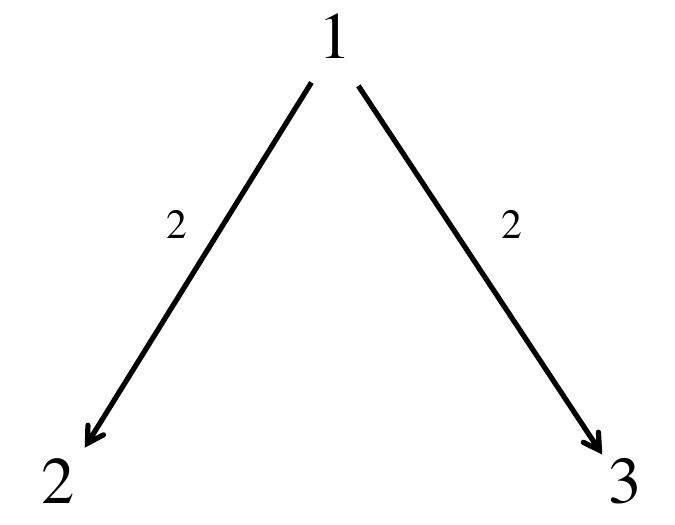}\end{minipage} & 
\begin{minipage}{0.5\linewidth}
\centering
$\underbrace{\{(1,2),(1,3)\}}_{T_1}\rhd \underbrace{\{(2,3),(3,2)\}}_{T_0}\rhd \underbrace{\{(2,1),(3,1)\}}_{T_2}$
\end{minipage}\\
\small (a) Profile $P$. &\small  (b) $\wmg(P)$. &\small  (c) $\eo(P)$.
\end{tabular}
\caption{\small An example of a profile, its $\wmg$, and its palindromic order. \label{fig:ex-palindromic}}
\end{figure}

Let $O = \underbrace{\{(1,2)\}}_{T_1}\rhd \underbrace{\{(1,3)\}}_{T_2}\rhd  \underbrace{\{(2,3),(3,2)\}}_{T_0}\rhd \underbrace{\{(3,1)\}}_{T_3}\rhd \underbrace{\{(2,1)\}}_{T_4}$. We have $\ties(\eo(P)) =3$, $\ties(O)=1$, and $O$ refines $\eo(P)$.
\end{ex}

We are now ready to formally define edge-order-based rules using palindromic orders.

\begin{dfn}[\bf Edge-order-based rules]\label{dfn:eorule}
A voting rule $\cor$ is said to be {\em edge-order-based}, if for every pair of profiles $P_1,P_2$ with $\eo(P_1)=\eo(P_2)$, we have $r(P_1)=r(P_2)$.
\end{dfn}

Many WMG-based rules, such as Copeland, Maximin, Schulze, and ranked pairs, are  edge-order-based. The domain of any edge-order-based rule $r$ can be naturally extended to palindromic orders. When applying Theorem~\ref{thm:union-poly} to  edge-order-based rules, each polyhedron $\cpoly{i}$ in $\upoly$ is indexed by a palindromic order $O$ with $k$ co-winners, such that $\cpolynint{i}$ corresponds to the histograms of $n$-profiles whose edge orders are $O$. 

We now define and characterize palindromic orders obtained from $n$-profiles.

\begin{dfn}   For any $n\in \mathbb N$, let $\pal{n}\triangleq\{\eo(P):P\in \ml(\ma)^n\}$. 
Let $\palo\subset\pale$ denote the set of palindromic orders $O$ whose middle tier is empty. \end{dfn}
\begin{prop}\label{prop:eo}
 For any $\ma$ and any $n\ge m^4$,
$\pal{n} = \left\{\begin{array}{ll} \pale &\text{if }2\mid n\\
\palo &\text{if }2\nmid n\\
\end{array}\right.$.
\end{prop}
The proof of Proposition~\ref{prop:eo} is delegated to Appendix~\ref{appendix:proof-lemeo}. Next, we define $\mo_{\cor,k,n}^{\pi}$ as the set of palindromic orders $O$ that satisfies three conditions:  (1) $O\in \pal{n}$; (2) there are exactly $k$ winners in $O$ under $\cor$, and (3) $O$ refines $\eo(\pi)$, where $\pi$ is viewed as a fractional profile. Let $\mo_{\cor,k,n}^{\Pi} \triangleq \bigcup_{\pi\in\conv(\Pi)}\mo_{\cor,k,n}^{\pi}$. When $\mo_{\cor,k,n}^{\Pi}\ne \emptyset$, we let $\ell_{\min}$ denote the minimum number of ties in palindromic orders in $\mo_{\cor,k,n}^{\Pi}$.  When $\mo_{\cor,k,n}^{\Pi}\ne \emptyset$, we let  $\ell_{\text{mm}}$ denote the maximin number of ties, where the maximum is taken for all $\pi\in \conv(\Pi)$, and for any given $\pi$, the minimum is taken for all palindromic orders in $\mo_{\cor,k,n}^{\pi}$.  We note that $\ell_{\min}$ and $\ell_{\text{mm}}$ depend on $\Pi$, $\cor$, $k$, and $n$, which are clear from the context. The formal definitions can be found in~Appendix~\ref{app:def-eo}. In fact, $\ell_{\min}$ and $\ell_{\text{mm}}$ correspond to $m!-\alpha_n$ and $m!-\beta_n$ in Theorem~\ref{thm:union-poly}. We are now ready to present  the theorem for EO-based rules.


\begin{thm}[\bf Smoothed likelihood of ties: edge-order-based rules]\label{thm:eorule}   For any fixed  $m\ge 3$, let $\mm= (\Theta,\ml(\ma),\Pi)$ be a strictly positive and closed single-agent preference model and let $\cor$ be an edge-order-based rule. For any $2\le k\le m$ and any $n\in\mathbb N$,  
\begin{align*}
&\slt{\Pi}{r}{m}{k}{n}=\left\{\begin{array}{ll}0 &\text{if } \forall O\in \pal{n}, |\cor(O)|\ne k\\
\exp(-\Theta(n)) &\text{otherwise if } \mo_{\cor,k,n}^{\Pi} = \emptyset\\
\Theta\left(n^{-\frac{\ell_{\min}}{2}}\right) &\text{otherwise}
\end{array}\right.\\
&\ilt{\Pi}{r}{m}{k}{n}=\left\{\begin{array}{ll}0 &\text{if } \forall O\in \pal{n}, |\cor(O)|\ne k\\
\exp(-\Theta(n)) &\text{otherwise if } \exists \pi\in\conv(\Pi) \text{ s.t. }\mo_{\cor,k,n}^{\pi} = \emptyset\\
\Theta\left(n^{-\frac{\ell_{\text{mm}}}{2}}\right) &\text{otherwise}
\end{array}\right.
\end{align*}
\end{thm}
The proof can be found in Appendix~\ref{app:proof-eorule}. In the remainder of this section, we apply Theorem~\ref{thm:eorule} to provide dichotomous characterizations of  $\max$-smooth likelihood of ties under Copeland$_\alpha$, maximin, Schulze, and ranked pairs for the model in Corollary~\ref{coro:scoringtie}, which includes IC as a special case. 


\begin{prop}[\bf \boldmath Max smoothed likelihood of ties: Copeland$_\alpha$]\label{prop:copeland}
 For any fixed  $m\ge 3$, let $\mm= (\Theta,\ml(\ma),\Pi)$ be a strictly positive and closed single-agent preference model with $\piuni\in \conv(\Pi)$. Let $l_\alpha = \min\{t\in{\mathbb N}: t\alpha\in\mathbb Z\}$. For any $2\le k\le m$ and any $n\in\mathbb N$,  
$$\slt{\Pi}{\copeland}{m}{k}{n}=\left\{\begin{array}{ll}
0 &\text{if } 2\nmid n, 2\mid k, \text{and }k\ge m-1\\
\Theta(n^{-\frac{k}{4}}) &\text{if } 2\mid n,  2\mid  k,  { and }\left\{\begin{array}{l}(1) k=m, \text{ or} \\ 
(2) k=m-1 \text{ and } \alpha\ge \frac 12, \text{ or}\\
(3) k=m-1 \text{ and } k\le l_\alpha(l_\alpha+1)
\end{array}\right.\\
\Theta\left(n^{-\frac{l_\alpha(l_\alpha+1)}{4}}\right) &\text{if } 2\mid n,  2\mid  k,  k=m-1, \alpha<\frac12,\text{ and }k> l_\alpha(l_\alpha+1)\\
\Theta(1) &\text{otherwise (i.e., if }2\nmid k \text{ or }k\le m-2\text{)}
\end{array}\right.$$
\end{prop}
The $\Theta(1)$ case appears most typical, which happens when $k$ is odd or $k\le m-2$. The $\Theta (n^{-\frac{l_\alpha(l_\alpha+1)}{4}}\ )$ case appears most interesting, because its degree  depends on the smallest natural number $l_\alpha$ such that $\alpha l_\alpha$ is an integer. For example, $l_0 =1$, $l_{1/3} =3$, $l_{2/5} =5$, and $l_\alpha = \infty$ for any irrational number $\alpha$ (which means that the $\Theta (n^{-\frac{l_\alpha(l_\alpha+1)}{4}})$ case does not happen because $k<\infty = l_\alpha(l_\alpha+1)$). While the $\Theta(1)$ case can probably be proved by standard central limit theorem and the union bound, we are not aware of a previous work on it. Standard techniques are too coarse for other cases.



\begin{prop}[\bf Max smoothed likelihood of ties: maximin]\label{prop:maximin}
 For any fixed  $m\ge 3$, let $\mm= (\Theta,\ml(\ma),\Pi)$ be a strictly positive and closed single-agent preference model  with $\piuni\in \conv(\Pi)$. For any $2\le k\le m$ and any $n\in\mathbb N$,  $\slt{\Pi}{\maximin}{m}{k}{n}=\Theta(n^{-\frac{k-1}{2}})$.
\end{prop}

\begin{prop}[\bf Max smoothed likelihood of ties: Schulze]\label{prop:schulze}
 For any fixed  $m\ge 3$, let $\mm= (\Theta,\ml(\ma),\Pi)$ be a strictly positive and closed single-agent preference model  with $\piuni\in \conv(\Pi)$. For any $2\le k\le m$ and any $n\in\mathbb N$,   $\slt{\Pi}{\schulze}{m}{k}{n}=\Theta(n^{-\frac{k-1}{2}})$.
\end{prop}

\begin{prop}[\bf Max smoothed likelihood of ties: ranked pairs]\label{prop:rp}
 For any fixed  $m\ge 3$, let $\mm= (\Theta,\ml(\ma),\Pi)$ be a strictly positive and closed single-agent preference model  with $\piuni\in \conv(\Pi)$.  For any $2\le k\le m$ and any $n\in\mathbb N$,   
 $\Omega(n^{-\frac{k-1}{2}}) \le \slt{\Pi}{\rp}{m}{k}{n} \le n^{-\Omega(\frac{\log k}{\log\log k})}$. Moreover, when $m\ge  k+5\lceil \log k\rceil$, $\slt{\Pi}{\rp}{m}{k}{n}=\Omega(n^{-\frac{\lceil \log k\rceil}{2}})$. When $k=2$, we have $\slt{\Pi}{\rp}{m}{2}{n}=\Theta(n^{-0.5})$.
\end{prop}

\noindent{\bf Proof sketches of  Propositions~\ref{prop:copeland},~\ref{prop:maximin},~\ref{prop:schulze}, and~\ref{prop:rp}.} The proofs are done by applying  Theorem~\ref{thm:eorule}. For any EO-based rules $\cor$ studied in this paper,  the condition for the $0$ case can be verified efficiently using  Proposition~\ref{prop:eo}  for any sufficiently large $n$. If the $0$ case does not happen, then the exponential case does not happen either, because for any $O\in \pal{n}$ such that $|\cor(O)| =k$, $O$ extends $\eo(\piuni)$, which is the palindromic order that only has the middle tier $T_0$. This also means that for every $\pi\in\conv(\Pi)$, we have $\mo_{\cor,k,n}^{\pi} \subseteq \mo_{\cor,k,n}^{\piuni}$. Consequently,  $\ell_{\text{min}}$ is achieved at $\piuni$. 

The bulk of proof then focuses on characterizing $O\in \pal{n}$ with the minimum number of ties such that $|\cor(O)| = k$. This can be more complicated than it appears,  for example for ranked pairs and for Copeland$_\alpha$ when $2\nmid $ and $\alpha$ is not $0$ or $1$. In particular, for ranked pairs we were only able to obtain (non-tight) upper and lower bounds. 
The full proofs of Propositions~\ref{prop:copeland},~\ref{prop:maximin},~\ref{prop:schulze}, and~\ref{prop:rp} can be found in Appendix~\ref{app:proof-copeland},~\ref{app:proof-mm},~\ref{app:proof-schulze}, and~\ref{app:proof-rp}, respectively. \hfill$\Box$

\subsection{STV and Coombs}
\label{sec:mrerules}

\begin{thm}[\bf Smoothed likelihood of ties: STV and Coombs]\label{thm:STV-Coombs} For any fixed  $m\ge 3$, let $\cor\in\{\stv,\coombs\}$ and let $\mm= (\Theta,\ml(\ma),\Pi)$ be a strictly positive and closed single-agent preference model. For any $2\le k\le m$ and  $n\in\mathbb N$,   $\slt{\Pi}{\cor}{m}{k}{n}$ (respectively, $\ilt{\Pi}{\cor_{\vec s}}{m}{k}{n}$) is either $0$, $\exp(-\Theta(n))$, or $\Theta(\text{poly}(n))$.\end{thm}

The formal statement of the theorem and its proof are delegated to Appendix~\ref{app:mrse}.  To accurately characterize the degree in the polynomial case, we introduce {\em PUT structures} (Definition~\ref{dfn:PUT})  as the counterpart of palindromic orders to define and analyze active polyhedra and the dimensions of their characteristic cones. See Appendix~\ref{app:mrse} for its formal definitions and an example. Like in Section~\ref{sec:eorules}, the theorem can be applied to characterize  max smoothed likelihood of ties for STV and Coombs  for   distributions $\Pi$  where $\piuni\in \conv(\Pi)$, as shown in the following proposition. 

\begin{prop}[\bf Max smoothed likelihood of ties: STV and Coombs]\label{prop:stv-Coombs}
 For any fixed  $m\ge 3$, let $\cor\in\{\stv,\coombs\}$ and  let $\mm= (\Theta,\ml(\ma),\Pi)$ be a strictly positive and closed single-agent preference model  with $\piuni\in \conv(\Pi)$. For any $2\le k\le m$ and any $n\in\mathbb N$,   
$$\slt{\Pi}{\cor}{m}{k}{n}=\left\{\begin{array}{ll}\Theta(n^{-\frac{k-1}{2}})&\text{if }\left\{\begin{array}{l} (1) m\ge 4 \text{, or } \\
(2) m=3\text{ and }k=2, \text{ or }\\
(3) m=k=3\text{ and }(2\mid n \text{ or } 3\mid n)\end{array}\right.\\
0 &\text{otherwise (i.e., }m=k=3\text{,  }2\nmid n,\text{ and }3\nmid n\text{)}\end{array}\right.
$$
\end{prop}
To prove the proposition, we prove a \citet{McGarvey53:Theorem}-type result for STV and Coombs (Lemma~\ref{lem:construction-STV-Coombs}) to characterize active polyhedra. The lemma might be of independent interest.  

\section{Experimental Studies}
We examine the fraction of profiles with two-or-more-way ties using simulated data and Preflib data~\cite{Mattei13:Preflib} under  Borda, plurality,  veto,  maximin, ranked pairs, Schulze,  Copeland$_{0.5}$,   STV, and Coombs.  All experiments were implemented in Python 3 and were conducted on a MacOS laptop with 3.1 GHz Intel Core i7 CPU and 16 GB memory.\\
\begin{minipage}[t][][b]{0.47\textwidth}
{\bf Simulated data.} We  generate profiles of $m=4$ alternatives under IC.   $n$ ranges from  $20$ to $200$. In each setting we generate $100000$ profiles. 

The goal of experiments on synthetic data is to provide a sanity check for the theoretical results in this paper. For clarity, we present the results for Borda, STV, maximin, ranked pairs, and Copeland$_{0.5}$   in Figure~\ref{fig:exp-short}. Results for all voting rules described above can be found in Figure~\ref{fig:exp-full} in Appendix~\ref{sec:exp-full}.  Figure~\ref{fig:exp-short} confirms Corollary~\ref{coro:scoringtie} and Propositions~\ref{prop:copeland}, \ref{prop:maximin},   \ref{prop:rp}, and \ref{prop:stv-Coombs} under IC as discussed in the Introduction: the probability of any-way ties is $\Theta(1)$ for Copeland$_{0.5}$ and is $\Theta(\frac{1}{\sqrt n})$ for  other rules.
\end{minipage}
\hfill
\begin{minipage}[t][][b]{0.5\textwidth}
\centering
  \includegraphics[width = \linewidth]{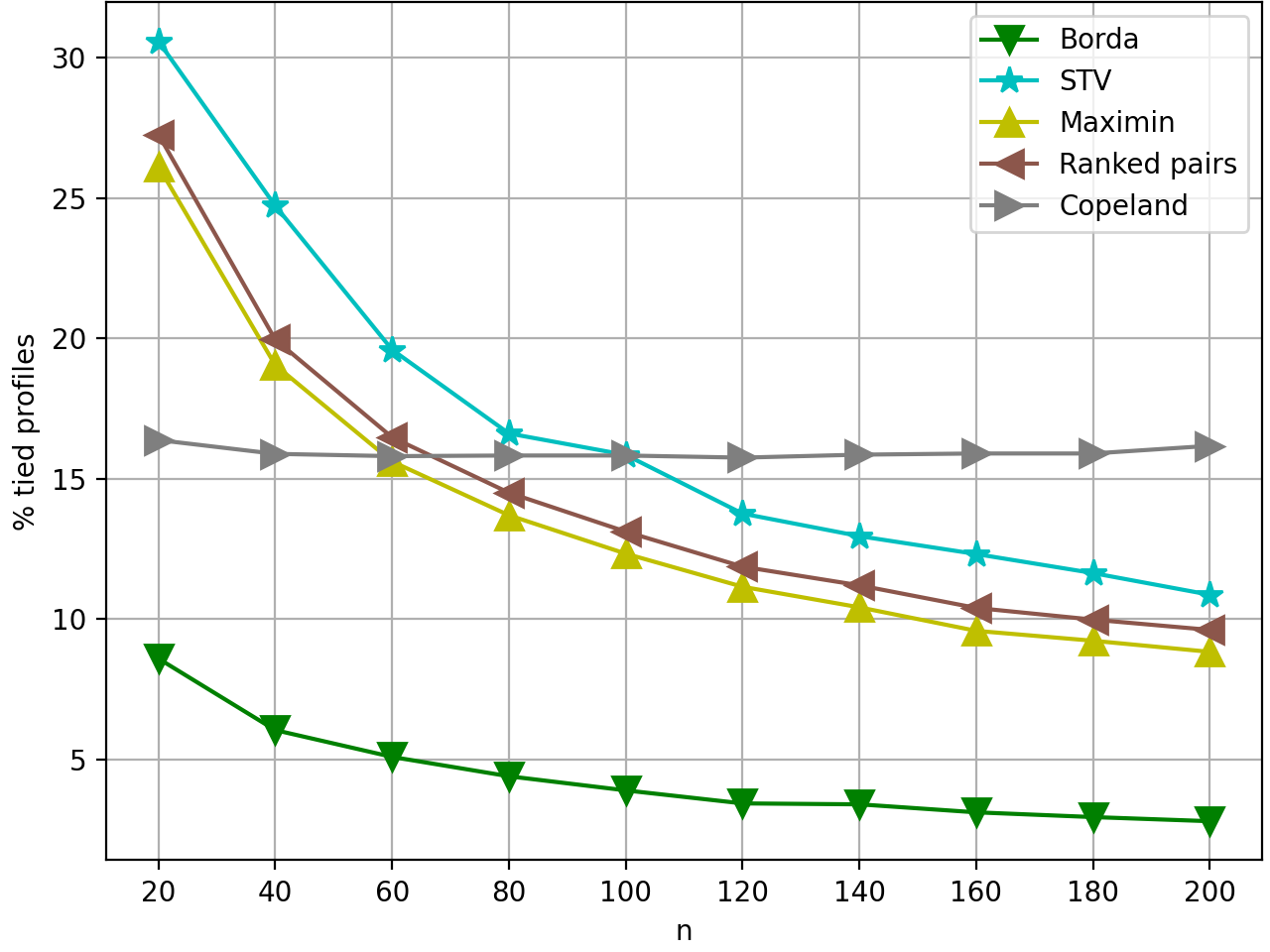}
\captionof{figure}{\small Fraction of tied profiles under IC. \label{fig:exp-short}}
\end{minipage}

\paragraph{\bf Preflib data.} Because the PUT versions of STV and ranked pairs are NP-hard to compute~\cite{Conitzer09:Preference,Brill12:Price}, we used AI-search-based implementations of STV and ranked pairs~\cite{Wang2019:Practical} and computed the fraction of tied profiles  among  the 307  profiles from Strict Order-Complete Lists (SOC) under election data category at Preflib~\cite{Mattei13:Preflib}.\footnote{
Preflib mentioned that this category  can be interpreted as election data, though not all of them come from real-life elections. 
See some statistics in Appendix~\ref{sec:exp-full}. We used the same dataset  in~\cite{Wang2019:Practical}, where PUT ranked pairs can finish in one hour.} The results are summarized in Table~\ref{tab:ties-preflib} below. We emphasize that the observations are drawn only from Preflib data and should not be interpreted as general conclusions on the likelihood of ties in presidential elections.

\begin{table}[htp]
\centering
\caption{Percentage of tied profiles in Preflib data in weakly increasing order.\label{tab:ties-preflib}}
\begin{tabular}{|c|c|c|c|c|c|c|c|c|c|}
\hline
 &\small  Borda&\small  Copeland$_{0.5}$& \small Plurality&\small  Maximin&\small  Schulze&\small   Ranked pairs&\small  STV &\small  Coombs&\small  Veto\\
\hline
Ties & 1.6\% &  2.6\% &  4.6\% &  6.8\% &  6.8\% &  6.8\%  &  7.5\% &  10.4\% &  31.3\% \\
\hline
\end{tabular}
\end{table}
Table~\ref{tab:ties-preflib} shows that ties occur least frequently under Borda (1.6\% of the profiles), which is consistent with the experiments on synthetic data in Figure~\ref{fig:exp-short}. Two interesting observations are: first,  ties are rare under Copeland$_{0.5}$ (2.6\%); and second, ties occur  frequently under veto (31.3\%), which mostly happen when the number of alternatives is larger than the number of voters---in such cases the election is guaranteed to be tied under veto. The two observations are quite different from  Figure~\ref{fig:exp-short}, which is probably because real-life preference data can be significantly different from IC,  as  widely acknowledged in the literature~\cite{Lehtinen2007:Unrealistic}.

\section{Future work} 
We see three immediate directions for future work. First, technically, how can we improve the results for more general models, especially by dropping the strictness assumption on $\Pi$? Second, how can we extend the study to other events of interest in voting, for example,  stability and margin of victory of voting rules, and more generally other topics such as multi-winner elections,  judgement aggregation, matching, and resource allocation? Third, what are the smoothed complexity in various computational aspects of voting~\citep{Baumeister2020:Towards}, such as winner determination~\citep{Xia2021:The-Smoothed}, manipulation, bribery and control?


\section{Acknowledgements}
We thank  Rupert Freeman, Qishen Han, Ao Liu, Marcus Pivato, Sikai Ruan, Rohit Vaish, Weiqiang Zheng, Bill Zwicker,  participants of the COMSOC video seminar, and anonymous reviewers for helpful comments. This work is supported by NSF \#1453542, ONR \#N00014-17-1-2621, and a gift fund from Google.


\bibliographystyle{ACM-Reference-Format}
\bibliography{/Users/administrator/GGSDDU/references}

\Omit{
\section{\bf Future Work} There are many directions and open questions for future work. The smoothed likelihood of ties can be studied for (1) more voting rules, for example Bucklin, STV, plurality with runoff, and other multi-stage voting rules, and/or (2) distributions where agents' noises are correlated, and/or (3) take $\inf$ instead of $\sup$ in Definition~\ref{dfn:smoothed-ties}, which can be used to prove positive results in the contexts where ties are desirable, e.g., voting power and voter turnout as discussed in the Introduction.
}
\newpage
\tableofcontents
\newpage

\appendix

\section{Materials for Section~\ref{sec:maintechthm}}
\subsection{Proof of Theorem~\ref{thm:maintechh}}
\label{app:maintech}
\appThm{thm:maintechh}{The main technical theorem}{
Given any $q\in\mathbb N$, any closed and strictly positive $\Pi$ over $[q]$, and any polyhedron $\poly$ with integer matrix $\ba$, for any $n\in\mathbb N$, 
\begin{align*}
&\sup_{\vec\pi\in\Pi^n}\Pr\left(\vXp \in \poly\right)=\left\{\begin{array}{ll}0 &\text{if } \polynint=\emptyset\\
\exp(-\Theta(n)) &\text{if } \polynint \ne \emptyset \text{ and }\polyz\cap\conv(\Pi)=\emptyset\\
\Theta\left(n^{\frac{\dim(\polyz)-q}{2}}\right) &\text{otherwise (i.e. } \polynint\ne \emptyset \text{ and }\polyz\cap\conv(\Pi)\ne \emptyset\text{)}
\end{array}\right.,\\
&\inf_{\vec\pi\in\Pi^n}\Pr\left(\vXp \in \poly\right)=\left\{\begin{array}{ll}0 &\text{if } \polynint=\emptyset\\
\exp(-\Theta(n)) &\text{if } \polynint \ne \emptyset \text{ and }\conv(\Pi)\not\subseteq \polyz\\
\Theta\left(n^{\frac{\dim(\polyz)-q}{2}}\right) &\text{otherwise (i.e. } \polynint\ne \emptyset \text{ and }\conv(\Pi)\subseteq\polyz
\text{)}\end{array}\right..
\end{align*}
}

\subsubsection{Proof Sketch of Theorem~\ref{thm:maintechh}}
\label{app:maintech-proof-sketch} In this subsection we  present a   proof sketch for the exponential and polynomial cases of the sup part, because the $0$ case is trivial. 

\paragraph{\bf  Intuition and proof sketch for the exponential bounds on Sup.}  We first note that  $\vXp$ is an integer-vector-valued random variable. Therefore, $\Pr(\vXp\in\poly) = \Pr(\vXp\in\polynint)$ and is mainly determined by two factors: (1) the distance between $\expect(\vXp)$ and $\polynint$, and (2) the density of integer vectors in $\polynint$. Standard concentration bounds, e.g., Hoeffding's inequality, tell us that when $\expect(\vXp)$ and $\polynint$ are $\Theta(n)$ away, the probability for $\vXp$ to be in $\polynint$ is exponentially small. This is the intuition  behind the exponential case, as illustrated in Figure~\ref{fig:mainthmex} (a). 

\paragraph{\bf  Intuition behind the polynomial  bounds on Sup.}  As illustrated in Figure~\ref{fig:mainthmex} (b), in the polynomial case,  $\cone(\Pi)$ is $O(1)$ away from $\poly$, and one may expect that $\sup_{\vec\pi\in\Pi^n}\Pr(\vXp\in\poly)$ is achieved when $\expect(\vXp)$ is close to $\poly$. Then, $\Pr(\vXp\in\poly)$ is mostly determined by the density of integer vectors in $\polynint$. A natural conjecture is that the density  can be measured by the dimension of $\poly$, but this is not true as  illustrated in Figure~\ref{fig:exthin} (a)  below, where 
$\dim(\poly)=2$, which is the same as $\dim(\poly)$  in  Figure~\ref{fig:mainthmex} (b). However,   $\Pr(\vXp\in\poly)$ in Figure~\ref{fig:exthin} (a) is smaller than that in Figure~\ref{fig:mainthmex} (b) as $n\ra\infty$, because the ``volume'' of $\polynint$ in Figure~\ref{fig:exthin} (a) does not increase as $n$ increases.

\begin{figure}[htp]
\centering 
\begin{tabular}{@{}c@{}c@{}c@{}}
\includegraphics[width = 0.48\textwidth]{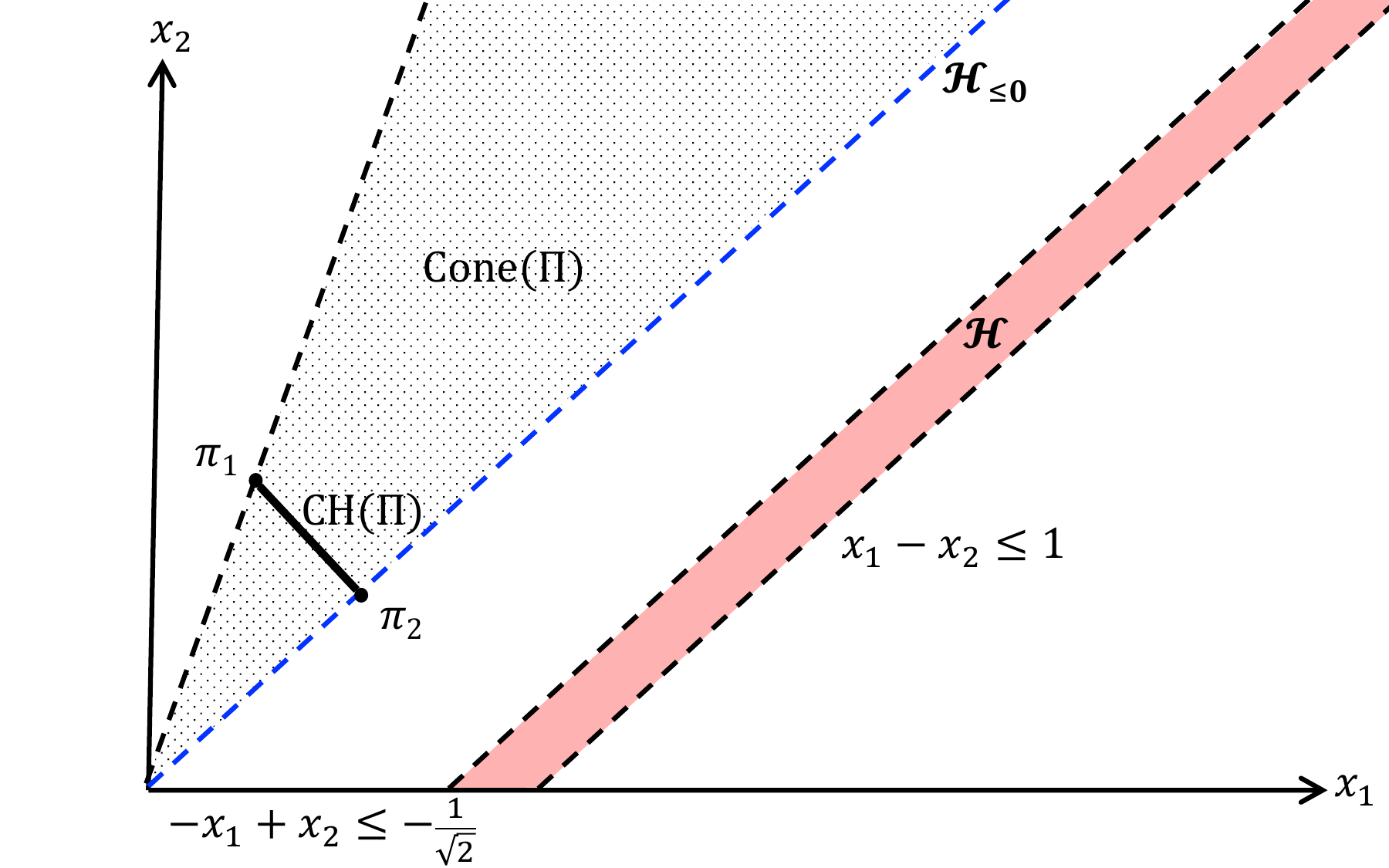} 
&\includegraphics[width = 0.24\textwidth]{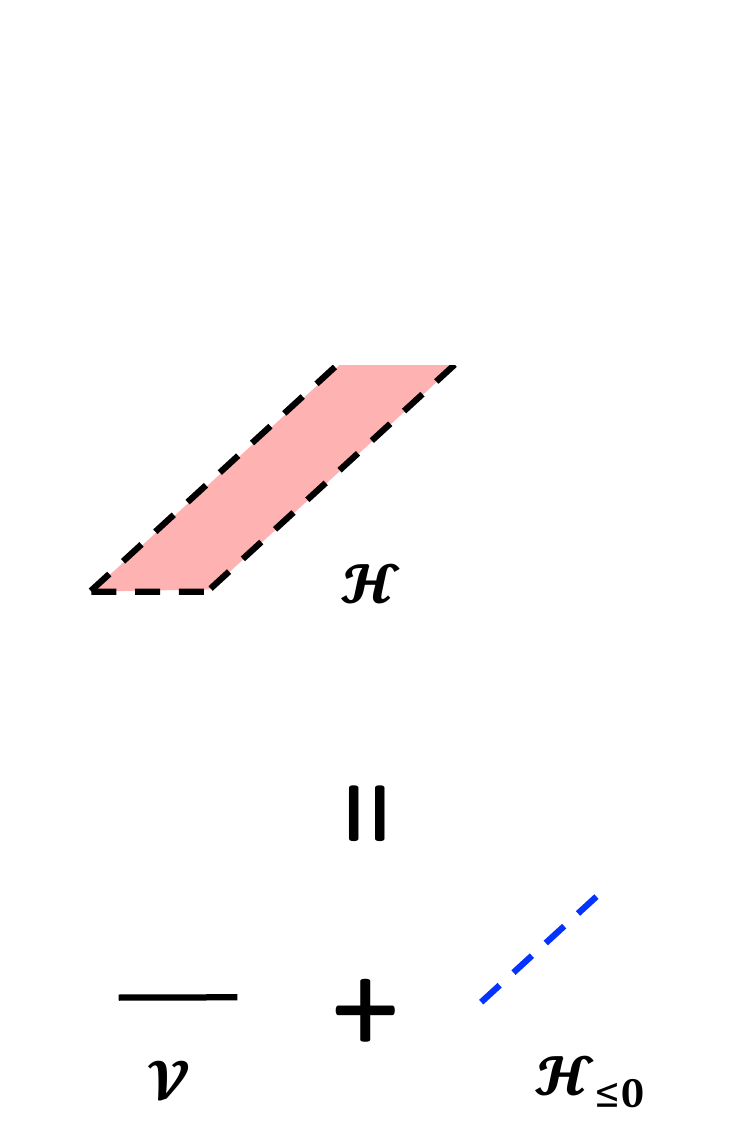} &\includegraphics[width = 0.24\textwidth]{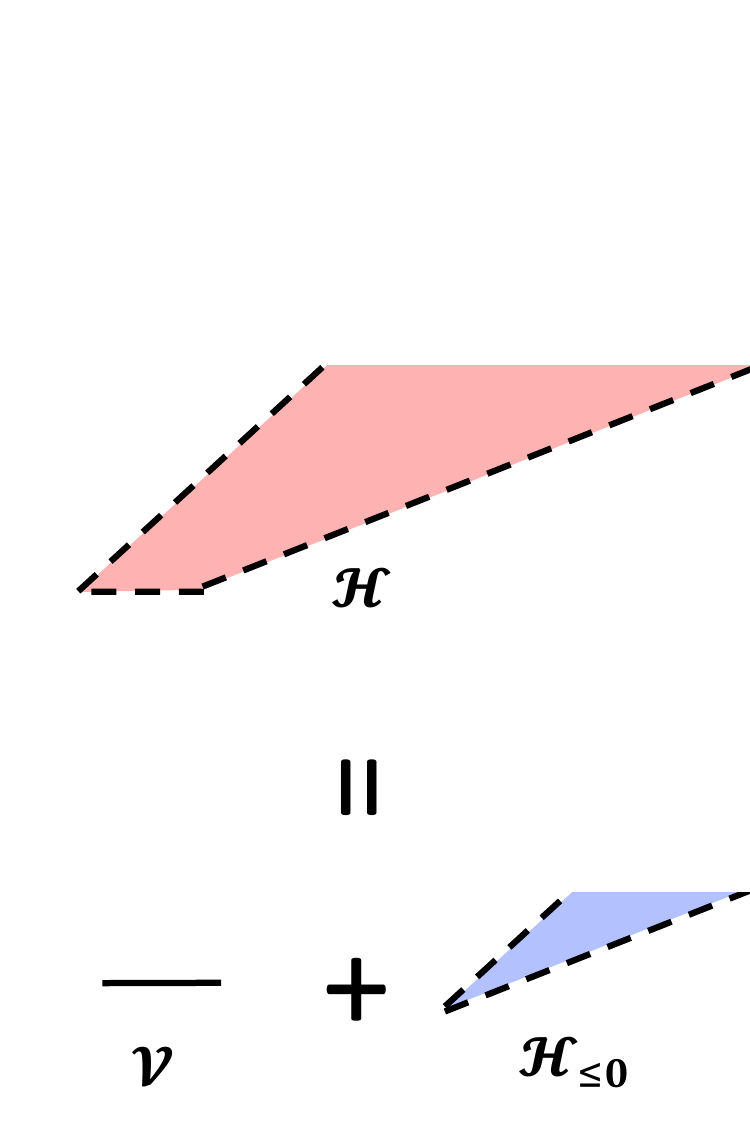} \\
\small (a) $\ba = \left[\begin{array}{rr}1&-1\\-1&1\end{array}\right]$ and $\vec b = \left[\begin{array}{r}1\\-\frac1{\sqrt 2}\end{array}\right]$. 
&\small (b) V-rep of (a). &\small (c) \small V-rep of Figure~\ref{fig:mainthmex} (b). \end{tabular}
\caption{\small An  example  of $\dim(\poly)>\dim{\polyz}$ and V-representations of $\poly$ in (a) and $\poly$ in Figure~\ref{fig:mainthmex} (b).\label{fig:exthin}}
\end{figure} 
It turns out that the dimension of $\polyz$ is the right measure. For example, $\dim(\polyz)=1$ in Figure~\ref{fig:exthin} (a) and $\dim(\polyz)=2$ in Figure~\ref{fig:mainthmex} (b). The rest of the proof leverages the interplay between the matrix representation and the {\em V-representation} of $\poly$, which is defined by the Minkowski-Weyl theorem (see, e.g., \citep[p.~100]{Schrijver1998:Theory}). More precisely, the  V-representation of $\poly$ is  $\poly = \mV+\polyz$, where $\mV$ is a finitely generated polyhedron and $\polyz$ is the characterization cone of $\poly$. See Figure~\ref{fig:exthin} (b) and (c) for the V-representations of $\poly$ in Figure~\ref{fig:exthin} (a) and  in Figure~\ref{fig:mainthmex} (b), respectively.

\paragraph{\bf  Polynomial upper bound  on Sup.}  To accurately upper-bound $\Pr(\vXp\in \poly)$, we partition the $q$ dimensions of $\vXp$ into two sets: $I_0$ and $I_1$, such that vectors in $\polynint$ can be enumerated by first enumerating their $I_1$ components with high flexibility, conditioned on which the $I_0$ components are more or less determined. More precisely,  the following two conditions are satisfied. 

{\bf  Condition (1).} {\em For any $\vec h_1\in {\mathbb Z}_{\ge 0}^{I_1}$, the restriction of $\polynint$ on $\vec h_1$, denoted by $\polynint|_{\vec h_1} = \{\vec h_0\in {\mathbb Z}_{\ge 0}^{I_0}:(\vec h_0,\vec h_1)\in\polynint\}$, contains a constant number (in $n$) of integer vectors.}

{\bf  Condition (2).} {\em With high (marginal) probability on the $I_1$ components of $\vXp$, the conditional probability for the $I_0$ components of $\vXp$ to be in $\polynint|_{\vec h_1}$ is $O\left(n^{\frac{\dim(\polyz)-q}{2}}\right)$.}

Once such $I_0$ and $I_1$ are defined, the upper bound follows after applying the law of total probability.

We use the matrix representation of $\poly$ to define $I_0$ and $I_1$ as follows, which is similar to the definitions in the proof of~\cite[Lemma 1]{Xia2020:The-Smoothed} except that our definition works for general $\ba$. Let $\ba^{=}$ denote the {\em implicit equalities} of $\ba$, which is the maximum set of rows of $\ba$ such that for all $\vec x\in \polyz$, we have $\ba^{=}\cdot \invert{\vec x} = \invert{\vec 0}$. We note that $\ba^=$ does not depend on $\vec b$, and $\rank(\ba^=)=q-\dim(\polyz)$~\citep[Equation (9), p.~100]{Schrijver1998:Theory}.  For example, in  Figure~\ref{fig:mainthmex} (b), $\ba^{=}=\emptyset$ and $\dim(\polyz) = 2$.

We then use the reduced row echelon form (a.k.a.~row canonical form)~\citep{Meyer2000:Matrix} of $\left[\begin{array}{c}\ba^{=}\\ \vec 1\end{array}\right]$ to define $I_0$ and $I_1$. More precisely, we apply the Gauss-Jordan elimination method to convert  the system of linear equations  $\left[\begin{array}{c}\ba^{=}\\ \vec 1\end{array}\right] \cdot \invert{\vec x} = \invert{\vec 0, n}$ to another system of linear equations $\invert{\vec x_{I_0}} = \bd\cdot \invert{\vec x_{I_1}, n}$, where $\vec x_{I_0}$ are the $I_0$ components of $\vec x$, $I_0\cup I_1 = [q]$, $|I_0| = \rank(\ba^=)+1 =q-\dim(\polyz) +1$, and  $\bd$  is an $(q-\dim(\polyz)+1)\times \dim(\polyz)$ rational matrix that does not depend on $n$.    For example, in  Figure~\ref{fig:mainthmex} (b), $I_0 = \{1\}$, $I_1=\{2\}$, and $\bd = \left[\begin{array}{c}-1\ \ 1\end{array}\right]$. See~\citep[Example~4 in the Appendix]{Xia2020:The-Smoothed} for a more informative example.

Then, we prove in Claim~\ref{claim:bounded} in Appendix~\ref{app:maintech-full-proof} that {\bf Condition (1)}  above holds. {\bf  Condition (2)} is proved by applying the point-wise anti-concentration bound~\citep[Lemma 3 in the Appendix]{Xia2020:The-Smoothed} and an alternative representation of the PMV as a simple Bayesian network as done in~\citep{Xia2020:The-Smoothed}.
 
\setlength{\columnsep}{7pt}
\setlength{\intextsep}{5pt}
\begin{wrapfigure}{R}{0.55\textwidth}
\centering
  \includegraphics[width = \linewidth]{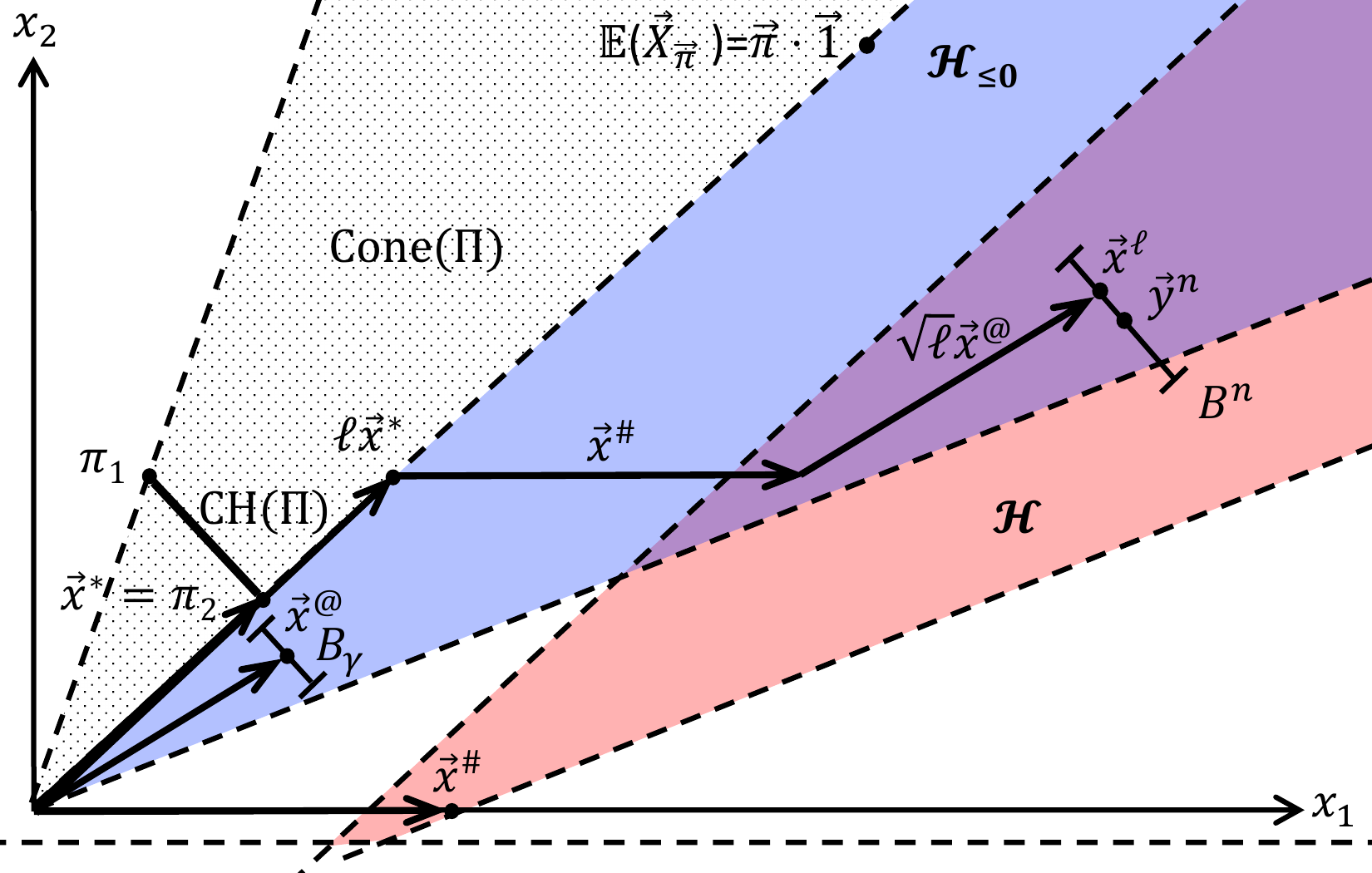}
\vspace{-8mm}
\caption{\small Proof of the poly lower bound in Theorem~\ref{thm:maintechh}. \label{fig:ex-poly-lower}
}
\end{wrapfigure}
\paragraph{\bf  Polynomial lower bound  on Sup.}  The proof of this part is the hardest and drastically different from the proofs in~\cite{Xia2020:The-Smoothed}. We will specify $\vec \pi$ and a $(\dim(\polyz)-1)$-dimensional region $B^n$  in $\polyn$  that is $O(\sqrt n)$ away from $\expect(\vXp)$ based on the V-representation of $\poly$, as illustrated in Figure~\ref{fig:ex-poly-lower}, which continues the setting of Figure~\ref{fig:mainthmex} (b). 
More precisely, we first choose the following three vectors arbitrarily and then fix them throughout the proof: let $\vec x^*\in \polyz\cap \conv(\Pi)$, $\vec x^{\#}\in \poly\cap {\mathbb Z}_{\ge 0}^{q}$, and let $\vec x^@$ be an inner point of $\polyz$. For example, in Figure~\ref{fig:ex-poly-lower}  we let $\vec x^* = \pi_2$ because $\pi_2$ is the only vector in $\polyz\cap \conv(\Pi)$.

Given any $n\in\mathbb N$, we define $\vec x^{\ell} \triangleq \ell\vec x^* + \vec x^{\#}+\sqrt\ell  \vec x^@$, where $\ell\in \mathbb R_{\ge 0}$ is chosen to  guarantee that  $\vec x^{\ell}\cdot \vec 1 = n$.  Then, we define the following vectors. 

$\bullet$ An integer vector $\vec y^n\in \polynint$ that is $O(1)$ away from $\vec x^{\ell}$ in $L_\infty$. The existence of such $\vec y^n$ is guaranteed by the sensitivity analysis of integer programming (\citep[Theorem~1(i)]{Cook86:Sensitivity} with $w = \vec 0$).

$\bullet$ A vector $\vec \pi = (\pi_1,\ldots,\pi_n)\in \Pi^n$ such that $\expect(\vXp) =  \sum_{j=1}^n \pi_j $  is $O(\sqrt n)$ away from $\vec x^{\ell}$. This is done by rounding $n$ multiplied by the representation of $\vec x^*$ as the convex combination of no more than $q$ distributions in $\Pi$, which is guaranteed by the Carathéodory's theorem for convex hulls.

$\bullet$ An $\Omega(\sqrt n)$ neighborhood of $\vec y^n$ in $\polyn$, denoted by $B^n$. This is done by first defining an $\gamma$ neighborhood of $\vec x^@$ in $\polyz$, denoted by $B_\gamma$, and then letting $B^n \triangleq \vec y^n+\sqrt\ell (B_\gamma - \vec x^@)$. 

The construction guarantees that $B^n$ contains $\Omega\left(n^{\frac{\dim(\polyz)-1}{2}}\right)$ many integer vectors, each of which is $O(\sqrt n)$ away from $\expect(\vXp)$. Then, we prove a point-wise concentration bound in Lemma~\ref{lem:point-wise-concentration-PMV} (which works for arbitrary strictly positive PMVs and is thus stronger than~\cite[Lemma~4 in the Appendix]{Xia2020:The-Smoothed}, which only holds for i.i.d.~PMVs) to show that for each integer vector $\vec x\in B^n$, the probability for $\vXp$ to take $\vec x$ is $\Omega(n^{\frac{1-q}{2}})$. The lower bound then becomes 
$$\Omega\left(n^{\frac{\dim(\polyz)-1}{2}}\right)\times \Omega\left(n^{\frac{1-q}{2}}\right) = \Omega\left(n^{\frac{\dim(\polyz)-q}{2}}\right),$$
which matches the upper bound asymptotically.

\appLem{lem:point-wise-concentration-PMV}{Point-wise concentration bound for PMVs}{For any $q\in \mathbb N$, any $\epsilon>0$, and any $\alpha>0$, there exists  $C_{q,\epsilon,\alpha}>0$  such that for any $n\in\mathbb N$, any $(n,q)$-PMV $\vXp$ where $\vec \pi$ is above $\epsilon$,  and any integer vector $\vec x \in {\mathbb Z^{q}_{\ge 0}}$  with $\vec x\cdot 1 = n$ and  $|\vec x - \expect(\vXp)|_\infty<\alpha\sqrt n$, we have:

 $\hfill\Pr(\vXp=\vec x)>C_{q,\epsilon,\alpha}\cdot n^{\frac{1-q}{2}}\hfill$
}

Lemma~\ref{lem:point-wise-concentration-PMV} is proved by extending the idea of solving a constrained optimization problem for Poisson binomial variables~\citep{Hoeffding1956:On-the-Distribution} to PMVs. Lemma~\ref{lem:point-wise-concentration-PMV}  might be of independent interest.

\subsubsection{Full Proof of  Theorem~\ref{thm:maintechh}}
\label{app:maintech-full-proof}
\begin{proof} It suffices to prove the theorem holds for all sufficiently large $n$. In other words, we will prove that given $\poly$ and $\Pi$, there exists $N\in\mathbb N$ such that the theorem holds for all $n>N$. This is because given any constant $N\in\mathbb N$, the theorem trivially holds for  and any $n\le N$---notice that any $\sup_{\vec\pi\in\Pi^n}\Pr\left(\vXp \in \poly\right) \ne 0$ (respectively, $\inf_{\vec\pi\in\Pi^n}\Pr\left(\vXp \in \poly\right)\ne 0$) can be viewed as $\exp(-\Theta(n))$ or $\Theta\left(n^{\frac{\dim(\polyz)-q}{2}}\right) $, and the zero case is true for any $n$. 


\paragraph{\bf Proof of the exponential upper bound on Sup.}  For any $n\in\mathbb N$ and any $\vec \pi\in \Pi^n$, let $\vec \mu_{\vec\pi} =(\mu_{\vec\pi,1},\ldots,\mu_{\vec\pi,q}) = \expect( \vXp/n)$ denote the mean of $\vXp/n$ and let $\vec\sigma_{\vec\pi}=(\sigma_{\vec\pi,1},\ldots,\sigma_{\vec\pi,q})$ denote the standard deviations of each component in $\vXp/n$. 
Because $\Pi$ is strictly positive, there exists $\epsilon_1 >0, \epsilon_2>0$ such that for all $n$, all $\vec \pi\in \Pi^n$, and all $i\le q$, we have $\epsilon_1< \mu_{\vec\pi,i}<\epsilon_2$ and $\frac{\epsilon_1}{\sqrt n}< \sigma_{\vec\pi,i}<\frac{\epsilon_2}{\sqrt n}$.  

We first prove that $\polyz$ is  sufficiently separated from $\conv(\Pi)$.  Notice that $\polyz$ is convex and closed by definition. Because $\Pi$ is closed and bounded, $\conv(\Pi)$ is convex, closed and compact. Because $\polyz\cap \conv(\Pi)=\emptyset$, by the strict hyperplane separation theorem, there exists a hyperplane that strictly separates  $\polyz$ and $\conv(\Pi)$. In other words, there exists $\epsilon'>0$ such that for any $\vec x_1\in \polyz$ and any $\vec x_2\in \conv(\Pi)$, we have $|\vec x_1-\vec x_2|_\infty>\epsilon'$, where $|\cdot|_\infty$ is the $L_\infty$ norm. 

We then prove that any $\vec x\in \polynint$ is $\Theta(n)$ away from $\expect(\vXp) = n \vec \mu_{\vec \pi}$ when $n = \vec x\cdot \vec 1$ is sufficiently large. By the Minkowski-Weyl theorem, we can write $\poly = \mV+\polyz = \{\vec v+\vec h: \vec v\in \mV, \vec h\in\polyz\}$. Let $C_{\max} = \max_{\vec x'\in \mV}|\vec x'|_\infty$. For any $\vec x\in \poly$, let $\vec x = \vec x'+\vec x_1$ where $\vec x'\in\mV$ and $\vec x_1\in \polyz$. We have:
$$|\vec x - \expect(\vXp)|_\infty = |\vec x - n \vec \mu_{\vec \pi}|_\infty =|\vec x' +\vec x_1 - n \vec \mu_{\vec \pi}|_\infty \ge n|\frac{\vec x_1}{n} - \vec \mu_{\vec \pi}|- C_{\max} \ge n\epsilon' -  C_{\max},$$ 
Therefore, when $\vXp \in \polynint$, $\vXp$ must be away from $\expect(\vXp)$ by at least $n\epsilon' -  C_{\max}$ in L$_\infty$. For any $n>\frac{2C_{\max}}{\epsilon ' }$, we have $n\epsilon' -  C_{\max}>\frac{\epsilon'}{2}n$. Therefore, 
\begin{align*}
&\Pr\left(\vXp \in \polynint\right) \le
\Pr\left( |\vXp -\expect(\vXp)|_\infty> n\epsilon' -  C_{\max}\right)
\le 
 \Pr\left( |\vXp - n\vec \mu_{\vec \pi}|_\infty> \frac{\epsilon'}{2}n\right)\\
  \le&\sum_{i=1}^q \Pr\left(|X_{\vec \pi,i} -n\mu_{\vec \pi,i}|>\frac{\epsilon'}{2} n\right)\le 2q\exp\left(-\frac{(\epsilon')^2 n }{4(1-2\epsilon)^2}\right)
\end{align*}
The last inequality follows after  Hoeffding's inequality (Theorem 2 in~\citep{Hoeffding63:Probability}), where recall that  $\epsilon$ is a constant such that $\Pi$ is above $\epsilon$. 

\paragraph{\bf Proof of the exponential lower bound  on Sup.}  In fact the lower bound can be achieved by any $\vec \pi\in \Pi^n$. Let $\vec y\in [q]^n$ denote an arbitrary vector such that $\hist(\vec y)\in\polynint$. Because $\Pi$ is above $\epsilon$, we have $\Pr(\vXp\in\poly) \ge \Pr(\vec Y = \vec y)  \ge \epsilon^{n} = \exp(n\log \epsilon)$, which is $\exp(- O(n))$.

\paragraph{\bf Proof of the polynomial upper bound on Sup.} We use the V-representation of $\poly = \mV+\polyz$ in this part of the proof.  Moreover, we will use the equivalent representation of $\ba$ as the {\em implicit equalities}, denoted by $\ba^=$, and {\em other inequalities}, denoted by $\ba^+$, formally defined as follows.
\begin{dfn}[\bf (2) on page 99 of~\citep{Schrijver1998:Theory}] For any integer matrix $\ba$, let $\ba^{=}$ denote the {\em implicit equalities}, which is the maximal set of rows of $\ba$ such that for all $\vec x\in \polyz$, we have $\ba^{=}\cdot \invert{\vec x} = \invert{\vec 0}$. Let $\ba^{+}$ denote the remaining rows of $\ba$.
\end{dfn}

We note that $\ba^=$ and $\ba^+$ do not depend on $\vec b$. As we will see soon, $\ba^=$ is the main constraint for  $\vXp$ to be in $\polyz$. To simplify notation, throughout the proof we let $o=\rank(\ba^=)=q-\dim(\polyz)$, where the equation holds due to~\citep[p.~100, Equation (9)]{Schrijver1998:Theory}. The following running example illustrates these notions and its setting will be used throughout this proof.

\begin{ex}[\bf\boldmath Running example: $\poly$ and $\ba^=$ for Borda winners being $\{1,2\}$ under IC]\label{ex:bordak2}We use a sub-case of  $k=2$ way ties over $m=3$ alternatives under Borda w.r.t.~IC for example. Notice that $q= m! = 6$. Each of the six outcomes is a linear order. Let  $1\succ 2\succ 3, 1\succ 3\succ 2, 2\succ 1\succ 3, 2\succ 3\succ 1, 3\succ 1\succ 2, 3\succ 2\succ1 $ denote outcomes $1,2,3,4,5,6$, respectively. 

For any $n\in \mathbb N$, let $\vec X_{\text{IC}}$ denote the histogram of random profile under IC. 
Then, the Borda co-winners are $\{1,2\}$ if and only if $\vec X_{\text{IC}}$ is in  polyhedron $\poly$ represented by the following linear inequalities, where the variables are $\vec x = (x_{123},x_{132},x_{213},x_{231},x_{312},x_{321})$:
\begin{align}
&x_{123}+2 x_{132} - x_{213}-2 x_{231} + x_{312}  -  x_{321}\le 0 \label{eq:borda1}\\
-&x_{123}-2 x_{132} + x_{213}+2 x_{231} - x_{312}  +  x_{321}\le 0  \label{eq:borda2}\\
-&2x_{123}- x_{132} - x_{213}+ x_{231} + x_{312}  + 2 x_{321}\le -1\label{eq:borda3}
\end{align}
Equation~(\ref{eq:borda1}) states that the Borda score of alternative $1$ is no more than the Borda score of alternative $2$. Equation~(\ref{eq:borda2}) states that the Borda score of  $2$ is no more than the Borda score of  $1$. Equation~(\ref{eq:borda3}) states that the Borda score of  $3$ is at least one less than the Borda score of alternative $1$. Because Borda scores are always integers, Equation~(\ref{eq:borda3})  is equivalent to requiring that the Borda score of  $3$ is strictly smaller than the Borda score of $1$. Then, we have:
$$\ba = \left[\begin{array}{rrrrrr}
1& 2& -1& -2&1&-1\\
-1& -2& 1& 2&-1&1\\
-2& -1& -1& 1&1&2
\end{array}\right]\text{ and }\vbb = \left[\begin{array}{r}0\\0\\-1\end{array}\right]$$
It is not hard to verify that $\ba^=$ consists of the first two rows, i.e.,
$$\ba^= = \left[\begin{array}{rrrrrr}
1& 2& -1& -2&1&-1\\
-1& -2& 1& 2&-1&1
\end{array}\right], \text{ and } o = \rank(\ba^=)=1
$$
\end{ex}

To calculate $\Pr(\vXp\in\poly)$, we will focus on the reduced row echelon form (a.k.a.~row canonical form)~\citep{Meyer2000:Matrix} of $\left[\begin{array}{c}\ba^{=}\\ \vec 1\end{array}\right]$, which can be computed by Gauss-Jordan elimination: there exists $I_0\subseteq [q]$ with $|I_0|=\rank\left(\left[\begin{array}{c}\ba^{=}\\ \vec 1\end{array}\right]\right) $ and a rational matrix $\bd$ such that $\left[\begin{array}{c}\ba^{=}\\ \vec 1\end{array}\right]\cdot \invert{\vec x} = \left[\begin{array}{c}\invert{\vec 0}\\ n\end{array}\right]$ if and only if $\invert{\vec x_{I_0}} = \bd\cdot \invert{\vec x_{I_1}, n}$, where 
$I_1=[q]\setminus I_0$. In other words, $\vec x_{I_1}$ can be viewed as ``free'' variables whose value can be quite flexible, and for any $\vec x\in \polyz$, $\vec x_{I_0}$ is completely determined by $\vec x_{I_1}$. 

We have $|I_0|=\rank\left(\left[\begin{array}{c}\ba^{=}\\ \vec 1\end{array}\right]\right) = \rank(\ba^=)+1$ because  $\vec 1$ is linearly independent with the rows in $\ba^=$. To see why this is true, suppose for the sake of contradiction that $\vec 1$ is linear dependent with rows in $\ba^=$. Then, by the definition of $\ba^=$, for all $\vec x\in\polyz$ we have $\vec x\cdot \vec 1 = 0$. Therefore,  for any $\vec x'\in \poly$, according to Minkowski-Weyl theorem, we can write $\vec x' = \vec v +\vec h$, where $\vec h\in \polyz$ and $\vec v$ is in a finitely generated polyhedron. This means that 
$$\vec x' \cdot\vec 1 = \vec v \cdot\vec 1+\vec h\cdot\vec 1 =\vec v \cdot\vec 1,$$
which means that 
$\vec x' \cdot\vec 1$ is upper bounded by a constant. However, this contradicts the premise of the polynomial case, because when $n$ is sufficiently large, $\polyn=\polynint=\emptyset$.  W.l.o.g.~let $I_0=\{1,\ldots, o+1\}$ and  $I_1=\{o+2,\ldots, q\}$. We also note that $\bd$ does not depend on $n$, which means that for any $\vec x\in{\mathbb R}^q$, $\ba^=\cdot\invert{\vec x} = \invert{\vec 0}$ if and only if $\invert{\vec x_{I_0}} = \bd\cdot \invert{\vec x_{I_1},  \vec x\cdot\vec 1}$. 

The following example illustrates Gauss-Jordan elimination in the setting in Example~\ref{ex:bordak2}.

\begin{ex}[\bf\boldmath Running example: Gauss-Jordan elimination, $\bd$, $I_0$, and $I_1$]
\label{ex:gaussian}
Continuing  Example~\ref{ex:bordak2}, we run Gauss-Jordan elimination as follows. 
\begin{align*}
\left[\begin{array}{r|c}\ba & 0\\ \vec 1 & n\end{array}\right] = &\left[\begin{array}{r r r r r r|c}
1& 2 &-1&-2&1&-1&0\\-1& -2 &1&2&-1&1&0\\
1&1&1&1&1&1& n\end{array}\right]\\
\xrightarrow{R1; R2+R1;R3-R1}&\left[\begin{array}{r r r r r r|c}
1& 2 &-1&-2&1&-1&0\\0& 0 &0&0&0&0&0\\
0&-1&2&3&0&2& n\end{array}\right]\\
\xrightarrow{R1+2R3; R2;-R3}& \left[\begin{array}{r r r r r r|c}
1& 0 &3&4&1&3&2n\\0& 0 &0&0&0&0&0\\
0&1&-2&-3&0&-2& -n\end{array}\right]
\end{align*}
Let $\vec x = [x_{123},x_{132},x_{213},x_{231},x_{312},x_{321}]$, we have 
$$ \left[\begin{array}{r r r r r r}
1& 0 &3&4&1&3 \\   
0&1&-2&-3&0&-2  \end{array}\right]\cdot \invert{\vec x } =\left[\begin{array}{r} 2n \\-n \end{array}\right],$$
which is equivalent to
$$\left[\begin{array}{r } x_{123}\\x_{132}
  \end{array}\right] = \left[\begin{array}{r r  }
1& 0  \\   
0&1  \end{array}\right]\times \left[\begin{array}{r } x_{123}\\x_{132}
  \end{array}\right] = \left[\begin{array}{r r r r r r}
  -3&-4&-1&-3 &2\\   
 2&3&0&2 &-1  \end{array}\right] \times \left[\begin{array}{c}x_{213}\\x_{231}\\x_{312}\\x_{321}\\n \end{array}\right]$$
Therefore, we have $\bd = \left[\begin{array}{r r r r r r}
  -3&-4&-1&-3 &2\\   
 2&3&0&2 &-1  \end{array}\right]$, $I_0 = \{1,2\}$, $I_1=\{3,4,5,6\}$. 
 \end{ex}

As in~\citep{Xia2020:The-Smoothed}, we adopt the following alternative representation of $Y_1,\ldots,Y_n$. For each $j\le n$, we use a random variable $Z_j\in\{0,1\}$ to represent whether the outcome of $Y_j$ is in $I_0$ (corresponding to $Z_j=0$) or is in $I_1$ (corresponding to $Z_j=1$). Then, we use another random variable $W_j\in [q]$ to represent the outcome of $Y_j$ conditioned on $Z_j$.

\begin{dfn}[\bf Alternative representation of $\bm{Y_1,\ldots,Y_{n}}$~\citep{Xia2020:The-Smoothed}]\label{dfn:altfory} For each $j\le n$, we define a Bayesian network with two random variables $Z_j \in \{0,1\}$ and $W_j\in [q]$, where $Z_j$ is the parent of $W_j$. The conditional probabilities are as follows.
\begin{itemize}
\item For each $\ell\in \{0,1\}$, $\Pr(Z_j = \ell) = \Pr(Y_j \in I_\ell)$. 
\item For each $\ell\in \{0,1\}$ and each $t\le q$, $\Pr(W_j = t|Z_j=\ell) = \Pr(Y_j = t|Y_j\in I_\ell)$.
\end{itemize} In particular, if $t\not\in I_\ell$ then $\Pr(W_j = t|Z_j=\ell)=0$.
\end{dfn}

\begin{ex}[\bf\boldmath Running example: alternative representation of  uniformly distributed $Y_j$]
\label{ex:y} For the purpose of presentation, we present $Z_j$ and $W_j$ for $Y_j$ that corresponds to the uniform distribution over $\ml(\ma)$. We have $\Pr(Z_j = 0) = 1/3$ and $\Pr(Z_j = 1) = 2/3$.
\begin{align*}
&\Pr(W_j = 1\succ 2\succ 3 | Z_j=0)=\Pr(W_j = 1\succ 3\succ 2 | Z_j=0) =\frac 12, \text{ and}\\
&\Pr(W_j = 2\succ 1\succ 3 | Z_j=1)=\Pr(W_j = 2\succ 3\succ1 | Z_j=1)\\
&=\Pr(W_j = 3\succ 1\succ 2 | Z_j=1)=\Pr(W_j =3\succ 2\succ 1 | Z_j=1)= \frac 14
\end{align*} 
All conditional probabilities not defined above are zeros.
\end{ex}

Applying the law of total probability, it is not hard to verify  that for any $j\le n$, $W_j$ follows the same distribution as $Y_j$. For any $\vec z\in \{0,1\}^n$, we let $\ind_0(\vec z)\subseteq [n]$ denote the indices of components of $\vec z$ that equal to $0$. Given $\vec z$, we define the following random variables. 
\begin{itemize}
\item Let $\vec W_{\ind_0(\vec z)}=\{W_j:j\in\ind_0(\vec z)\}$. That is, $\vec W_{\ind_0(\vec z)}$ consists of random variables $\{W_j: z_j = 0\}$.
\item Let $\hist(\vec W_{\ind_0(\vec z)})$ denote the vector of the $o+1 = |I_0|$ random variables that correspond to the histogram of $\vec W_{\ind_0(\vec z)}$ restricted to $I_0$. Technically, the domain of  every random variable in $\vec W_{\ind_0(\vec z)}$ is $[q]$, but since  they only receive positive probabilities on $I_0$, they are treated as random variables over $I_0$ when $\hist(\vec W_{\ind_0(\vec z)})$ is defined. 
\item Similarly, let $\vec W_{\ind_1(\vec z)}=\{W_j:j\in\ind_1(\vec z)\}$ and let $\hist(\vec W_{\ind_1(\vec z)})$ denote  the vector of $|I_1| = q-o-1$ random variables that correspond to the histogram of $\vec W_{\ind_1(\vec z)}$.
\end{itemize}
\begin{ex}[\bf \boldmath Running example: $\ind_0(\vec z)$ and $\ind_1(\vec z)$]
\label{ex:z}
Continuing Example~\ref{ex:y}, suppose $n = 5$ and $\vec z = (0,1,1,0,1)$. We have 
\begin{itemize}
\item $\ind_0(\vec z) = \{1,4\}$, $\vec W_{\ind_0(\vec z)} =\vec W_{\{1,4\}}= \{W_1,W_4\}$, and $\hist(\vec W_{\ind_0(\vec z)})$ represents the histogram of two i.i.d.~uniform distributions over $\{1\succ2\succ 3,1\succ 3\succ 2\}$.
\item  $\ind_1(\vec z) = \{2,3,5\}$,  $\vec W_{\ind_1(\vec z)}  =\vec W_{\{2,3,5\}}= \{W_2,W_3, W_5\}$,   and $\hist(\vec W_{\ind_1(\vec z)})$ represents the histogram of three i.i.d.~uniform distributions over $\{2\succ1\succ 3,2\succ3\succ 1,3\succ1\succ 2,3\succ2\succ 1\}$.
\end{itemize}
\end{ex}

 For any $\vec h_{1}\in {\mathbb Z}_{\ge 0}^{q-o-1}$, we let $\polynint|_{\vec h_{1}}$ denote the $I_0$ components of $\vec h\in \polynint$ whose $I_1$ components are $\vec h_1$. 
Formally,  
$$\polynint|_{\vec h_{1}}=\{\vec h_0\in {\mathbb Z}_{\ge 0}^{o+1}: (\vec h_0,\vec h_1)\in \polynint\}$$ 
By definition, $\vec x\in \polynint$ if and only if $\vec x_{I_0}\in \polynint|_{\vec x_{I_1}}$.  We note that $\polynint|_{\vec h_{1}}$ can contain two or more elements, because even though $\vec x_{I_0}$ is completely determined by $\vec x_{I_1}$ for any $\vec x\in \polyz$, this relationship may not hold for $\vec x\in \polynint$. Later in Claim~\ref{claim:bounded} we will prove that the number of vectors in $\polynint|_{\vec h_{1}}$ is bounded above by a constant that does not depend on $n$.

\begin{ex}[\bf \boldmath Running example: $\polynint|_{\vec h_{1}} $]\label{ex:h} Continuing Example~\ref{ex:z}, let $\vec h_1 = (1,1,1,0)$. Then, we have $\polynint|_{\vec h_{1}} = \{(2,0)\}$. Notice that in this example $|\polynint|_{\vec h_{1}}|\le 1$ for all $\vec h_1$, because the $\vbb$ components corresponding to $\ba^=$ are $0$'s, which means that  $\vec h_0$ is determined by $\vec h_1$, i.e., $\vec h_0 = \bd \cdot \invert{\vec h_1,n}$. 

It is possible that $\polynint|_{\vec h_{1}}=\emptyset$, because $\vec h_0 = \bd\cdot\invert{\vec h_1,n}$  may not be a vector of non-negative integers. For example, when $\vec h_1 = (1,2,0,0)$, we have $\vec h_0=\bd\cdot\invert{\vec h_1,n} = (-1,3)$, which means that $\polynint|_{\vec h_{1}}=\emptyset$.
\end{ex}

Next, we apply the law of total probability to the ($\vec Z$, $\vec W$) representation of $\vXp$ and $\polynint|_{\vec h_{1}}$, to obtain the following estimate on $\Pr(\vXp\in \poly)=\Pr(\vXp\in \polyn)=\Pr(\vXp\in \polynint)$.
\begin{align}
&\Pr(\vXp\in \poly) =\Pr(\vXp\in \polynint)\notag\\
 =&\sum_{\vec z\in \{0,1\}^n}\Pr(\vec Z = \vec z) \Pr\left(\vXp\in \polynint\;\middle\vert\; \vec Z = \vec z \right) \hspace{15mm}(\text{\bf The law of total probability})\notag\\
=&\sum_{\vec z\in \{0,1\}^n}\Pr(\vec Z = \vec z) \Pr\left(\hist(\vec W_{\ind_0(\vec z)}) \in \polynint|_{\hist(\vec W_{\ind_1(\vec z)})} \;\middle\vert\; \vec Z = \vec z \right)\notag\\
=&\sum_{\vec z\in \{0,1\}^n}\Pr(\vec Z = \vec z)\sum_{\vec h_1\in {\mathbb Z}_{\ge 0}^{q-o-1}}\Pr\left(\hist(\vec W_{\ind_1(\vec z)}) =\vec h_1 \;\middle\vert\; \vec Z = \vec z\right)\notag\\
&\times \Pr\left(\hist(\vec W_{\ind_0(\vec z)}) \in \polynint|_{\vec h_1} \;\middle\vert\; \vec Z = \vec z, \hist(\vec W_{\ind_1(\vec z)}) =\vec h_1\right) \hspace{5mm}(\text{\bf The law of total probability}) \notag
\\
=&\sum_{\vec z\in \{0,1\}^n}\Pr(\vec Z = \vec z)\sum_{\vec h_1\in {\mathbb Z}_{\ge 0}^{q-o-1}}\Pr\left(\hist(\vec W_{\ind_1(\vec z)}) =\vec h_1 \;\middle\vert\; \vec Z = \vec z\right)\times \Pr\left(\hist(\vec W_{\ind_0(\vec z)}) \in \polynint|_{\vec h_1} \;\middle\vert\; \vec Z = \vec z\right)\label{eq:independence}
\\
\le &\sum_{\vec z\in \{0,1\}^n: |\ind_0(\vec z)| \ge 0.9 \epsilon n} \Pr(\vec Z = \vec z) \sum_{\vec h_1\in {\mathbb Z}_{\ge 0}^{q-o-1}}\Pr\left(\hist(\vec W_{\ind_1(\vec z)}) =\vec h_1 \;\middle\vert\; \vec Z = \vec z\right)\notag\\
&\hspace{40mm}\times \Pr\left(\hist(\vec W_{\ind_0(\vec z)}) \in \polynint|_{\vec h_1} \;\middle\vert\; \vec Z = \vec z\right) + \Pr(|\ind_0(\vec z)| < 0.9\epsilon n)\label{eq:histz}
\end{align}
where we recall that $|\ind_0(\vec z)| $ denotes the number of $0$'s in $\vec z$. Equation (\ref{eq:independence}) holds because according to the Bayesian network structure, $W_j$'s are independent of each other given $Z_j$'s, which means that for any $\vec z\in\{0,1\}^n$, $\hist(\vec W_{\ind_0(\vec z)})$ and $\hist(\vec W_{\ind_1(\vec z)})$ are independent given $\vec z$. The following example illustrates the summand in (\ref{eq:histz}) when $\vec z = (0,1,1,0,1)$ and $0.9\epsilon<\frac 25$, following the setting of Example~\ref{ex:h}. 
\begin{ex}[\bf \boldmath Running example: summand in (\ref{eq:histz})]
\label{ex:sum}
Continuing Example~\ref{ex:h}, the summand in (\ref{eq:histz})  becomes the following:
\begin{align*}
& \Pr(\vec Z = \vec z) \sum_{\vec h_1\in {\mathbb Z}_{\ge 0}^{q-o-1}}\Pr\left(\hist(\vec W_{\ind_1(\vec z)}) =\vec h_1 \;\middle\vert\; \vec Z = \vec z\right) \times \Pr\left(\hist(\vec W_{\ind_0(\vec z)}) \in \polynint|_{\vec h_1} \;\middle\vert\; \vec Z = \vec z\right) \\
=&  \Pr(\vec Z = (0,1,1,0,1)) \sum_{\vec h_1\in {\mathbb Z}_{\ge 0}^{q-o-1}}\Pr\left(\hist(\vec W_{\{2,3,5\}}) =\vec h_1 \;\middle\vert\; \vec Z = (0,1,1,0,1)\right)\\
&\hspace{20mm} \times \Pr\left(\hist(\vec W_{\{1,4\}}) \in \polynint|_{\vec h_1} \;\middle\vert\; \vec Z = (0,1,1,0,1)\right)
\end{align*}
As a more concrete example,  let $\vec h_1 = (1,1,1,0)$ as in Example~\ref{ex:h}, we summand in the equation above becomes the following:
\begin{align*} 
&  \Pr\left(\hist(\vec W_{\{2,3,5\}}) = (1,1,1,0) \;\middle\vert\; \vec Z = (0,1,1,0,1)\right) \times \Pr\left(\hist(\vec W_{\{1,4\}}) \in \{(2,0)\} \;\middle\vert\; \vec Z = (0,1,1,0,1)\right)
\end{align*}
In words, it is the product of the following two terms:
\begin{itemize}
\item [(1)] the probability for the histogram of $\vec W_{\{2,3,5\}}$ to be $(1,1,1,0)$. I.e.,  the second, third, and fifth agents' votes are $\{2\succ 1\succ 2,2\succ 3\succ 1,3\succ 1\succ 2\}$ in any order, and 
\item [(2)] the probability for the histogram of $\vec W_{\{1,4\}}$ to be $(2,0)$. I.e., both the first and the fourth agents vote for $1\succ 2\succ 3$.
\end{itemize}
We emphasize that  this example only illustrates (\ref{eq:histz}) for $\vec z = (0,1,1,0,1)$ and $\vec h_1 = (1,1,1,0)$.  (\ref{eq:histz}) requires summing over other combinations of $\vec z$ and $\vec h_1$.
\end{ex}

To upper-bound (\ref{eq:histz}), we will show that  given $|\ind_0(\vec z)| \ge 0.9 \epsilon n$, for any $\vec h_1 \in {\mathbb Z}_{\ge 0}^{q-o-1}$,
\begin{equation}
\label{eq:histo}
\Pr\left(\hist(\vec W_{\ind_0(\vec z)}) \in \polynint|_{\vec h_1} \;\middle\vert\; \vec Z = \vec z\right)=O((0.9 \epsilon n)^{-\frac{o}{2}})= O(n^{-\frac{o}{2}})
\end{equation}
(\ref{eq:histo}) follows after combining the following two parts. 
\begin{itemize}
\item {\bf Part 1: Claim~\ref{claim:bounded} below,} which states that  the number of integer vectors in $\polynint|_{\vec h_1}$ is upper bounded by a constant that only depends on $\poly$, which means that it does not depend on $n$.
\begin{claim}\label{claim:bounded} There exists a constant $C_{\poly}>0$ such that for each $\vec h_1\in {\mathbb Z_{\ge 0}^{q-o-1}}$, $|\polynint|_{\vec h_1}|\le C_{\poly}$.
\end{claim}
\begin{proof}
We first prove an observation, which states that there exists a constant $C^*$ that only depends on $\poly$, such that for any $\vec x\in \polynint$, we have $|\invert{\vec x_{I_0}}- \bd\cdot \invert{\vec x_{I_1},n}|_\infty<C^*$.

According to the V-representation of $\poly$, we can write $\vec x = \vec v + \vec x'$, where $\vec v\in\mV$  is bounded, and $\vec x'\in \polyz$ which means that $\invert{\vec x'_{I_0}}- \bd\cdot \invert{\vec x'_{I_1}, \vec x'\cdot\vec 1} = \invert{\vec 0}$. Recall that  $C_{\max}$ is the maximum $L_\infty$ norm of  vectors in $\mV$. It follows that $|\vec v|_\infty\le C_{\max}$, which means that 
$$|\vec v\cdot \vec 1| = \left|\vec x\cdot\vec 1-\vec x'\cdot\vec1\right| = \left|n-\vec x'\cdot\vec1\right|\le qC_{\max}$$ 

 Let $C'$ denote the maximum absolute value of entries in $\bd$ and let $C^*=C_{\max} + 2qC_{\max}C'$. We have:
\begin{align*}
& |\invert{\vec x_{I_0}}- \bd\cdot \invert{\vec x_{I_1},n}|_\infty =  |\invert{\vec v_{I_0}+\vec x'_{I_0}}- \bd\cdot \invert{\vec v_{I_1}+\vec x'_{I_1},n}|_\infty \\
=& |\invert{\vec v_{I_0}}- \bd\cdot \invert{\vec v_{I_1},n-\vec x'\cdot \vec 1}|_\infty  \le C_{\max} + 2qC_{\max}C' = C^*
\end{align*}

For any $\vec h_0\in \polynint|_{\vec h_1}$, because $ (\vec h_0,\vec h_1)\in \polynint$, according to the observation above, we have 
$$|\invert{\vec h_0}- \bd\cdot \invert{\vec h_1,n}|_\infty<C^*$$
Therefore, $\polynint|_{\vec h_1}$ is contained in an $|I_0|$-dimensional cube whose edge length is $2C^*$ and is centered at $\bd\cdot \invert{\vec h_1,n}$. It is not hard to verify that the cube contains no more than $(2C^*+1)^q$ integer points. This proves claim by letting $C_{\poly} = (2C^*+1)^q$. \end{proof}

\item {\bf Part 2: The point-wise anti-concentration bound~\cite[Lemma 3 in the Appendix]{Xia2020:The-Smoothed}.} For completeness, we recall the lemma in our notation below. 

\paragraph{\bf \boldmath Lemma$'$ (\cite[Lemma~3 in the Appendix]{Xia2020:The-Smoothed}). }{\em Given $q^*\in\mathbb N$ and $\epsilon>0$, there exists a constant $C^*>0$ such that for any $n^*\in \mathbb N$ and strictly positive (by $\epsilon$) vector $\vec \pi^*$ of $n^*$ distributions over $[q^*]$,   and any vector $\vec x^* \in {\mathbb Z^{q^*}_{\ge 0}}$, we have $\Pr(\vec X_{\vec \pi^*}=\vec x^*)<C^*(n^*)^{\frac{1-q^*}{2}}$. 
}

We note that the constant in $O(n^{-\frac{o}{2}})$ in (\ref{eq:histo}) only depends on $\poly$ (therefore $q$) and $\epsilon$ but not on $\Pi$ or $n$.
\end{itemize}
Then, (\ref{eq:histo}) follows after applying  Lemma$'$ to (constantly many) vectors in $\polynint|_{\vec h_1}$ (guaranteed by Claim~\ref{claim:bounded}),  by letting $q^* = |I_0|=o+1$ and $n^* = |\ind_0(\vec z)|$. The next example illustrates the application of Lemma$'$ in the setting of Example~\ref{ex:sum}.
\begin{ex}[\bf \boldmath Running example: Equation (\ref{eq:histo})]
\label{ex:running10}
Continuing Example~\ref{ex:sum}, recall that in this running example $\vec z = (0,1,1,0,1)$ (Example~\ref{ex:z}) and $\vec h_1 = (1,1,1,0)$ (Example~\ref{ex:h}). Therefore, (\ref{eq:histo}) becomes:
$$\Pr\left(\hist(\vec W_{\{1,4\}}) \in \{(2,0)\} \;\middle\vert\; \vec Z = (0,1,1,0,1)\right)$$
Then, we let $q^* = |I_0| = 2$ and $n^* = 2$ in Lemma$'$ and apply it to $\vec x^* = (2,0)$, which is the only vector in  $\polynint|_{\vec h_1}$.
\end{ex} 


Back to (\ref{eq:histz}), we now upper-bound the $ \Pr(|\ind_0(\vec z)| < 0.9\epsilon n)$ part in (\ref{eq:histz}). Because random variables in $\vec Y$ are above $\epsilon$, for all $j\le n$, $Z_j$ takes $0$ with probability at least $\epsilon$. Therefore, $\expect(|\ind_0(\vec Z)|) \ge \epsilon n$. By Hoeffding's inequality, $\Pr(|\ind_0(\vec Z)| < 0.9\epsilon n)$ is exponentially small in $n$, which is $O(n^{-\frac{o}{2}})$ when $n$ is sufficiently large. 

Putting all together,  we have:
\begin{align*}
&\Pr(\vXp\in \poly) \le \sum_{\vec z\in \{0,1\}^n: |\ind_0(\vec z)|\ge 0.9\epsilon n}  \Pr(\vec Z = \vec z)\sum_{\vec h_1\in {\mathbb Z}_{\ge 0}^{q-o-1}}\Pr\left(\hist(\vec W_{\ind_1(\vec z)}) =\vec h_1 \;\middle\vert\; [\vec Z]_{\ind_1(\vec z)} = \vec 1\right)\\
&\hspace{40mm}\times  O(n^{-\frac{o}{2}}) + O(n^{-\frac{o}{2}}) = O(n^{-\frac{o}{2}}) 
\end{align*}
This proves the polynomial upper bound when  $\polynint\ne\emptyset$  and $\polyz\cap \conv(\Pi)\ne\emptyset$, where the constant in $O(n^{-\frac{o}{2}})$ depends on $\poly$ and $\epsilon$ (but not on $\Pi$ or $n$).

\paragraph{\bf Proof of the polynomial lower bound on Sup.}  
The proof is done in the following five steps. In {\bf Step 1}, for any $n\in \mathbb N$ that is sufficiently large, we  define  a non-negative integer vector $\vec y^n\in \polynint$  and its neighborhood $B^n\subseteq \polyn$  such that vectors in $B^n$ are $O(\sqrt n)$ away from $\cone(\Pi)$.  In {\bf Step 2}, we  define a vector $\vec\pi\in \Pi^{n}$ chosen by the adversary to achieve the lower bound. In {\bf Step 3}, we prove that $B^n$ contains $\Theta\left(n^{\frac{\dim(\polyz)-1}{2}}\right)$ many non-negative integer vectors. In {\bf Step 4}, we show that when preferences are generated according to $\vec\pi$, for any non-negative integer $\vec x\in B^n$, the probability for $\vXp$ to be $\vec x$ is $\Theta(n^{\frac{(1-q)}{2}})$. Finally, in {\bf Step 5} we show that the probability for $\vXp$ to be in $\poly$ is at least $\Omega\left(n^{\frac{\dim(\polyz)-1}{2}}\right)\times \Omega\left(n^{\frac{(1-q)}{2}}\right) = \Omega(n^{-\frac o 2})$.

\paragraph{\bf Step~1. For any sufficiently large $\bm{n\in\mathbb N}$, define a non-negative integer vector $\bm{\vec y^n\in \polynint}$ and its neighborhood $\bm{B^n\subseteq \polyn}$.} 
$\vec y^n$ will be defined  in {\bf Step 1.1} as an integer approximation to $\vec x^{\ell} = \ell\vec x^*+\sqrt\ell  \vec x^{@}+\vec x^{\#}$, whose components are defined as follows. See Figure~\ref{fig:ex-poly-lower} for an illustration.
\begin{itemize}
\item $\vec x^*\in \polyz\cap \conv(\Pi)$, which may not be integral. 
\item $\vec x^@$ is an inner point of $\polyz$, which means that $\ba^{=}\cdot \invert{\vec x^@} = \invert{\vec 0}$ and $\ba^{+}\cdot \invert{\vec x^@} < \invert{\vec 0}$. Note that $\vec x^@$ may not be integral, non-negative, in $\poly$, or in $\conv(\Pi)$.
\item $\vec x^{\#} \in \poly$ is an integer vector, which may not be in $\cone(\Pi)$. For example, $\vec x^{\#}$ can be any integer vector in $\poly_{n^{\#}}^{\mathbb Z}$ for the smallest $n^{\#}\in\mathbb N$ such that $\poly_{n^{\#}}^{\mathbb Z}\ne\emptyset$. Such $n^{\#}$ exists because otherwise the $0$ case of the theorem holds.
\item $\ell\in \mathbb R_{\ge 0}$ is a number that is used to guarantee  $\vec y^n\cdot \vec 1 =\vec x^{\ell}\cdot \vec 1 = n$.  That is, $\vec x^{\ell}\cdot\vec 1 = \ell +\sqrt \ell (\vec x^@\cdot\vec 1)+ \vec x^{\#}\cdot\vec 1= n$, or equivalently, because $\ell\ge 0$, we have $\ell = \left(-\frac{\vec x^@\cdot\vec 1}{2}+\sqrt{n-\vec x^{\#}\cdot\vec 1+(\frac{\vec x^@\cdot\vec 1}{2})^2}\right)^2$. We note that $\ell$ may not be an integer.
\end{itemize}
$\vec x^*$, $\vec x^@$, and $\vec x^{\#}$ are arbitrarily chosen but fixed throughout the proof. Let us illustrate  a choice of them in the following example.
\begin{ex}[\bf\boldmath Running example: $\vec x^*, \vec x^@, \vec x^{\#}$, and $\vec x^\ell$]
\label{ex:running11}
Continuing Example~\ref{ex:running10}, recall that $\Pi$ only contains the uniform distribution. Therefore,  $\vec x^*=(\frac16,\frac16,\frac16,\frac16,\frac16,\frac16)$, which is the only distribution in $  \polyz\cap \conv(\Pi)$. Let $\vec x^@ = (1.2,-0.2,1.2,-0.2,0,0)$ and $\vec x^{\#} = (2,0,1,1,1,0)$ (where $n^{\#}=\vec x^{\#}\cdot\vec 1=5$). Then,
$$\vec x^l = \ell\cdot \underbrace{(\frac16,\frac16,\frac16,\frac16,\frac16,\frac16)}_{\vec x^*}+\sqrt\ell  \cdot \underbrace{(1.2,-0.2,1.2,-0.2,0,0)}_{\vec x^@}+\underbrace{(2,0,1,1,1,0)}_{\vec x^{\#} }$$

When $n=100$,  $\ell +  2\sqrt \ell+ 5=100$, which means that $\ell = (\sqrt{96}-1)^2$. 
\end{ex}

\paragraph{\bf Step 1.1~Define a non-negative integer approximation $\bm{\vec y^n}$ to $\bm{\vec x^{\ell}}$.} 
First, we note that $\ell$ only depends on $\poly$ (because of the choices of $\vec x^@$ and $\vec x^{\#}$) and $n$ but not on $\Pi$ or $\epsilon$. Notice that while $\vec x^{\ell}\cdot\vec 1 = n$, it may contain negative components because $\vec x^@$ may contain negative components. Let  $n$ to be sufficiently large so that $\vec x^{\ell}\in \mathbb R_{\ge 0}^q$, which means that $\vec x^{\ell}\in \polyn$. This can be done because (i) each component of $\vec x^{\ell}$ is $\Theta(n)$, because $\vec x^{\ell}$ is largely determined by $\ell\vec x^*$ and $\ell = \Theta(n)$, and (ii) each component of $\vec x^*$ is at least $\epsilon$, because $\vec x^*\in \conv(\Pi)$.

Let $\ba _n = \left[\begin{array}{r}\ba\\ \vec 1\\ -\vec{1}\\ -{\mathbb I}_q\end{array}\right]$ and $\vec b _n = (\vec b, n, -n,\underbrace{0,\ldots,0}_q)$, where ${\mathbb I}_q$ is the $q\times q$ identity matrix. It follows that $\polyn$ is defined by $\ba_n$ and $\vec b _n$. In other words, 
$$\polyn = \{\vec x: \ba_n\cdot\invert{\vec x}\le \invert{\vec b_n}\}$$
Let $C_{\poly}^1 = q\Delta(\ba_n)$, where $\Delta(\ba_n)$ is the maximum absolute value among the determinants of all square submatrices of $\ba_n$. Notice that $C_{\poly}^1$ only depends on $\ba$ but not on $n$, $\epsilon$, or $\Pi$. It is not hard to verify that $\vec x^{\ell}\in \polyn$ and  recall that we have assumed $\polynint\ne\emptyset$ as a condition for the polynomial bound on the sup part of the theorem, which means that $\polyn$ contains an integer vector. Therefore, by \citep[Theorem~1(i)]{Cook86:Sensitivity} (where we let $\vec w = \vec 0$), there exists a (non-negative) integer vector $\vec y^n\in \polyn$ such that $|\vec y^n - \vec x^{\ell}|_\infty\le C_{\poly}^1$. 
Note that by definition  $\vec y^n\in \polynint$. For completeness, we recall \citep[Theorem~1(i)]{Cook86:Sensitivity} in our notation as follows.

\paragraph{\bf Theorem 1(i) in~\citep{Cook86:Sensitivity}.}{\em  Let $A$ be an $L\times q$ integer matrix, $\vec b$ be a $q$-dimensional vector, and let $\vec w$  be a $q$-dimensional vector such that $A\cdot \invert{\vec x}\le \invert{\vec b}$ has an integer solution and $\max\{\vec w\cdot \vec x: A\cdot \invert{\vec x}\le \invert{\vec b}\}$ exists. Then, for each optimal solution $\vec x_{\text{opt}}$ to $\max\{\vec w\cdot \vec x: A\cdot \invert{\vec x}\le \invert{\vec b}\}$, there exists an optimal solution $\vec z_{\text{opt}}\in \mathbb Z^q$ to $\max\left\{\vec w\cdot \vec x: A\cdot \invert{\vec x}\le \invert{\vec b}, \vec x\text{ is integral}\right\}$ with $|\vec x_{\text{opt}}-\vec z_{\text{opt}}|_\infty\le q\Delta(A)$.}

\paragraph{\bf \boldmath Step 1.2~Define $\vec y^n$'s neighborhood $B^n\subseteq \polyn$.} For any $\gamma>0$, we first define a neighborhood $B_\gamma\subseteq \mathbb R^q$ of $\vec x^@$ that consists of vectors $\vec x$ such that (i) $\ba^= \cdot \invert{\vec x} = (\vec 0)^\top$, (ii) $|\vec x_{I_1}-\vec x^@_{I_1}|_\infty\le \gamma$, and (iii) $\vec x\cdot \vec 1=\vec x^@\cdot \vec 1$. Formally,
$$B_\gamma \triangleq\{ \vec x: \forall i\in I_1, |x_i - x^@_i|\le \gamma  \text{ and } \invert{\vec x_{I_0}} = \bd\cdot \invert{\vec x_{I_1}, \vec x^@\cdot \vec 1}\}$$
Recall that $\vec x^@$ is an inner point of $\polyz$, which means that $\ba^= \cdot\invert{\vec x^@} = \invert{\vec 0}$, or equivalently, $\invert{\vec x_{I_0}^@} = \bd\cdot \invert{\vec x_{I_1}^@, \vec x^@\cdot \vec 1}$. Therefore, for any $\vec x \in B_\gamma$, we have 
$$\vec x_{I_0}  - \vec x^@_{I_0} = \bd\cdot \left[\invert{\vec x_{I_1} , \vec x \cdot \vec 1}-\invert{\vec x_{I_1}^@, \vec x^@\cdot \vec 1}\right] = \bd\cdot\invert{\vec x_{I_1} - \vec x_{I_1}^@,0}$$
Therefore, it is not hard to verify that  
\begin{equation}
\label{dfn:Bgamma-alt}
B_\gamma - \vec x^@ = \{ \vec \Delta\in {\mathbb R}^q: \forall i\in I_1, |\Delta_i|\le \gamma  \text{ and } \invert{\vec \Delta_{I_0}} = \bd\cdot \invert{\vec \Delta_{I_1}, 0}\}
\end{equation}
Recall that $\vec x^@$ is an inner point of $\polyz$, which means  that
$$\ba^= \cdot\invert{\vec x^@} = \invert{\vec 0}\text{ and }\ba^+ \cdot\invert{\vec x^@} < \invert{\vec 0}$$ 
Therefore, there exists  $\gamma>0$ such that for any $\vec x\in B_\gamma$, we have $\ba^= \cdot\invert{\vec x} = \invert{\vec 0}$ and $\ba^+ \cdot\invert{\vec x} < \invert{\vec 0}$.  In other words, all vectors in $B_\gamma$ are inner points in $\polyz$. It is not hard to verify that $\dim(B_\gamma) = q-o-1 = \dim(\polyz)-1$, where the $-1$ comes from the additional linear constraint $\vec x\cdot \vec 1 = \vec x^@\cdot \vec 1$. Notice that  $B_\gamma$ only depends on $\polyz$ but not on $n$, $\epsilon$, or $\Pi$. Let 
$$B^n \triangleq \vec y^n + \sqrt \ell (B_\gamma - \vec x^@) = \{\vec y^n + \sqrt \ell \vec \Delta:\vec \Delta\in B_\gamma-\vec x^@\}$$
Intuitively, $B^n$ is defined by first scaling up $(B_\gamma - \vec x^@)$ by $\sqrt\ell$, and then add it on top of $\vec y^n$. This means that at a high level $B^n$ consists of inner points of a local space that is similar to $\polyz$  and is centered at $\vec y^n$, plus an additional linear  constraint that requires $\vec x\cdot \vec 1 = n$. Note that $\vec x^\ell$ may not be in $B^n$. 

\begin{ex}[\bf \boldmath Running example: $B_\gamma$ and $B_n$]
\label{ex:running12}
Continuing Example~\ref{ex:running11}, let $\gamma=1$, we have
\begin{align*}
B_1 =\{ (-3 x_{213}-4 x_{231} &- x_{312}-3 x_{321} +4,  2 x_{213}+3 x_{231} +2 x_{321}-2, x_{213},x_{231},x_{312},x_{321}): \\
& x_{213}\in [0.2, 2.2], x_{231}\in [-1.2, 0.8], x_{312}\in [-1,1],x_{321}\in [-1,1]\}
\end{align*}
\begin{align*}
B_1 -\vec x^@=\{ (-3 \Delta_{213}-4 \Delta_{231} - \Delta_{312}-3 \Delta_{321},  2 \Delta_{213}+&3 \Delta_{231} +2 \Delta_{321}, \Delta_{213},\Delta_{231},\Delta_{312},\Delta_{321}): \\
& (\Delta_{213},\Delta_{231},\Delta_{312},\Delta_{321})\in [-1,1]^4\}
\end{align*}
Let $n=100$, we have $B^{100} = \vec y^n +  (\sqrt{96}-1)^2\cdot (B_1 -\vec x^@) $. 
\end{ex}

In the following three steps, we prove  $B^n\subseteq \polyn$ for any sufficiently large $n$.
\begin{itemize}
\item First, we prove $B^n\subseteq \poly$. Because $\vec y^n \in\polynint$, we have  $\ba^=\cdot \invert{\vec y^n}\le  \invert{\vec b^=}$, where $\vec b^=$ is the subvector of $\vec b$ that corresponds to $\ba^=$. For all $\vec \Delta\in B_\gamma - \vec x^@$, recall from (\ref{dfn:Bgamma-alt}) that $\vec \Delta_{I_0} = \bd\cdot \invert{\vec \Delta_{I_1}, 0}$, which means that $\ba^=\cdot \invert{\Delta} =  \invert{\vec 0}$. Therefore, 
$$\ba^=\cdot \invert{\vec y^n +\sqrt\ell \vec\Delta} = \ba^=\cdot \invert{\vec y^n} +\ba^=\cdot \invert{ \sqrt\ell \vec\Delta}= \ba^=\cdot \invert{\vec y^n}\le  \invert{\vec b^=}$$ 
Also notice that: 
$$\vec y^n +\sqrt\ell \vec\Delta = \vec y^n- \vec x^\ell +\vec x^\ell+\sqrt\ell \vec\Delta
=  (\vec y^n- \vec x^\ell) +  \ell \vec x^* + \sqrt\ell(\vec x^@+\vec \Delta) +\vec x^{\#}\notag$$
Therefore, we have the following bound on $\ba^+\cdot \invert{\vec y^n +\sqrt\ell \vec\Delta}$:
\begin{align}
 \ba^+\cdot \invert{\vec y^n +\sqrt\ell \vec\Delta} =& \ba^+\cdot \invert{ \vec y^n- \vec x^\ell} + \ba^+\cdot \invert{ \ell \vec x^*} + \ba^+\cdot \invert{\sqrt\ell(\vec x^@+\vec \Delta) }+ \ba^+\cdot \invert{\vec x^{\#}}\notag\\
\le &\invert{O(1)\cdot \vec{1}+\vec 0 - \Omega(\sqrt\ell)\cdot \vec{1}}+\invert{\vec b^+} \label{eq:Bn}
\end{align}
(\ref{eq:Bn}) follows after noticing that $|\vec y^n- \vec x^\ell|_\infty = O(1)$, $\vec x^*\in \polyz$, $\vec x^@+\vec\Delta$ is an inner point in $\polyz$ (which mean that $\ba^+\cdot \invert{\vec x^@+\vec \Delta }\le -\Omega(1)\cdot \vec 1$), and $\vec x^{\#}\in \poly$. Therefore, when $\ell$ is sufficiently large, we have:
$$\ba^+\cdot \invert{\vec y^n +\sqrt\ell \vec\Delta} \le \invert{\vec b^+}$$
This means that $\vec y^n +\sqrt\ell \vec\Delta\in \poly$ and therefore $B^n\in\poly$.  
\item Second, for any $\vec \Delta\in B_\gamma - \vec x^@$ we have $\vec \Delta\cdot \vec 1 = 0$. Therefore, $(\vec y^n +\sqrt\ell \vec\Delta)\cdot \vec 1 = n$.  
\item Third, we prove $B^n\subseteq {\mathbb R}_{\ge 0}^q$. Recall that $\vec y^n$ is $O(1)$ away from $\vec x^\ell$ (as constructed in Step 1.1). Each component of $\vec x^\ell$ is $\Omega(\ell)$, because each component of $\vec x^*$ is at least $\epsilon>0$. Also note that for any $\vec\Delta\in  B_\gamma - \vec x^@$, $|\vec\Delta|_\infty = O(1)$. Therefore, each component of each vector in $B^n$ is $\Omega(\ell) - O(\sqrt\ell)-O(1)$, which is strictly positive when $\ell$ is sufficiently large.
\end{itemize}

\paragraph{\bf \boldmath Step~2. Define $\vec \pi = (\pi_1,\ldots,\pi_n) \in \Pi^{n}$ s.t.~$\expect(\vXp) = \sum_{j=1}^n\pi_j $ is $O(\sqrt n)$ away from $B^n$ in $L_\infty$.} Because $\vec x^*\in \polyz\cap \conv(\Pi)$, by Carathéodory's theorem for convex/conic hulls (see e.g., \cite[p.~257]{Lovasz2009:Matching}), we can write  $\vec x^*$ as the convex combination of $1\le t\le q$ distributions in $\Pi$ regardless of the cardinality of $\Pi$, which can be infinity. Formally, let $\vec x^* = \sum_{i=1}^t \alpha_i \pi_i^*$, where for each $i\le t$, $\alpha_i>0$ and $\pi_i^*\in \Pi$, and $\sum_{i=1}^t \alpha_i = 1$.  We note that $\vec x^*\ge \epsilon\cdot \vec 1$, because $\Pi$ is strictly positive (by $\epsilon$).  

We now define $\vec\pi\in \Pi^n$ and then prove that $\vec \pi\cdot \vec 1$  is $O(\sqrt n)$ away from $\ell\vec x^* = \sum_{i=1}^t \ell\alpha_i \pi_i^*$ in $L_\infty$. Formally, for each $i\le t-1$, let $\vec \pi_i^\ell$ denote the vector of $\beta_i=\lfloor \ell \alpha_i\rfloor$ copies of $\pi_i^*$. Let $\vec \pi_k^\ell$ denote the vector of $\beta_t=n-\sum_{i=1}^{t-1} \beta_i$ copies of $\pi_t^*$. It follows that for any $i\le t-1$, $|\beta_i-\ell \alpha_i| \le 1$, and $|\beta_t-\ell \alpha_t|\le t + (\vec y^n\cdot \vec 1 - \ell \vec x^{*}\cdot \vec 1) = O(\sqrt\ell) = O(\sqrt n)$, where the constant in $O(\sqrt n)$ depends on $\poly$ (because $\vec y^n\cdot \vec 1 - \ell \vec x^{*}\cdot \vec 1$ depends on $\poly$) but not on $\Pi$ (because $t\le q$), $\epsilon$, or $n$. 

Let $\vec \pi = (\vec \pi_1^\ell,\ldots,\vec \pi_t^\ell)$, or equivalently
$$\vec \pi = (\underbrace{\pi_1^*,\ldots,\pi_1^*}_{\beta_1},\underbrace{\pi_2^*,\ldots,\pi_2^*}_{\beta_2},\ldots,  \underbrace{ \pi_t^*, \ldots, \pi_t^*}_{\beta_t})$$
It follows that $\expect(\vXp) = \vec\pi\cdot\vec 1 = \sum_{i=1}^t \beta_i\pi_i^*$, which means that  
$$|\expect(\vXp)-\ell\vec x^*|_\infty = \left|\sum_{i=1}^t (\beta_i-\ell\alpha_i)\pi_i^*\right|_\infty = O(\sqrt n)$$

Recall that $\vec x^\ell -\ell\vec x^* = \sqrt\ell \vec x^@+\vec x^{\#}$. This means that $\expect(\vXp) = \vec\pi\cdot\vec 1$  is also $O(\sqrt n)$ away from $\vec x^{\ell}$ in $L_\infty$.  Recall that any vector in $B^n$ is $O(\sqrt\ell)$ away from $\vec y^n$, which is $O(1)$ away from $\vec x^\ell$. Therefore, $\vec \pi\cdot \vec 1$  is $O(\sqrt n)$ away from any vector in $B^n$ in $L_\infty$, where the constant in the asymptotic bound depends on $\poly$ but not on $\Pi$, $\epsilon$, or $n$.


\paragraph{\bf\boldmath Step~3.~$B^n$ contains $\Omega\left({\sqrt{n}}^{\dim(\polyz)-1}\right)$-many integer vectors.} Intuitively, this is true because $B^n$ consists of enough vectors from an  neighborhood of $\vec y^n$ that looks like $\polyz$ with the additional linear constraint $\vec x\cdot \vec 1 = n$. Therefore, the $I_1$ components of vectors in  $B^n$ can be viewed as flexible variables, each of which can take any value in an $\Omega(\sqrt n)$ interval, which  contains  $\Omega(\sqrt n)$  integers. Once the $I_1$ component of a vector in $B^n$ is given, its $I_0$ components are more or less determined (see, e.g., Claim~\ref{claim:bounded} and its proof).
 
More precisely, we will enumerate $\Omega\left({\sqrt{n}}^{\dim(\polyz)-1}\right)$ many integer vectors in $B^n$ of the form 
$$\vec y^n + \left(\bd\cdot \invert{\vec \Delta_{I_1}, 0},\vec \Delta_{I_1}\right)\text{, where }\vec \Delta_{I_1}\in {\mathbb Z}_{\ge 0}^{I_1}$$
Recall that we have assumed w.l.o.g.~that $I_0= \{1,\ldots,o+1\}$ and $I_1=\{o+1,\ldots,q\}$. Let $\rho$ be the least common multiple of the denominators of entries in $\bd$. For example, $\rho=1$ in Example~\ref{ex:gaussian} because all entries in $\bd$ are integers. 

Then, we enumerate $\vec \Delta_{I_1} = (\Delta_{o+2},\ldots,\Delta_{q})\in {\mathbb Z}_{\ge 0}^{I_1}$ that satisfies the following two conditions.
\begin{itemize}
\item {\bf Condition~1.}  For each $o+2\le j\le q$,  $\rho$ divides $\Delta_{j}$, and 
\item {\bf Condition~2.}  For each $o+2\le j\le q$,  $|\Delta_{j}|< \frac{\gamma}{2}\sqrt n$.  Recall that $\gamma>0$ is the constant used to define $B_\gamma$. 
\end{itemize}
Condition~1 guarantees that $\vec \Delta_{I_0} = \bd\cdot \invert{\vec \Delta_{I_1}, 0}\in {\mathbb Z}^{I_0}$. Condition~2 guarantees that $\vec y^n +\vec \Delta \in B^n$ when $\frac{\gamma}{2}\sqrt n <\gamma\sqrt \ell$, which holds for any sufficiently large $n$ because $\lim_{n\ra\infty}\frac{n}{\ell}=1$. Notice that the total number of combinations of  $\vec \Delta_{I_1}$'s that satisfy both conditions is at least $\left(\lfloor \frac{\gamma\sqrt n}{\rho}\rfloor\right)^{|I_1|}$. When $n$ is sufficiently large so that $ \frac{\gamma\sqrt n}{\rho}>1$ and $\frac{\gamma}{2}\sqrt n <\gamma\sqrt \ell$, the number of integer vectors in $B^n$ is $\Omega\left({\sqrt{n}}^{|I_1|}\right) = \Omega\left({\sqrt{n}}^{\dim(\polyz)-1}\right)$, where the constant in the asymptotic lower bound depends on $\poly$  but not on $\Pi$, $\epsilon$, or $n$. This proves Step 3.

\paragraph{\bf \boldmath Step~4.~The probability for $\vXp$ to be any given vector in $\bm{B^n}$ is $\Omega\left(n^{\frac{(1-q)}{2}}\right)$.}  
This follows after Step~2 (${\vec\pi}\cdot \vec 1$ is $O(\sqrt n)$ away from $B^n$) and Lemma~\ref{lem:point-wise-concentration-PMV} below, which extends the point-wise concentration bound for i.i.d.~Poisson multinomial variables~\citep[Lemma~4 in the Appendix]{Xia2020:The-Smoothed} to general Poisson multinomial variables that correspond to strictly positive (but not necessarily identical) distributions.

\begin{lem}[\bf Point-wise concentration bound for Poisson multinomial variables]\label{lem:point-wise-concentration-PMV} For any $q\in \mathbb N$, any $\epsilon>0$, and any $\alpha>0$, there exists  $C_{q,\epsilon,\alpha}>0$  such that for any $n\in\mathbb N$, any $(n,q)$-Poisson multinomial random variable $\vXp$ where $\vec \pi$ is above $\epsilon$,  and any integer vector $\vec x \in {\mathbb Z^{q}_{\ge 0}}$  with $\vec x\cdot 1 = n$ and  $|\vec x - \expect(\vXp)|_\infty<\alpha\sqrt n$, we have:
$$\Pr(\vXp=\vec x)>C_{q,\epsilon,\alpha}\cdot n^{\frac{1-q}{2}}$$
\end{lem} 
\begin{proof} The proof proceeds in three steps. In {\bf Step (i)}, we prove that it suffices to prove the lemma for a special  $\vec \pi$, where at most $2^q$ types of distributions are used. This is achieved by analyzing the following linear program given $\epsilon>0$, $\vec x\in {\mathbb Z}_{\ge 0}^q$, and $\vec\mu\in {\mathbb R}_{\ge 0}^q$ are given, and the variables are $\vec \pi = (\pi_1,\ldots,\pi_n)$, where each $\pi_j$ is a distribution over $[q]$ that is above $\epsilon$.
\begin{equation}
\label{equ:optpmv}
\begin{split}
\min_{\vec\pi}&\Pr(\vXp = \vec x)\\
s.t.~& \sum_{j=1}^n \pi_j = \vec \mu\\
&\text{for each }j\le n, \pi_j\ge \epsilon\cdot \vec 1 \text{ and }\pi_j\cdot\vec 1 = 1
\end{split}
\end{equation}
At a high level, (\ref{equ:optpmv}) can be viewed as an extension of ideas and techniques for Poisson binomial variables developed by Hoeffding~\citep{Hoeffding1956:On-the-Distribution} to PMVs. In {\bf Step (ii)}, we prove the lemma for any $\vec \pi'$ that consists of a constant number of different distributes (each distribution may appear multiple times in $\vec \pi'$). This can be viewed as an extension of the point-wise concentration bound for i.i.d.~Poisson multinomial variables~\cite[Lemma~4 in the Appendix]{Xia2020:The-Smoothed} to PMVs of constant number of different distributions. In {\bf Step (iii)} we combine results in Step 1 and 2 to prove the lemma.

\paragraph{\bf Step (i).} For any $B\subseteq [q]$, let $\restrict{\vec \pi}{B}$ denote the collection (equivalently, subvector) of distributions in $\vec\pi$ whose $B$-components are exactly $\epsilon$. 
\begin{ex}\label{ex:pi-B}
For example, let  $q=5, \epsilon = 0.1, n=4$, $\vec \pi = (\pi_1,\pi_2,\pi_3,\pi_4)$, where 

\begin{center}
\begin{tabular}{|c|c|c|}
\hline $\pi_1$ & $\pi_1$& $\pi_3=\pi_4 =\piuni$\\
\hline  $(0.25,0.2,0.3,0.15, 0.1)$ & $(0.2,0.4,0.15,0.15, 0.1)$& $(0.2,0.2,0.2,0.2, 0.2)$\\
\hline
\end{tabular}
\end{center}
We have $\restrict{\vec \pi}{{\emptyset}} = (\pi_3,\pi_4)$, $\restrict{\vec \pi}{\{5\}} = (\pi_1,\pi_2)$, and for any other $B\subseteq [5]$,  $\restrict{\vec \pi}{B} = \emptyset$.
\end{ex}

By definition, $\vec \pi$ and $\bigcup_{B\subseteq [q]}\restrict{\vec \pi}{B}$ contain the same (multi-)set of distributions. We have the following claim about the optimal solutions to (\ref{equ:optpmv}).
\begin{claim}\label{claim:optpi}
(\ref{equ:optpmv}) has an optimal solution $\vec \pi^*$ where for each $B\subseteq [q]$ with $\restrict{\vec\pi^*}{B}\ne\emptyset$, all distributions in $\restrict{\vec\pi^*}{B}$ are the same.
\end{claim}
\begin{proof} Let $\Pi_\epsilon$ denote the set of all distributions over $[q]$ that are above  $\epsilon$. It follows that $\Pi_\epsilon$ is compact and $\Pr(\vXp = \vec x)$ is continuous, because it can be viewed as  a polynomial in $\vec \pi$. Therefore, due to the extreme value theorem, (\ref{equ:optpmv}) has solutions. Let $\vec \pi^*$ denote an arbitrary solution with the maximum total number of probabilities  in distributions in $\vec \pi^*$ that equal to $\epsilon$, that is,
$$\vec\pi^* =(\pi_1^*,\ldots, \pi_n^*) \in \arg\max\nolimits_{\vec \pi \text{ is a solution to } (\ref{equ:optpmv})}|\{ j\le n,i\le q: \pi_{j}(i) = \epsilon\}|$$
Suppose for the sake of contradiction the claim is not true. Then, there exists $B\subseteq [q]$ such that $\restrict{\vec\pi^*}{B}$ contains at least two different distributions. W.l.o.g.~let $B = \{q'+1,\ldots,q\}$ and let the two distributions be $\pi_1^*$ and $\pi_2^*$ such that for some $q^*\le q'$, we have that for each $1\le i\le q^*$, $\pi_1^*(i) \ne \pi_2^*(i)$ and for each $q^*+1\le i\le q'$, $\pi_1^*(i) = \pi_2^*(i)$. It follows that $q^*\ge 2$.

Let $\vec \pi = (\pi_1,\ldots, \pi_n)\in \Pi_\epsilon^n$ denote an arbitrary vector of  $n$ distributions in $\Pi_\epsilon$ (i.e., $\vec \pi$ may not be a solution to (\ref{equ:optpmv})) such that  for  every $1\le i\le q^*$, $\pi_1(i) \ne \pi_2(i)$ and for each $q^*+1\le i\le q'$, $\pi_1(i) = \pi_2(i)$.  For any $\vec \psi = (\psi_1,\ldots, \psi_{q^*})$ such that $\vec \psi\cdot\vec 1 =  0$ and $|\vec\psi|_\infty$ is sufficiently small, we let $\vec\pi_{\vec\psi}$ denote the vector of distributions that is obtained from $\vec \pi$ by replacing $\pi_1$ by $\pi_1+ (\vec \psi, \vec 0)$ and replacing $\pi_2$ by $\pi_2- (\vec \psi, \vec 0)$. 

\begin{ex}\label{ex:pi-psi}
Continuing Example~\ref{ex:pi-B}, we let $B=\{5\}$. Then, $q' = 4$, $q^*=3$, $\vec\psi = (\psi_1,\psi_2,\psi_3)$ and
$$\vec\pi_{\vec\psi} = ((0.25+\psi_1,0.2+\psi_2,0.3+\psi_3,0.15, 0.1), (0.2-\psi_1,0.4-\psi_2,0.15-\psi_3,0.15, 0.1), \pi_3,\pi_4)$$
\end{ex}

For any $\vec x\in {\mathbb Z_{\ge 0}^q}$ with $\vec x\cdot\vec 1 = n$ and any $\vec \pi$, we will prove that $\Pr\left(\vec X_{\vec\pi_{\vec\psi}} = \vec x\right)-\Pr\left(\vXp = \vec x\right)$ can be calculated as follows.
\begin{align}
&\Pr\left(\vec X_{\vec\pi_{\vec\psi}} = \vec x\right)-\Pr\left(\vXp = \vec x\right)\notag\\
=&\sum_{i\le q^*} F_{ii} (-\psi_i^2 + \psi_i(\pi_2(i) - \pi_1(i)))+ \sum_{1\le i<t\le q^*} F_{it} (-2\psi_i\psi_t + \psi_i(\pi_2(t) - \pi_1(t))+ \psi_t(\pi_2(i) - \pi_1(i)))\label{eq:pfdiff}
\end{align}
where for any $1\le i\le t\le q^*$, $F_{it}$ is the probability for $\vec X_{ \{\pi_3,\ldots,\pi_n\}}$ to be the vector that is obtained from $\vec x$ by subtracting $1$ from the $i$-th element and the $t$-th element (and from $i$-th element twice if $i=t$). Note that  some $F_{it}$'s can be $0$. 

\begin{ex}\label{ex:diff}
Continuing Example~\ref{ex:pi-psi}, we let $\vec x = (2,1,1,0,0)$. Recall that $\pi_3$ and $\pi_4$ are uniform distributions. We have $F_{11} = \frac{1}{25}$, $F_{12} = F_{13} = F_{23} = \frac{2}{25}$, and $F_{22}=F_{33}=0$.
Then,  (\ref{eq:pfdiff}) becomes:
\begin{align*}
\Pr\left(\vec X_{\vec\pi_{\vec\psi}} =\vec x\right)-\Pr\left(\vXp = \vec x\right) 
&= \frac{1}{25}(-\psi_1^2-0.05\psi_1) + \frac{2}{25}(-2\psi_1\psi_2+0.2\psi_1-0.05\psi_2)\\\
&+ \frac{2}{25}(-2\psi_1\psi_3-0.15\psi_1-0.05\psi_3)+ \frac{2}{25}(-2\psi_2\psi_3-0.15\psi_2+0.2\psi_3)
\end{align*}
\end{ex}

Formally, (\ref{eq:pfdiff}) is proved by applying the law of total probability to the histogram of  $(\pi_3,\ldots,\pi_n)$. For any $1\le i< t\le q^*$,  let $\vec e_{i,t}\in \{0,1\}^q$ denote the vector that takes $1$ on the $i$-th component and the $t$-th component, and takes $0$ on other components. For any $1\le i\le q^*$, let  $\vec e_{ii}\in \{0,2\}^q$ denote the  vector that takes $2$ on the $i$-th component and takes $0$ on other components.  According to the law of total probability, we have:
\begin{align*}
\Pr\left(\vXp = \vec x\right) = & \sum_{1\le i\le t\le q^*} \Pr\left(\vec X_{(\pi_3,\ldots,\pi_n)} = \vec x - \vec e_{it}\right)\times \Pr\left(X_{(\pi_1,\pi_2)} = \vec e_{it}|\vec X_{(\pi_3,\ldots,\pi_n)} = \vec x - \vec e_{it}\right)\\
=&  \sum_{1\le i\le t\le q^*} F_{it}\times \Pr\left(X_{(\pi_1,\pi_2)} = \vec e_{it}|\vec X_{(\pi_3,\ldots,\pi_n)} = \vec x - \vec e_{it}\right)\\
=&  \sum_{1\le i\le t\le q^*} F_{it}\times \Pr\left(X_{(\pi_1,\pi_2)} = \vec e_{it}\right)\end{align*}
The last equation holds because the $n$ random variables are independent. A similar formula can be obtained for  $\Pr\left(\vec X_{\vec\pi_{\vec\psi}} = \vec x\right)$. That is, 
$$\Pr\left(\vec X_{\vec\pi_{\vec\psi}} = \vec x\right)  =  \sum_{1\le i\le t\le q^*} F_{it}\times \Pr\left(X_{(\pi_1+ (\vec \psi, \vec 0),\pi_2- (\vec \psi, \vec 0))} = \vec e_{it}\right)$$
Therefore,
\begin{equation}
\label{eq:pfdiff-total}
\Pr\left(\vec X_{\vec\pi_{\vec\psi}} = \vec x\right)  - \Pr\left(\vXp = \vec x\right) =   \sum_{1\le i\le t\le q^*} F_{it}\times \left[\Pr(X_{(\pi_1+ (\vec \psi, \vec 0),\pi_2- (\vec \psi, \vec 0))} = \vec e_{it}) - \Pr(X_{(\pi_1,\pi_2)} = \vec e_{it})\right]
\end{equation}

Next, we calculate $\Pr\left(X_{(\pi_1+ (\vec \psi, \vec 0),\pi_2- (\vec \psi, \vec 0))} = \vec e_{it}\right) - \Pr\left(X_{(\pi_1,\pi_2)} = \vec e_{it}\right)$ for $i=t$ and $i<t$, respectively.
\begin{itemize}
\item When $i=t$, we have:
\begin{align}
&\Pr\left(X_{(\pi_1+ (\vec \psi, \vec 0),\pi_2- (\vec \psi, \vec 0))} = \vec e_{ii}\right) - \Pr\left(\vec X_{(\pi_1,\pi_2)}\right)\notag\\ 
=& (\pi_1(i)+\psi_i)(\pi_2(i)-\psi_i) - \pi_1(i)\pi_2(i)\notag\\
 =& -\psi_i^2+\psi_i(\pi_2(i)-\pi_1(i))\label{eq:ii}
\end{align} 
\item When $1\le i<t\le q^*$, we have:
\begin{align}
&\Pr\left(X_{(\pi_1+ (\vec \psi, \vec 0),\pi_2- (\vec \psi, \vec 0))} = \vec e_{ii}\right) - \Pr\left(\vec X_{(\pi_1,\pi_2)}\right)\notag\\ 
=& \left[(\pi_1(i)+\psi_i)(\pi_2(t)-\psi_t) + (\pi_1(t)+\psi_t)(\pi_2(i)-\psi_i)\right] - \left[ \pi_1(i)\pi_2(t) + \pi_1(t)\pi_2(i)\right]\notag\\
=& -2\psi_i\psi_t + \psi_i(\pi_2(t) - \pi_1(t))+ \psi_t(\pi_2(i) - \pi_1(i)) \label{eq:it}
\end{align}
(\ref{eq:pfdiff}) follows after combining (\ref{eq:pfdiff-total}),  (\ref{eq:ii}), and (\ref{eq:it}).
\end{itemize}

For convenience, we rewrite (\ref{eq:pfdiff}) in matrix form. For any $1\le i<t\le q^*$, let $F_{ti} = F_{it}$ and let
$\bF = (F_{it})_{q^*\times q^*}$ denote the $q^*\times q^*$ symmetric matrix. $\vec\delta = (\pi_2(1) - \pi_1(1), \ldots, \pi_2(q^*) - \pi_1(q^*))$. According to the definition of $q^*$, no component of $\vec\delta $ is $0$. With the matrix notation, (\ref{eq:pfdiff}) becomes
\begin{align}
&\Pr\left(\vec X_{\vec\pi_{\vec\psi}} = \vec x\right)-\Pr\left(\vXp = \vec x\right)
= -\vec \psi \cdot \bF\cdot \invert{ \vec \psi} + \vec \psi \cdot \bF\cdot \invert{\vec\delta}\label{eq:pfdiffmatrix}
\end{align}

\begin{ex}\label{ex:F}
Continuing Example~\ref{ex:diff}, we have
$\bF = \left [ \begin{array}{ccc}
\frac{1}{25}&\frac{2}{25} &\frac{2}{25}  \\
\frac{2}{25}&0 &\frac{2}{25}  \\
\frac{2}{25}&\frac{2}{25} &0
 \end{array}\right]$
 and $\vec \delta = (-0.05,0.2,-0.15)$.
\end{ex}

Let $A$ denote the $(q^*-1)\times q^*$ matrix $\left [ \begin{array}{cccc}
1& & & -1\\
&\ddots & & \vdots\\
& & 1& -1
 \end{array}\right]$ and $\vec\psi'= (\psi_1,\ldots,\psi_{q^*-1})$. Recall that $\vec\psi \cdot\vec 1=0$, we have $\psi_{q^*} = -\psi_1-\cdots-\psi_{q^*-1}$, which means that $\vec\psi = \vec\psi' \cdot A$.  Because $\pi_1$ and $\pi_2$ are probability distributions, we have $\vec\delta \cdot\vec 1 = \vec 0$, which means that $\delta_{q^*} = -\delta_1-\cdots-\delta_{q^*-1}$. Therefore,  let $\vec\delta' = (\pi_2(1) - \pi_1(1), \ldots, \pi_2(q^*-1) - \pi_1(q^*-1))$, we have $\vec\delta = \vec\delta'\cdot A$.  Let $\bF' = A\cdot \bF\cdot \invert{A}$. Then, (\ref{eq:pfdiffmatrix}) becomes:
 \begin{align}
 \label{eq:pfdiffsimp}
\Pr\left(\vec X_{\vec\pi_{\vec\psi}} = \vec x\right)-\Pr\left(\vXp = \vec x\right)= -\vec\psi' \cdot \bF'\cdot \invert{\vec\psi'} + \vec\psi' \cdot \bF'\cdot \invert{\vec\delta'} 
 \end{align}
 
\begin{ex}
\label{ex:Fprime}
Continuing Example~\ref{ex:F}, we have  $A = \left [ \begin{array}{ccc}
1& 0& -1\\0& 1& -1
 \end{array}\right]$, $\bF' = \left [ \begin{array}{cc}
-\frac{3}{25}&-\frac{2}{25}   \\
-\frac{2}{25}&-\frac{4}{25}
 \end{array}\right]$
 and $\vec \delta' = (-0.05,0.2)$. Therefore, $\bF'\cdot \invert{\vec\delta'} =\left[\begin{array}{c}-\frac{1}{100}\\ -\frac{7}{250}\end{array} \right]$ and (\ref{eq:pfdiffsimp}) becomes:
$$\Pr\left(\vec X_{\vec\pi_{\vec\psi}} = \vec x\right)-\Pr\left(\vXp = \vec x\right)= -[\psi_1,\psi_2]\cdot \left [ \begin{array}{cc}
-\frac{3}{25}&-\frac{2}{25}   \\
-\frac{2}{25}&-\frac{4}{25}
 \end{array}\right] 
 \cdot \left [\begin{array}{c}\psi_1\\ \psi_2\end{array} \right]+ [\psi_1,\psi_2]\cdot \left [\begin{array}{c}-\frac{1}{100}\\ -\frac{7}{250}\end{array} \right] $$
\end{ex}
 
Notice that $\bF'$ depends on both $\vec x$ and $\vec\pi$. Next, we consider the case for $\vec \pi^*$, which we recall is an optimal solution to (\ref{equ:optpmv}) and $q^*$ is defined based on $\vec\pi^*$. We will prove that $\bF'\cdot \invert{\vec\delta'} = \invert{\vec 0}$, where $\bF'$ is the matrix corresponding to $\vec x$ and $\vec\pi^*$. Suppose for the sake of contradiction that this is not true. W.l.o.g.~suppose the first component of $\bF'\cdot \invert{\vec\delta'}$ is non-zero. Then, by letting $\psi_2=\cdots=\psi_{q^*-1}=0$, (\ref{eq:pfdiffsimp}) becomes $A\psi_1^2+B\psi_1$ for some constants $A$ and $B$ with $B\ne 0$, which means that  there exists $\psi_1\ne 0$ such that  (\ref{eq:pfdiffsimp}) is strictly less than zero.  This contradicts the assumption that $\vec \pi^*$ is an optimal solution to (\ref{equ:optpmv}). The following example shows how to choose $\vec\psi$ when $\bF'\cdot \invert{\vec\delta'} \ne \invert{\vec 0}$, to obtain another feasible solution with smaller objective value. Notice that in this example $\vec x$ is not an optimal solution to (\ref{equ:optpmv}).

\begin{ex} Continuing Example~\ref{ex:Fprime}, notice that the first component of $\bF'\cdot \invert{\vec\delta'} $ is non-zero. Therefore, by letting $\psi_2=0$, we have 
$$\Pr\left(\vec X_{\vec\pi_{\vec\psi}} = \vec x\right)-\Pr\left(\vXp = \vec x\right) = \frac{3}{25}\psi_1^2-\frac{1}{100}\psi_1$$ 
Let $\psi_1 = 0.05$, which means that $\psi_3 = -0.05$. It is not hard to verify that $\vec\pi_{\vec\psi}$ is a feasible solution to  (\ref{equ:optpmv}) with a smaller objective value.   
\end{ex}

Therefore, for any $\gamma\in \mathbb R$, if we let $\vec\psi =\gamma  \vec\delta$, then (\ref{eq:pfdiffsimp}) becomes zero, which means that if $\vec \pi_{\gamma  \vec\delta}^*$ is strictly positive (by $\epsilon$), then it is also an optimal solution to (\ref{equ:optpmv}).
Recall that all components of $\vec\delta$ are non-zero. Therefore, we can start from $\gamma = 0$ and gradually increase the value of $\gamma$  until any of the first $q^*$ components in $\pi_1^*+(\gamma  \vec\delta,\vec 0)$ or in $\pi_2^*-(\gamma  \vec\delta,\vec 0)$ becomes $\epsilon$. Then, it is not hard to verify that $\vec \pi_{\gamma  \vec\delta}^*$ is an optimal solution to (\ref{equ:optpmv}) with strictly more probabilities that equal to $\epsilon$. This contradicts the assumption that  $\vec \pi^*$ contains maximum number of probabilities that equal to $\epsilon$ among optimal solutions to (\ref{equ:optpmv}), and therefore concludes the proof of Claim~\ref{claim:optpi}.
\end{proof}

\paragraph{\bf Step (ii).} We prove the following special case of the lemma.
\begin{claim}\label{claim:anticonmultitype}
 For any $\epsilon>0,\alpha>0$, $q\in\mathbb N$, and $Q\in\mathbb N$, there exist constants $C_{q,\epsilon,\alpha,Q}$ and $N$ such that for any set $\Pi_Q$ of $Q$ distributions over $[q]$ that are strictly positive by $\epsilon$, any $\vec\pi'\in \Pi_Q^n$, and any integer vector $\vec x \in {\mathbb Z^{q}_{\ge 0}}$ with $\vec x\cdot 1 = n$ and $|\vec x - \expect(\vec X_{\vec \pi'})|_\infty<\alpha\sqrt n$, we have
$$\Pr\left(\vec X_{\vec \pi'}=\vec x\right)>C_{q,\epsilon,\alpha,Q}\cdot n^{\frac{1-q}{2}}$$
\end{claim}
\begin{proof} Because $|\Pi_Q|=Q$, there exists $\pi^*\in \Pi_Q$ that appears in $\vec \pi'$ for at least $\lceil\frac nQ\rceil$ times. Let $n' = \lceil \frac nQ\rceil$ and let $\vec X_1$ denote the $(n',q)$-PMV that corresponds to $(\underbrace{\pi^*,\ldots,\pi^*}_{n'})$. Let  $\vec X_2$ denote the $(n-n',q)$-PMV that corresponds to the remaining distributions in $\vec\pi'$.  Recall that each distribution in $\vec\pi'$ is  strictly positive by $\epsilon$. By Hoeffding's inequality, let  $\alpha' = \frac{(1-\epsilon)^2}{2}\log (4q)$, for each $i\le q$, we have
$$\Pr\left(|[\vec X_2]_i-[\expect(\vec X_2)]_i|>\alpha'\sqrt{n-n'}\right)\le \frac1{2q}$$
Therefore, by the union bound, we have:
$$\Pr\left(|\vec X_2-\expect(\vec X_2)|_\infty>\alpha'\sqrt{n-n'}\right)\le \frac12$$
Notice that $\vec X_1$ and $\vec X_2$ are independent. Now we can calculate $\Pr\left(\vec X_{\vec \pi'}=\vec x\right)$ by the law of total probability, by enumerating the target values for $\vec X_2$, denoted by $\vec x_2$, as follows.
\begin{align}
&\Pr\left(\vec X_{\vec \pi'}=\vec x\right)\\
= & \sum_{\vec x_2\in {\mathbb Z}_{\ge 0}^q: \vec x_2\cdot \vec 1 = n-n'}\Pr\left(\vec X_1 = \vec x-\vec x_2 | \vec X_2 = \vec x_2\right)\cdot \Pr\left(\vec X_2 = \vec x_2\right) \hspace{10mm} \text{\bf (The law of total probability)}\notag\\
= &\sum_{\vec x_2\in {\mathbb Z}_{\ge 0}^q: \vec x_2\cdot \vec 1 = n-n'}\Pr\left(\vec X_1 = \vec x-\vec x_2\right)\cdot \Pr\left(\vec X_2 = \vec x_2\right)\hspace{10mm} \text{\bf\boldmath ($\vec X_1$ and $\vec X_2$ are independent)}\notag\\
\ge & \sum_{\vec x_2\in {\mathbb Z}_{\ge 0}^q: \vec x_2\cdot \vec 1 = n-n' \text{ and }|\vec x_2 - \expect(\vec X_2)|_\infty\le \alpha'\sqrt{n-n'}}\Pr\left(\vec X_1 = \vec x-\vec x_2\right)\cdot \Pr\left(\vec X_2 = \vec x_2\right)\notag\\
=&  \sum_{\vec x_2\in {\mathbb Z}_{\ge 0}^q: \vec x_2\cdot \vec 1 = n-n' \text{ and }|\vec x_2 - \expect(\vec X_2)|_\infty\le \alpha'\sqrt{n-n'}}\Omega\left(n^{\frac{1-q}{2}}\right)\cdot \Pr\left(\vec X_2 = \vec x_2\right)\label{eq:xoneconcentration}\\
=& \Omega\left(n^{\frac{1-q}{2}}\right)\cdot \Pr\left(|\vec X_2 - \expect(\vec X_2)|_\infty\le \alpha'\sqrt{n-n'}\right) = \Omega\left(n^{\frac{1-q}{2}}\right)\notag
\end{align}
(\ref{eq:xoneconcentration}) follows after applying the point-wise concentration bound for i.i.d.~Poisson multinomial  variables~\citep[Lemma~4 in the Appendix]{Xia2020:The-Smoothed}  to $\vec X_1$ in the following way.  Recall that $\expect(\vXp)= \expect(\vec X_1)+\expect(\vec X_2)$, $|\vec x - \expect(\vXp)|_\infty <\alpha \sqrt n$, and $|\vec x_2 - \expect(\vec X_2)|_\infty\le \alpha'\sqrt{n-n'}$. We have: 
$$| \vec x_1-\expect(\vec X_1)|_\infty = |\vec x - \expect(\vec X_1) - (\vec x_2-\expect(\vec X_2))|_\infty  =   \alpha \sqrt n + \le \alpha'\sqrt{n-n'}\le (\alpha Q+\alpha'(Q-1))\sqrt {n'}$$
This completes the proof of Claim~\ref{claim:anticonmultitype}.
\end{proof}

\paragraph{\bf Step (iii).} By Claim~\ref{claim:optpi}, there exists a vector $\vec\pi^*$ of $n$ distributions over $[q]$, each of which is above $\epsilon$, such that (i) $\Pr(\vXp = \vec x) \ge \Pr(\vec X_{\vec \pi^*} = \vec x)$ and $\expect(\vXp) = \expect(\vec X_{\vec \pi^*})$, and (ii) $\vec\pi^*$ consists of no more than $Q=2^q$ different distributions that correspond to different subsets of $[q]$. Note that $\vec \pi^*$ may not be in $\Pi^n$. Lemma~\ref{lem:point-wise-concentration-PMV} follows after applying Claim~\ref{claim:anticonmultitype} with $\vec \pi' = \vec \pi^*$.\end{proof}

\noindent {\bf Step~5. Final calculations.} Finally, combining Step~3 and 4, we  have: 
\begin{align*}
&\Pr\left(\vXp\in \poly\right)\ge   \Pr\left(\vXp\in B^n\right) \ge  |B^n|\times \min_{\vec x\in B^n} \Pr\left(\vXp=\vec x\right)\\
 \ge &
\Omega\left(n^{\frac{\dim(\polyz)-1}{2}}\right)\times \Omega\left(n^{\frac{(1-q)}{2}}\right) = \Omega\left(n^{-\frac o 2}\right),\end{align*}
where the constants in the asymptotic bounds depend on $\poly$ (and therefore $q$) and $\epsilon$ but not on other parts of $\Pi$ or $n$. We require $n$ to be sufficiently large to guarantee the existence of $\vec x^{*}$ and $\vec x^{@}$, and $\frac{\gamma\sqrt n}{\rho}>1$ and $\frac{\gamma}{2}\sqrt n <\gamma\sqrt \ell$. 

This proves the polynomial lower bound on Sup.

\paragraph{\bf Proof of the exponential bounds on Inf.}  The proof  is similar to the proof of the exponential case in Theorem~\ref{thm:maintechh}. The lower bound is straightforward and the upper bound is proved by choosing an arbitrary $\pi \in \conv(\Pi)$ such that $\pi\not\in\polyz$, and define $\vec\pi$ in a similar way as in the proof of the polynomial lower bound of the Sup part.

\paragraph{\bf Proof of the polynomial upper bound on Inf.} Notice that the condition for this case implies that the polynomial case for Sup holds, which means that there exists $\vec\pi\in \Pi^n$ such that $\Pr(\vXp\in\poly) = \Theta\left(n^{\frac{\dim(\polyz)-q}{2}}\right)$. This immediately implies an $O\left(n^{\frac{\dim(\polyz)-q}{2}}\right)$ upper bound on Inf.

\paragraph{\bf Proof of the polynomial low bound on Inf.} The  proof is similar to the proof of the polynomial lower bound on the Sup part in Theorem~\ref{thm:maintechh}, except that when defining $\vec y^n$, we let $\vec x^* = \sum_{j=1}^n\pi_j/n\in \conv(\Pi)\cap \polyz$. More precisely, we have the following lemma that will be used in the proof of other propositions in this paper. Notice that the constant $C_{\Pi,\poly}$ in the lemma only depends on $\Pi$ (therefore $q$ and $\epsilon$) and $\poly$ but not on $\vec\pi$.
\begin{lem}
\label{lem:lower-any-pi}
For any $q\in\mathbb N$, any closed and strictly positive $\Pi$ over $[q]$, and any polyhedron $\poly$ with integer matrix $\ba$, there exists  $C_{\Pi,\poly}>0$ such that for any $n\in\mathbb N$ with $\polynint\ne\emptyset$ and any $\vec\pi=(\pi_1,\ldots,\pi_n)\in \Pi^n$ with $\sum_{j=1}^n\pi_j\in \polyz$, 
$$\Pr\left(\vXp\in\poly\right) \ge C_{\Pi,\poly}\cdot n^{-\frac{\dim(\polyz)}{2}}$$
\end{lem}
\begin{proof}  For any $n\in \mathbb N$ that is sufficiently large, we define  a non-negative integer vector $\vec y^n\in \polynint$  and  a neighborhood $B^n\subseteq \polyn$ of $\vec y^n$ that is similar to Step 1 of proof for the polynomial lower bound in the Sup case of Theorem~\ref{thm:maintechh}. The only difference is that we let $\vec x^* = \sum_{j=1}^n\pi_j/n = \expect(\vXp)/n\in \conv(\Pi)\cap \polyz$. The rest of the proof is the same as Step 3, 4, and 5 of the proof for the polynomial lower bound in the Sup case of Theorem~\ref{thm:maintechh}.

To see that the constant $C_{\Pi,\poly}$ does not depend on $\vec \pi$, we first notice that there exists a constant $\alpha_{\Pi,\poly}$ such that $\expect(\vXp) = n\vec x^*$ is no more than $\alpha_{\Pi,\poly}(\sqrt n)$ away from $\vec x^\ell = \ell\vec x^* + \sqrt \ell \vec x^{@}+\vec x^{\#}$, which is $O(1)$ away from  $\vec y^n$. Recall that $\vec y^n$ is the ``center'' of  $B^n$, whose ``radius'' only depends on $\ell$ and $B_\gamma$ but not on $\vec x^*$. Therefore, $\vec y^n$ (which is $O(1)$ close to $\vec x^\ell$ and the constant only depends on $\ba$ but not on $\vec x^*$) is no more than $\alpha_{\Pi,\poly}'\sqrt n$ away from any vector in $B^n$ for some $\alpha_{\Pi,\poly}'>0$. Finally, we note that the constant in the point-wise concentration bound (Lemma~\ref{lem:point-wise-concentration-PMV}) does not depend on $\vec\pi$.
\end{proof}

This completes the proof of Theorem~\ref{thm:maintechh}.
\end{proof}

\subsection{Proof of Theorem~\ref{thm:union-poly}}
\label{sec:maintech-union-proof}

\appThm{thm:union-poly}{\boldmath Smoothed Likelihood of PMV-in-$\upoly$}{Given any $q, I\in\mathbb N$, any closed and strictly positive $\Pi$ over $[q]$, and any $\upoly = \bigcup_{i\in I}\cpoly{i}$  characterized by integer matrices, for any $n\in\mathbb N$, 
\begin{align*}
&\sup_{\vec\pi\in\Pi^n}\Pr\left(\vXp \in \upoly\right)=\left\{\begin{array}{ll}0 &\text{if } \alpha_n = -\infty\\
\exp(-\Theta(n)) &\text{if } -\infty<\alpha_n<0\\
\Theta\left(n^{\frac{\alpha_n-q}{2}}\right) &\text{otherwise (i.e. } \alpha_n>0\text{)}
\end{array}\right.,\\
&\inf_{\vec\pi\in\Pi^n}\Pr\left(\vXp \in \upoly\right)=\left\{\begin{array}{ll}0 &\text{if } \beta_n = -\infty\\
\exp(-\Theta(n)) &\text{if } -\infty<\beta_n<0\\
\Theta\left(n^{\frac{\beta_n-q}{2}}\right) &\text{otherwise (i.e. } \beta_n>0\text{)}\end{array}\right..
\end{align*}}
\begin{proof} For convenience, we first recall Inequality (\ref{eq:upoly-high-level}) in the main text. 
$$
\max\nolimits_{i\le I}\Pr\left(\vXp \in \cpoly{i}\right)\le \Pr\left(\vXp \in \upoly\right)\le \sum\nolimits_{i\le I}\Pr\left(\vXp \in \cpoly{i}\right) \text{\ \ \ \ \ (\ref{eq:upoly-high-level})}
$$
The theorem is proved by combining (\ref{eq:upoly-high-level}) and applications of Theorem~\ref{thm:maintechh} to $\Pi$ and $\cpoly{i}$. We first introduce some notation.  Let  $\upolynint$ denote the set of all non-negative integer vectors in $\upoly$ whose $L_1$ norm is $n$.  That is,
$$\upolynint = \bigcup\nolimits_{i\le I} \cpolynint{i}$$
Given a distribution $\pi$ over $[q]$, let  $\calI_{\upoly,n}^\pi\subseteq\{1,\ldots, I\}$ denote the set of indices  $i$ in $\cpoly{i}$ such that the weight on $(\pi,\cpoly{i})$ in the activation graph $\calG_{\Pi,\upoly,n}$ is positive at $n$. Equivalently, $\calI_{\upoly,n}^\pi$ can be defined as follows.
$$\calI_{\upoly,n}^\pi = \{i\le I: \cpolynint{i}\ne\emptyset \text{ and } \pi\in\cpolyz{i}\}$$

For example,  $\calI_{\upoly,n}^\pi = \{2\}$ in Figure~\ref{fig:upoly-graph}. 

Next, we let   $\Pi_{\upoly,n}\subseteq \conv(\Pi)$ denote the  distributions $\pi$ in $\conv(\Pi)$ such that $\calI_{\upoly,n}^\pi \ne\emptyset$, or equivalently, $\md{\upoly}{\pi} >0$. Namely, 
$$\Pi_{\upoly,n} = \{\pi\in\conv(\Pi): \calI_{\upoly,n}^\pi\ne\emptyset\}$$

We start with the proof for the $\sup$ part.

\paragraph{\bf \boldmath Proof for the $\sup$ part.} We  prove the $\sup$ part by discussing the three cases ($0$, exponential, and polynomial)  as follows. 
\begin{itemize}
\item {\bf  \boldmath The $0$ case of sup.} This case is straightforward because $\alpha_n=-\infty$ means that no polyhedron in $\upoly$ is active at $n$, which means that  $\upoly$ does not contain any non-negative integer vector whose size is $n$, while any outcome of the PMV $\vXp$ is a non-negative integer vector whose size is $n$. 
\item {\bf The exponential case of sup.} We recall that in this case $\upolynint \ne \emptyset$ and $\Pi_{\upoly,n}=\emptyset$. To prove the {\bf exponential lower bound of sup}, it suffices to prove that  there exists $\vec\pi\in\Pi^n$ such that $\Pr\left(\vXp \in \upoly\right) = \exp(-O(n))$. Because $\upolynint = \bigcup_{i\le I} \cpolynint{i}\ne \emptyset$, there exists $i^*\le I$ such that $\cpoly{i^*}$ is active at $n$, i.e., $\cpolynint{i^*}\ne \emptyset$. The exponential lower bound follows after applying the lower bound of (\ref{eq:upoly-high-level}) to an arbitrary $\vec\pi\in\Pi^n$: 
$$\Pr\left(\vXp \in \upoly\right) \ge \max\nolimits_{i\le I}\Pr\left(\vXp \in \cpoly{i}\right) \ge \Pr\left(\vXp \in \cpoly{i^*}\right) \ge \epsilon^{n}=\exp(-O(n))$$
To prove the {\bf exponential upper bound of sup}, it suffices to prove that for every $\vec\pi\in\Pi^n$, we have $\Pr\left(\vXp \in \upoly\right) = \exp(-\Omega(n))$. This will be proved by combining the upper bound in (\ref{eq:upoly-high-level}) and the $0$ or exponential upper bounds on the sup part of Theorem~\ref{thm:maintechh} when applied to $\Pi$ and each $\cpoly{i}$. 

More precisely, because $\Pi_{\upoly,n}=\emptyset$, for every $\pi\in\Pi^n$, we have $\calI_{\upoly,n}^\pi=\emptyset$. This mean that for every $\pi\in\conv(\Pi)$ and every $i\le I$, either $\cpolynint{i}=\emptyset$ or $\pi\not\in\cpolyz{i}$ (or both hold). If the former holds, then the $0$ case of Theorem~\ref{thm:maintechh} can be applied to $\Pi$ and $\cpoly{i}$, which means that for any $\vec\pi\in\Pi^n$, we have $\Pr\left(\vXp \in \cpoly{i}\right)=0$. If the former does not hold and the latter holds, then the exponential case of Theorem~\ref{thm:maintechh} applies to $\Pi$ and $\cpoly{i}$, which means that for every $\vec\pi\in\Pi^n$, we have $\Pr\left(\vXp \in \cpoly{i}\right) = \exp(-\Omega(n))$. Notice that $I$ is a constant. Therefore, following the upper bound of (\ref{eq:upoly-high-level}), for every $\vec\pi\in\Pi^n$, we have:
$$\Pr\left(\vXp \in \upoly\right) \le \sum\nolimits_{i\le I}\Pr\left(\vXp \in \cpoly{i}\right) \le I \cdot\exp(-\Omega(n))=\exp(-\Omega(n))$$

\item {\bf The polynomial bounds of sup.} We recall that in this case $\upolynint \ne \emptyset$ and $\Pi_{\upoly,n}\ne \emptyset$. Like the exponential case, the proof is done by combining (\ref{eq:upoly-high-level}) and the applications of Theorem~\ref{thm:maintechh} to $\Pi$ and each $\cpoly{i}$. 
To prove the {\bf polynomial lower bound of sup}, it suffices to prove that  there exists $\vec\pi\in\Pi^n$ such that $\Pr\left(\vXp \in \upoly\right) = \Omega(n^{\frac{\alpha_n-q}{2}})$. Because $\Pi_{\upoly,n}\ne \emptyset$, we let $(\pi^*,i^*)$ denote the pair that achieves $\alpha_n$. More precisely, let
$$(\pi^*,i^*)=\arg \max\nolimits_{\pi\in\Pi_{\upoly,n}}\max\nolimits_{i\in \calI_{\upoly,n}^\pi}\dim(\cpolyz{i})$$
Following the definition of $(\pi^*,i^*)$, we have $\cpolynint{i^*}\ne \emptyset$, $\pi^*\in\cpolyz{i^*}$ (which means that $\cpolyz{i^*}\cap\conv(\Pi)\ne \emptyset$), and $\dim(\cpoly{i^*})=\alpha_n$. This means that the polynomial case of the sup part of Theorem~\ref{thm:maintechh} holds when the theorem is applied to $\Pi$ and $\cpoly{i^*}$. Therefore, there exists 
$\vec\pi\in\Pi^n$ such that $\Pr\left(\vXp \in \cpoly{i^*}\right) = \Omega(n^{\frac{\dim(\cpoly{i^*})-q}{2}})$. Notice that  we cannot immediately let $\vec\pi = (\pi^*,\ldots,\pi^*)$ because it is possible that $\pi^*\notin \Pi$.
It follows after the lower bound of (\ref{eq:upoly-high-level}) that 
\begin{equation*}
\label{eq:sup-poly-lower}\Pr\left(\vXp \in \upoly\right) \ge \max\nolimits_{i\le I}\Pr\left(\vXp \in \cpoly{i}\right) \ge \Pr\left(\vXp \in \cpoly{i^*}\right)= \Omega\left(n^{\frac{\dim(\cpoly{i^*})-q}{2}}\right) = \Omega\left(n^{\frac{\alpha_n-q}{2}}\right)
\end{equation*}
This proves the polynomial lower bound of sup.

To prove the {\bf polynomial upper bound  of sup}, it suffices to prove that for every $\vec\pi\in\Pi^n$, we have $\Pr\left(\vXp \in \upoly\right) = O\left(n^{\frac{\alpha_n-q}{2}}\right)$. This will be proved by combining the upper bound in (\ref{eq:upoly-high-level}) and the ($0$, exponential, or polynomial) upper bounds on the sup part of Theorem~\ref{thm:maintechh} when the theorem is applied to $\Pi$ and every $\cpoly{i}$. More precisely, for each $i\le I$, applying the sup part of Theorem~\ref{thm:maintechh} to $\Pi$ and $\cpoly{i}$ results in the following three cases. 
\begin{itemize}
\item {\bf\boldmath The $0$ case.} If $\cpolynint{i}=\emptyset$, then for every $\vec\pi\in\Pi^n$, we have $\Pr\left(\vXp \in \cpoly{i}\right)=0$.
\item {\bf The exponential case.}  If $\cpolynint{i}\ne \emptyset$ and $\cpolyz{i}\cap \conv(\Pi)=0$, then for every $\vec\pi\in\Pi^n$, we have $\Pr\left(\vXp \in \cpoly{i}\right)=\exp(-\Omega(n))$.
\item {\bf The polynomial case.} Otherwise ($\cpolynint{i}\ne \emptyset$ and $\cpolyz{i}\cap \conv(\Pi)\ne 0$),  for every $\vec\pi\in\Pi^n$, we have $\Pr\left(\vXp \in \cpoly{i}\right)=O\left(n^{\frac{\dim(\cpolyz{i})-q}{2}}\right)$.
\end{itemize}
In the polynomial case, for any $\pi\in \cpolyz{i}\cap \conv(\Pi)$, we have $i\in \calI_{\upoly,n}^\pi\ne \emptyset$, which means that $\pi\in \Pi_{\upoly,n}$. This means that $\dim(\cpolyz{i})\le \alpha_n$. Therefore, in all  three cases we have $\Pr\left(\vXp \in \cpoly{i}\right)=O\left(\frac{\alpha_n-q}{2}\right)$. 
Again, recall that $I$ is a constant. Therefore, following the upper bound in (\ref{eq:upoly-high-level}), for every $\vec\pi\in\Pi^n$, we have:
$$\Pr\left(\vXp \in \upoly\right) \le \sum\nolimits_{i\le I}\Pr\left(\vXp \in \cpoly{i}\right) \le I \cdot O\left(n^{\frac{\alpha_n-q}{2}}\right)=O\left(n^{\frac{\alpha_n-q}{2}}\right)$$
This proves the polynomial upper bound of sup.
\end{itemize}

\paragraph{\bf \boldmath Proof for the  $\inf$ part.} We now turn to the $\inf$ part by discussing the three cases ($0$, exponential, and polynomial) as follows. At a high level, the proofs for the  upper (respectively, lower) bounds of inf are similar to the proofs for the lower (respectively, upper) bounds of sup. We include the formal proof below for completeness.
\begin{itemize}
\item {\bf  \boldmath The $0$ case of $\inf$.} Like the $0$ case of the $\sup$ part, this case is straightforward. 
\item {\bf The exponential case of inf.} We recall that in this case $\upolynint \ne \emptyset$ and $\Pi_{\upoly,n}\ne \conv(\Pi)$.  To prove the {\bf exponential lower bound of inf}, it suffices to prove that for every $\vec\pi\in\Pi^n$, we have $\Pr\left(\vXp \in \upoly\right) = \exp(-O(n))$.  Because $\upolynint = \bigcup_{i\le I} \cpolynint{i}\ne \emptyset$, by definition there exists $i^*\le I$ such that $\cpolynint{i^*}\ne \emptyset$. The exponential lower bound follows after the lower bound of (\ref{eq:upoly-high-level}): 
$$\Pr\left(\vXp \in \upoly\right) \ge \max\nolimits_{i\le I}\Pr\left(\vXp \in \cpoly{i}\right) \ge \Pr\left(\vXp \in \cpoly{i^*}\right) \ge \epsilon^{n}=\exp(-O(n))$$

To prove the {\bf exponential upper bound of inf}, it suffices to prove that  there exists $\vec\pi\in\Pi^n$ such that $\Pr\left(\vXp \in \upoly\right) = \exp(-\Omega(n))$.  Because $\Pi_{\upoly,n}\ne \conv(\Pi)$, there exists $\pi\in \conv(\Pi)$ such that $\pi\not\in \Pi_{\upoly,n}$, which means that $\calI_{\upoly,n}^\pi = \emptyset$.  Let $\vec \pi  = (\pi_1,\ldots,\pi_n) \in \Pi$ denote an arbitrary vector such that $\sum_{j=1}^n \pi_j $ is $\Theta(1)$ away from $n \cdot \pi$. Because $\calI_{\upoly,n}^\pi = \emptyset$, for all $i\le I$, either $\cpolynint{i}=\emptyset$ or $\pi\not\in \cpolyz{i}$ (or both hold). If the former holds, then by applying the $0$ upper bound on the {sup} part of Theorem~\ref{thm:maintechh} to $\Pi$ and $\cpoly{i}$,  we have $\Pr\left(\vXp \in \cpoly{i}\right) = 0$.  If the former does not hold and the latter holds, then following Hoeffding's inequality and the union bound (applied to each of the $q$ coordinate of $\vXp$),  we have $\Pr\left(\vXp \in \cpoly{i}\right) = \exp(-\Omega(n))$.  Therefore, following the upper bound of (\ref{eq:upoly-high-level}),  we have:
$$\Pr\left(\vXp \in \upoly\right) \le \sum\nolimits_{i\le I}\Pr\left(\vXp \in \cpoly{i}\right) \le I \cdot\exp(-\Omega(n))=\exp(-\Omega(n))$$
This proves the exponential upper bound of inf.

\item {\bf The polynomial bounds of inf.} We recall that in this case $\upolynint \ne \emptyset$ and $\Pi_{\upoly,n}= \conv(\Pi)$. Like in the exponential case, the proof is done by combining (\ref{eq:upoly-high-level}) and the applications of Theorem~\ref{thm:maintechh} to $\Pi$ and each $\cpoly{i}$. 
To prove the {\bf polynomial lower bound of inf}, it suffices to prove that for every $\vec\pi = (\pi_1,\ldots,\pi_n)\in\Pi^n$, we have $\Pr\left(\vXp \in \upoly\right) = \Omega(n^{\frac{\beta_n-q}{2}})$. Let $\pi^* = \frac 1n \sum_{i=1}^n\pi_i\in \conv(\Pi)$. Because $\Pi_{\upoly,n}= \conv(\Pi)$, we have $\pi^*\in \Pi_{\upoly,n}$, which means that $\calI_{\Pi,n}^{\pi^*}\ne \emptyset$. Therefore, there exists $i^*\in \calI_{\Pi,n}^{\pi^*}$ such that $\dim(\cpolyz{i^*})\ge \beta_n$. Additionally, recall that for all $i\in \calI_{\Pi,n}^{\pi^*}$ we have $\cpolynint{i^*}\ne \emptyset$ and $\pi^*\in\cpolyz{i^*}$. Therefore, Lemma~\ref{lem:lower-any-pi} can be applied to $\vec \pi$ and  $\cpolynint{i^*}$, giving us $\Pr\left(\vXp \in \cpoly{i^*}\right) = \Omega\left(n^{\frac{\beta_n-q}{2}}\right)$.   It follows after the lower bound of (\ref{eq:upoly-high-level}) that: 
$$\Pr\left(\vXp \in \upoly\right) \ge \max\nolimits_{i\le I}\Pr\left(\vXp \in \cpoly{i}\right) \ge \Pr\left(\vXp \in \cpoly{i^*}\right)\ge \Omega\left(n^{\frac{\dim(\cpoly{i^*})-q}{2}}\right) = \Omega\left(n^{\frac{\beta_n-q}{2}}\right)
$$
This proves the polynomial lower bound of inf.

To prove the {\bf polynomial upper bound of inf}, it suffices to prove that  there exists $\vec\pi\in\Pi^n$ such that $\Pr\left(\vXp \in \upoly\right) = O(n^{\frac{\beta_n-q}{2}})$. Because $\Pi_{\upoly,n}=\conv(\Pi)$, we define $(\pi^*,i^*)$ to be the pair that achieves $\beta_n$. More precisely, let
$$(\pi^*,i^*)=\arg \min\nolimits_{\pi\in\Pi_{\upoly,n}}\max\nolimits_{i\in \calI_{\upoly,n}^\pi}\dim(\cpolyz{i})$$
It follows that $\cpolynint{i^*}\ne \emptyset$, $\pi^*\in\cpolyz{i^*}$ (which means that $\cpolyz{i^*}\cap\conv(\Pi)\ne \emptyset$), and $\dim(\cpoly{i^*})=\beta_n$.  Let $\vec \pi = (\pi_1,\ldots,\pi_n)\in\Pi^n$ denote an arbitrary vector such that $\sum_{j=1}^n \pi_n$ is $\Theta(1)$ from $n\pi^*$.  The polynomial upper bound of inf is proved by combining the upper bound in (\ref{eq:upoly-high-level}) and the ($0$, exponential, or polynomial) upper bounds on the sup part of Theorem~\ref{thm:maintechh} when the theorem is applied to $\Pi$ and every $\cpoly{i}$. More precisely, for each $i\le I$, applying the sup part of Theorem~\ref{thm:maintechh} to $\vec \pi$ and $\cpoly{i}$ results in the following three cases. 
\begin{itemize}
\item {\bf\boldmath The $0$ case.} If $\cpolynint{i}=\emptyset$, then   $\Pr\left(\vXp \in \cpoly{i}\right)=0$.
\item {\bf The exponential case.}  If $\cpolynint{i}\ne \emptyset$ and $\pi^*\not\in \cpolyz{i}$, then $\expect(\vXp)$  is $\Theta(n)$ away from $\cpolyz{i}$. Following Hoeffding's inequality and the union bound (applied to each of the $q$ coordinate of $\vXp$), we have $\Pr\left(\pvX{\vec\pi^*} \in \cpoly{i}\right)=\exp(-\Omega(n))$.
\item {\bf The polynomial case.} Otherwise ($\cpolynint{i}\ne \emptyset$ and $\pi^*\in \cpolyz{i}$, which means that $i\in \calI_{\Pi,n}^{\pi^*}$), we have $$\Pr\left(\vXp \in \cpoly{i}\right)=O\left(n^{\frac{\dim(\cpolyz{i})-q}{2}}\right) \le O\left(n^{\frac{\beta_n-q}{2}}\right),$$
because $$\dim(\cpoly{i})\le \dim(\cpoly{i^*})=\beta_n$$
\end{itemize}
Therefore, in all  three cases above, we have $\Pr\left(\pvX{\vec\pi^*} \in \cpoly{i}\right)=O(\frac{\beta_n-q}{2})$. 
Again, recall that $I$ is a constant. Therefore, following the upper bound in (\ref{eq:upoly-high-level}), we have:
$$\Pr\left(\pvX{\vec\pi^*} \in \upoly\right) \le \sum\nolimits_{i\le I}\Pr\left(\vXp \in \cpoly{i}\right) \le I \cdot O\left(n^{\frac{\beta_n-q}{2}}\right)=O\left(n^{\frac{\alpha_n-q}{2}}\right)$$
This proves the polynomial upper bound of inf.
\end{itemize}\end{proof}

\section{Appendix for Section~\ref{sec:posrules}: Integer Positional Scoring Rules}

\subsection{Proof of Theorem~\ref{thm:score}}
\label{app:proof-score}

\appThm{thm:score}{Smoothed likelihood of ties: positional scoring rules}{
Let $\mm= (\Theta,\ml(\ma),\Pi)$ be a strictly positive and closed single-agent preference model and let $\vec s$ be an integer scoring vector. For any $2\le k\le m$ and any $n\in\mathbb N$, 
$$\slt{\Pi}{r_{\vec s}}{m}{k}{n}=\left\{\begin{array}{ll}0 &\text{if } \forall P\in \ml(\ma)^n, |\cor_{\vec s}(P)|\ne k\\
\exp(-\Theta(n)) &\text{otherwise, if } \forall\pi\in \conv(\Pi), |\cor_{\vec s}(\pi)|< k\\
\Theta(n^{-\frac{k-1}{2}}) &\text{otherwise}
\end{array}\right.$$
$$\ilt{\Pi}{r_{\vec s}}{m}{k}{n}=\left\{\begin{array}{ll}0 &\text{if } \forall P\in \ml(\ma)^n, |\cor_{\vec s}(P)|\ne k\\
\exp(-\Theta(n)) &\text{otherwise, if }\exists\pi\in \conv(\Pi), |\cor_{\vec s}(\pi)|< k\\
\Theta(n^{-\frac{k-1}{2}}) &\text{otherwise}
\end{array}\right..$$
}
\begin{proof}  The theorem is proved by modeling the set of profiles with $k$ winners as the union of constantly many polyhedra, then applying Theorem~\ref{thm:union-poly}. More precisely, we have the following three steps.  In Step 1, for each potential winner set $T\subseteq \ma$,  we define a polyhedron $\ppoly{\vec s, T}$ that characterizes profiles whose winners are $T$. In Step 2, we prove properties of $\ppoly{\vec s, T}$, and in particular,  $\dim(\polyz^{\vec s, T}) = m! - |T|+1$. In Step 3 we formally apply Theorem~\ref{thm:union-poly} to $\upoly = \bigcup_{T\subseteq \ma: |T|=k}\ppoly{\vec s, T}$.

\paragraph{\bf \boldmath Step 1: Define $\bm{\ppoly{\vec s, T}}$.} For any $T\subseteq \ma$, $\ppoly{\vec s, T}$ consists of (i) equations that represent the scores of alternatives in $T$ being equal, and (ii) inequalities that represent the score of any  alternative in $T$ being strictly larger than the score of any alternative in $\ma\setminus T$. Formally, we first define a set of constraints to model the score difference between two alternatives. 

\begin{dfn}[\bf Score difference vector]\label{dfn:varcons} For any scoring vector $\vec s = (s_1,\ldots,s_m)$ and any pair of different alternatives $a,b$, let $\score_{a,b}^{\vec s}$ denote the $m!$-dimensional vector indexed by rankings in $\ml(\ma)$: for any $R\in\ml(\ma)$, the $R$-element of $\score_{a,b}^{\vec s}$ is $s_{j_1}-s_{j_2}$, where $j_1$ and $j_2$ are the ranks of $a$ and $b$ in $R$, respectively.
\end{dfn}
Let $\vec x_\ma = (x_{R}:R\in\ml(\ma))$ denote the vector of $m!$ variables, each of which represents the multiplicity of a linear order in a profile. Therefore, $\score_{a,b}^{\vec s}\cdot \vec x_\ma$ represents the score difference between $a$ and $b$ in the profile whose histogram is $\vec x_\ma$.  

For any $T\subseteq \ma$, we define the polyhedron $\poly^{\vec s,T}$ as follows.

\begin{dfn}\label{dfn:convscore}
For any integer scoring vector $\vec s$ and any $T\subseteq \ma$, we let $\pbe{\vec s, T}$ denote the matrix whose row vectors are $\{\score_{a,b}^{\vec s}: a\in T, b\in T, a\ne b\}$. Let  $\pbs{\vec s, T}$ denote the matrix whose row vectors are $\{\score_{a,b}^{\vec s}: a\not\in T, b\in T\}$. Let $\pba{\vec s, T} = \left[\begin{array}{c}\pbe{\vec s, T} \\ \pbs{\vec s, T}\end{array}\right]$,  $\vec b = (\vec 0, {-\vec 1})$, and let $\ppoly{\vec s, T}$ denote the corresponding polyhedron.
\end{dfn}

\paragraph{\bf \boldmath Step 2: Prove properties about $\ppoly{\vec s, T}$.} We prove the following claim about $\ppoly{\vec s, T}$.
\begin{claim}\label{claim:scoresolumg}
For any integer scoring vector $\vec s$ and any $T\subseteq\ma$,  we have:
\begin{enumerate}[label=(\roman*)]
\item for any integral profile $P$,  $\hist(P)\in \poly^{\vec s,T}$ if and only if $\cor_{\vec s}(P) = T$;
\item for any $\vec x\in\mathbb R^q$, $\pi\in {\ppolyz{\vec s,T}}$ if and only if $T\subseteq \cor_{\vec s}(\vec x)$;
\item $\dim(\ppolyz{\vec s,T})= m! - |T|+1$.
\end{enumerate}
\end{claim}
\begin{proof} (i)  and (ii) follow after Definition~\ref{dfn:convscore}. To prove (iii), we first note that $\dim(\ppolyz{\vec s,T})=m!-\rank(\ba^=)$, where $\ba^=$ is the essential equalities of $\pba{\vec s,T}$, according to equation (9) on page 100 in~\citep{Schrijver1998:Theory}. We note that $\ba^= = \be^{\vec s,T}$, because by definition $\be^{\vec s,T}\subseteq \ba^=$. To prove that no row in $\bs^{\vec s,T}$ is in $\ba^=$, it suffices to show that there exists $\vec x\in \ppolyz{\vec s,T}$ such that $\bs^{\vec s,T}\cdot\invert{\vec x}< \invert{\vec 0}$. This is proved by constructing a profile $P$ such that $\cor_{\vec s}(P) = T$ and then let $\vec x = \hist(P)$. We first define two cyclic permutations: $\sigma_1$ over $T$ and $\sigma_2$ over $\ma\setminus T$. More precisely, let 
$$\sigma_1=1\ra 2\ra\cdots \ra |T|\ra 1\text{ and } \sigma_2 = |T|+1\ra |T|+2\ra\cdots \ra m\ra |T|+1$$
Then we let 
$$P = \{\sigma_1^{i}(\sigma_2^j(1\succ 2\succ \cdots \succ m)): i\le |T|, j\le m-|T|\}$$
It is not hard to see that $\cor_{\vec s}(P) = T$ and consequently $\be^{\vec s,T} \supseteq \pba{=}$, which means that $\be^{\vec s,T} = \pba{=}$.

Therefore, it suffices to prove that $\rank(\be^{\vec s,T}) = |T|-1$. We first prove that $\rank(\be^{\vec s,T})\le |T|-1$. W.l.o.g.~let $T = \{1,\ldots, k\}$, where $k=|T|$. It is not hard to verify that all rows of $\be^{\vec s,T}$ can be represented by linear combinations of $|T|-1$ rows $\left[\begin{array}{c}\score_{1,2}\\ \vdots \\ \score_{1,k}\end{array}\right]$. This means that $\rank(\be^{\vec s,T})\le |T|-1$.

Next, we prove $\rank(\be^{\vec s,T})\ge |T|-1$ by contradiction. Suppose for the sake of contradiction that $\rank(\be^{\vec s,T})\le |T|-2$. Then, there exists $\ell\le k-2$ rows of $\be^{\vec s,T}$, denoted by $\{\score_{a_i,b_i}^{\vec s}: i\le \ell\}$, such that any $\score_{a,b}^{\vec s}$ is a linear combination of them. Let $G$ denote the unweighted undirected graph over $\{1,\ldots,k\}$ with the following $\ell$ edges: $\{\{a_i,b_i\}:i\le \ell\}$. Because $G$ contains $k-2$ edges and $k$ nodes, $G$ is not connected. W.l.o.g.~suppose there is no edge in $G$ between $\{1,\ldots, t\}$ and $\{t+1,\ldots, k\}$ for some $1\le t\le k-1$. We now construct a profile $P$ to show that $\score_{1,k}^{\vec s}$ is not a linear combination of $\{\score_{a_i,b_i}^{\vec s}:i\le \ell\}$. We first define a linear order $R$ and two cyclic permutations.
$$R=[1\succ \ldots\succ t\succ k+1\succ \cdots \succ m\succ t+1\succ \cdots \succ k]$$
Let $\sigma_1=1\ra\cdots\ra t\ra 1$ denote the cyclic permutation among $\{1,\ldots, t\}$ and let $\sigma_2=t+1\ra\cdots\ra k\ra t+1$ denote the cyclic permutation among $\{t+1,\ldots, k\}$. Let 
$$P=\{\sigma_1^i(\sigma_2^j(R)): 1\le i\le t, 1\le j\le k-t\}$$ 
It is not hard to verify that in $P$, the total scores of $\{1,\ldots, t\}$ are the same, the scores of $\{t+1,\ldots, k\}$ are the same, and the former is strictly larger than the latter. Therefore, for all $i\le \ell$, $\score_{a_i,b_i}^{\vec s}\cdot\hist(P)=0$, but $\score_{1,k}^{\vec s}\cdot \hist(P)> 0$, which means that $\score_{1,k}^{\vec s}$ is not a linear combination of $\{\score_{a_i,b_i}^{\vec s}:i\le \ell\}$, which is a contradiction. This means that $\rank(\be^{\vec s,T})\ge |T|-1$.

This completes the proof of Claim~\ref{claim:scoresolumg}.
\end{proof}

\paragraph{\bf \boldmath Step 3: Apply Theorem~\ref{thm:union-poly}.}  Let $\upoly = \bigcup_{T\subseteq \ma: |T|=k}\ppoly{\vec s, T}$. It follows that for any profile $P$, $|\cor_{\vec s}(P)| = k$ if and only if $\hist(P)\in \upoly$. Therefore, for any $n$ and any $\vec\pi\in\Pi^n$, we have:
$$\Pr\nolimits_{P\sim\vec\pi} (|\cor_{\vec s}(P)| = k) = \Pr(\vXp\in \upoly)$$
Recall that for all $T_k\subseteq \ma$ with $|T_k|=k$, from Claim~\ref{claim:scoresolumg} (iii) we have $\dim(\ppoly{\vec s,T_k}) = m!-|T_k|+1 = m!-k+1$. This means that $\alpha_n = \beta_n = m!-k+1$ when they are non-negative (which holds for certain $n$). 

Therefore, to prove Theorem~\ref{thm:score}, it suffices to prove that the conditions for the $0$, exponential, and polynomial cases in Theorem~\ref{thm:score} are equivalent to the conditions for the $0$, exponential, and polynomial cases in Theorem~\ref{thm:union-poly} (applied to $\upoly$ and $\Pi$), respectively.  The $0$ case is straightforward. Recall from Claim~\ref{claim:scoresolumg}(iii) that $\dim(\ppolyz{\vec s,T})= m! - |T|+1$. Therefore, to apply Theorem~\ref{thm:union-poly} to obtain  Theorem~\ref{thm:score}, it suffices to prove that for any $\pi\in\conv(\Pi)$, $|\cor_{\vec s}(\pi)|\ge k$ if and only if $\pi\in \upolyz$ (which is equivalent to $\alpha_n\ne -\infty$). The ``if'' direction holds because if $\pi\in \upolyz$, then there exists $T\subseteq\ma$ with $|T|=k$ such that $\pi\in \ppolyz{\vec s,T}$, which means that  $T\subseteq \cor_{\vec s}(\pi)$ by Claim~\ref{claim:scoresolumg}(ii). Therefore, $|\cor_{\vec s}(\pi)|\ge k|$. The ``only if'' direction holds because if $|\cor_{\vec s}(\pi)|\ge k|$, then let $T\subseteq \cor_{\vec s}(\pi)$ denote an arbitrary set with $|T|=k$. It follows from Claim~\ref{claim:scoresolumg}(ii) that $\pi\in \ppolyz{\vec s,T}$, which implies that $\pi\in \upolyz$. 
\end{proof}

\section{Appendix for Section~\ref{sec:eorules}: EO-Based Rules}
\label{app:smoothedties}

\subsection{Proof of Proposition~\ref{prop:eo}}
\label{appendix:proof-lemeo}
\appPropnoname{prop:eo}{}{
 For any $\ma$ and any $n\ge m^4$,
$\{\eo(P):P\in \ml(\ma)^n\} = \left\{\begin{array}{ll} \pale &\text{if }2\mid n\\
\palo &\text{if }2\nmid n\\
\end{array}\right.$.
}
\begin{proof}
Let $\eo_n = \{\eo(P):P\in \ml(\ma)^n\}$. We first prove the $2\mid n$ case.  By definition $\eo_n\subseteq \pale$. We prove $\pale\subseteq \eo_n$ by explicitly constructing an $n$-profile $P$ such that $\eo(P) = O$ for any $O\in \pale$. Let $O=T_1\rhd \cdots\rhd T_t \rhd T_0\rhd T_{t+1}\rhd\cdots \rhd T_{2t}$ denote the tier representation of $O$. We define a weighted directed graph $G_O$ such that $$w_{G_O}(e) = \left\{\begin{array}{rl}
2(t+1-i)&\text{if }e\in T_i \text{ for some }1\le i\le t\\
-2(t+1-i)&\text{if }\bar e\in T_i \text{ for some }1\le i\le t\\
0&\text{if }e\in T_0
\end{array}\right.$$
By McGarvey's theorem~\cite{McGarvey53:Theorem}, there exists a profile $P_O$ of no more than $m(m+1)t<m^4$ votes such that $\wmg(P_O) = G_O$. Let  $R$ be an arbitrary linear order and let $\bar R$ denote its reverse order. Let $P_2 = \{R,\bar R\}$. It follows that $\wmg(P_2)$ is the empty graph. Let $P = P_O+\frac{n-|P_O|}{2}\times P_2$. It follows that $|P|=n$ and $\wmg(P)=\wmg(P_O) = G_O$, which means that $\eo(P) = O$. This means that $\pale\subseteq \eo_n$. Therefore, the lemma holds for $2\mid n$.

Next, we prove the $2\nmid n$ case. For any $n$-profile $P$, the middle tier of $\eo(P)$ is empty because the weight on any edge must be odd. This means that $\eo_n\subseteq \palo$. Like the $2\mid n$ case, we prove $\palo\subseteq \eo_n$ by explicitly constructing an $n$-profile $P$ such that $\eo(P) = O$ for any $O\in \palo$. More precisely, let $O=T_1\rhd \cdots\rhd T_t \rhd T_0\rhd T_{t+1}\rhd\cdots \rhd T_{2t}$ denote the tier representation of $O$, where $T_0=\emptyset$. We define a weighted directed graph $G_O$ such that $$w_{G_O}(e) = \left\{\begin{array}{rl}
2(t-i)+1&\text{if }e\in T_i \text{ for some }1\le i\le t\\
-2(t-i)-1&\text{if }\bar e\in T_i \text{ for some }1\le i\le t\end{array}\right.$$
By McGarvey's theorem~\cite{McGarvey53:Theorem}, there exists a profile $P_O$ of no more than $m(m+1)t<m^4$ votes such that $\wmg(P_O) = G_O$. Let $P = P_O+\frac{n-|P_O|}{2}\times P_2$. It follows that $|P|=n$ and $\wmg(P)=\wmg(P_O) = G_O$, which means that $\eo(P) = O$. This means that $\palo\subseteq \eo_n$. Therefore, the lemma holds for $2\nmid n$.
\end{proof}

\subsection{Formal Definition of Edge-Order-Based Rules (Section~\ref{sec:eorules})}
\label{app:def-eo}

\begin{dfn}Given an edge-order-based rule $\cor$, any $2\le k\le m$, any $n\in\mathbb N$, and any distribution $\pi$ over $\ml(\ma)$, let 
$$\mo_{\cor,k,n}^{\pi} = \{O\in\pal{n}: |\cor(O)|=k\text{ and }O\text{ refines }\eo(\pi)\}$$ 
For any set of distributions $\Pi$ over $\ml(\ma)$, let $\mo_{\cor,k,n}^{\Pi} = \bigcup_{\pi\in\conv(\Pi)}\mo_{\cor,k,n}^{\pi}$. When $\mo_{\cor,k,n}^{\Pi}\ne \emptyset$, we let  $\ell_{\min} = \min\nolimits_{O\in  \mo_{\cor,k,n}^{\Pi}}\ties(O)$. When $\mo_{\cor,k,n}^{\pi}\ne \emptyset$ for all $\pi\in\conv(\Pi)$, we let 
$\ell_{\text{mm}} = \max\nolimits_{\pi\in\conv(\Pi)}\min\nolimits_{O\in  \mo_{\cor,k,n}^{\pi}}\ties(O)$.
\end{dfn}

\subsection{Proof of Theorem~\ref{thm:eorule}}
\label{app:proof-eorule}

\appThm{thm:eorule}{Smoothed likelihood of ties: edge-order-based rules}{
Let $\mm= (\Theta,\ml(\ma),\Pi)$ be a strictly positive and closed single-agent preference model and let $\cor$ be an edge-order-based rule. For any $2\le k\le m$ and any $n\in\mathbb N$,  
\begin{align*}
&\slt{\Pi}{r}{m}{k}{n}=\left\{\begin{array}{ll}0 &\text{if } \forall O\in \pal{n}, |\cor(O)|\ne k\\
\exp(-\Theta(n)) &\text{otherwise if } \mo_{\cor,k,n}^{\Pi} = \emptyset\\
\Theta\left(n^{-\frac{\ell_{\min}}{2}}\right) &\text{otherwise}
\end{array}\right.\\
&\ilt{\Pi}{r}{m}{k}{n}=\left\{\begin{array}{ll}0 &\text{if } \forall O\in \pal{n}, |\cor(O)|\ne k\\
\exp(-\Theta(n)) &\text{otherwise if } \exists \pi\in\conv(\Pi) \text{ s.t. }\mo_{\cor,k,n}^{\pi} = \emptyset\\
\Theta\left(n^{-\frac{\ell_{\text{mm}}}{2}}\right) &\text{otherwise}
\end{array}\right.
\end{align*}
}
\begin{proof} Like in the proof of Theorem~\ref{thm:score}, the theorem is proved by modeling the set of profiles with $k$ winners as the union of constantly many polyhedra, then applying Theorem~\ref{thm:union-poly}. More precisely, we have the following three steps.  In Step 1, for each palindromic order $O$, we define a polyhedron $\ppoly{O}$ that characterizes the profiles whose palindromic orders are $O$. In Step 2, we prove properties about $\ppoly{O}$, in particular $\dim(\ppolyz{O}) = m! - \ties(O)$. In Step 3 we formally apply Theorem~\ref{thm:union-poly} to $\upoly = \bigcup_{O\in \pal{n}: |\cor(O)|=k}\ppoly{O}$.

We first recall the definition of pairwise difference vector in~\citep{Xia2020:The-Smoothed}.
\begin{dfn}[\bf Pairwise difference vector~\citep{Xia2020:The-Smoothed}]\label{dfn:pairdiff} For any pair of different alternatives $a,b$, let $\pair_{a,b}$ denote the $m!$-dimensional vector indexed by rankings in $\ml(\ma)$: for any $R\in\ml(\ma)$, the $R$-element of $\pair_{a,b}$ is $1$ if $a\succ_R b$; otherwise it is $-1$. 
\end{dfn}

\paragraph{\bf \boldmath Step 1: Define $\ppoly{O}$.} For any palindromic order $O\in \pale$, we define a polyhedron $\ppoly{O}$ that represents profiles whose edge order is $O$.
\begin{dfn}\label{dfn:convpalindromic}
For any $O\in \pale$, we let $\pbe{O}$ denote the matrix whose row vectors are $\{\pair_{a,b} - \pair_{c,d}: (a,b)\equiv_O(c,d)\}$; we let $\pbs{O}$ denote the matrix whose row vectors are $\{\pair_{c,d} - \pair_{a,b}: (a,b)\rhd_O(c,d)\}$.   Let $\pba{O} = \left[\begin{array}{c}\pbe{O} \\ \pbs{O}\end{array}\right]$,  $\vec b = (\vec 0, {-\vec 1})$, and let $\ppoly{O}$ denote the corresponding polyhedron.
\end{dfn}
$\pbe{O}$ contains redundant rows and is defined as in Definition~\ref{dfn:convpalindromic} for notational convenience. For any $1\le i\le t$, rows in $\pbe{O}$ that correspond to $T_i$ are the same as the rows  in $\pbe{O}$ that correspond to $T_{2t+1-i}$. 

\paragraph{\bf \boldmath Step 2: Prove properties about $\ppoly{O}$.} We have the following claim about $\ppoly{O}$.
\begin{claim}\label{claim:solpalindromic}
For any profile $P$ and any $O\in \pale$, we have:
\begin{enumerate}[label=(\roman*)]
\item $\hist(P)\in \ppoly{O}$ if and only if $\eo(P) = O$.
\item $\hist(P)\in {\ppolyz{O}}$ if and only if $O$  refines $\eo(P)$.  
\item $\dim(\ppolyz{O})= m! - \ties(O)$.
\end{enumerate}
\end{claim}
\begin{proof} Part (i) and (ii) follow after the definition of $\ppoly{O}$ and $\ppolyz{O}$, respectively. For part (iii), let $O=T_1\rhd\cdots\rhd T_t\rhd T_0\rhd T_{t+1}\rhd\cdots\rhd T_{2t}$ denote the tier representation of $O$. Let $\ba^=$ denote the essential equalities of $\pba{O}$. Again, it is not hard to verify that $\ba^= = \pbe{O}$ due to McGarvey's theorem~\citep{McGarvey53:Theorem}. Therefore, it suffices to prove  $\rank(\pbe{O}) = \ties(O)$. 

We first note that $\rank(\pbe{O})\le \ties(O)$, because for any $0\le i\le t$ such that $T_i=\{e_1,\ldots, e_\ell\}$ and any profile $P$, the following $\ell-1$ equations 
$$(\pair_{e_1}-\pair_{e_2})\cdot \hist(P)=0, (\pair_{e_1}-\pair_{e_3})\cdot \hist(P)=0,\ldots, (\pair_{e_1}-\pair_{e_\ell})\cdot \hist(P)=0$$
 imply that all edges in $T_i$ have the same weights in $\wmg(P)$.  To prove that $\rank(\pbe{O})\ge \ties(O)$, suppose for the sake of contradiction that $\rank(\pba{O})< \ties(O)$. Then, there exists a tier $T_i= \{e_1,\ldots, e_\ell\}$ where $0\le i\le t$ and no more than $s<\ell-1$ rows in $\pbe{O}$, denoted by  $\be' = \invert{\pair_{e_{i_1}}-\pair_{e_{j_1}}, \ldots,\pair_{i_s}-\pair_{j_s}}$, whose linear combinations include other rows that correspond to $T_i$. Let $\be_i$ denote the rows of $\pbe{O}$ that correspond to $T_i$. By McGarvey's theorem~\citep{McGarvey53:Theorem}, there exists a profile $P$ such that the $w_P(e_{i_1})=w_P(e_{j_1}),\ldots, w_P(e_{i_s})=w_P(e_{j_s})$, but not all $e_1,\ldots,e_\ell$ have the same weights in $\wmg(P)$. This means that $\be'\cdot \invert{\hist(P)}= \invert{\vec 0}$ but $\be_i\cdot \invert{\hist(P)}\ne \invert{\vec 0}$, which means that not all rows in $\be_i$ are linear combinations of rows in $\be'$, which is a contradiction. This proves that $\rank(\pbe{O}) = \ties(O)$, which implies part (iii) of the claim.
\end{proof}

\paragraph{\bf \boldmath Step 3: Apply Theorem~\ref{thm:union-poly}.}  Let $\upoly = \bigcup_{O\in \pal{n}: |\cor(O)|=k}\ppoly{O}$. It follows that for any profile $P$, $|\cor(P)| = k$ if and only if $\hist(P)\in \upoly$. Therefore, for any $n$ and any $\vec\pi\in\Pi^n$, we have:
$$\Pr\nolimits_{P\sim\vec\pi} (|\cor(P)| = k) = \Pr(\vXp\in \upoly)$$
Recall that for all $O\in \pal{n}$ with $|\cor(O)|=k$, from Claim~\ref{claim:solpalindromic} (iii) we have $\dim(\ppoly{O}) = m!-\ties(P)+1$. This means that $\alpha_n =m!-\ell_{\min}$ and $\beta_n = m!-\ell_{\text{mm}}$ when they are non-negative (which holds for certain $n$). 

Therefore, to prove Theorem~\ref{thm:eorule}, it suffices to prove that the conditions for the $0$, exponential, and polynomial cases in Theorem~\ref{thm:eorule} are equivalent to the conditions for the $0$, exponential, and polynomial cases in Theorem~\ref{thm:union-poly} (applied to $\upoly$ and $\Pi$), respectively.  This follows a similar reasoning as in Step 3 of the proof of Theorem~\ref{thm:score} combined with Claim~\ref{claim:solpalindromic} (ii). This completes the proof for Theorem~\ref{thm:eorule}.
\end{proof}

\subsection{Proof of Proposition~\ref{prop:copeland}}
\label{app:proof-copeland}

\appProp{prop:copeland}{Max smoothed likelihood of ties: Copeland$_\alpha$}{
Let $\mm= (\Theta,\ml(\ma),\Pi)$ be a strictly positive and closed single-agent preference model with $\piuni\in \conv(\Pi)$. Let $l_\alpha = \min\{t\in{\mathbb N}: t\alpha\in\mathbb Z\}$. For any $2\le k\le m$ and any $n\in\mathbb N$,  
$$\slt{\Pi}{\copeland}{m}{k}{n}=\left\{\begin{array}{ll}
0 &\text{if } 2\nmid n, 2\mid k, \text{and }\left\{\begin{array}{l}k=m\text{, or}\\k=m-1\end{array}\right.\\
\Theta(n^{-\frac{k}{4}}) &\text{if } 2\mid n,  2\mid  k,  { and }\left\{\begin{array}{l}(1) k=m, \text{ or} \\ 
(2) k=m-1 \text{ and } \alpha\ge \frac 12, \text{ or}\\
(3) k=m-1 \text{ and } k\le l_\alpha(l_\alpha+1)
\end{array}\right.\\
\Theta(n^{-\frac{l_\alpha(l_\alpha+1)}{4}}) &\text{if } 2\mid n,  2\mid  k,  k=m-1, \alpha<\frac12,\text{ and }k> l_\alpha(l_\alpha+1)\\
\Theta(1) &\text{otherwise (i.e., if }2\nmid k \text{ or }k\le m-2\text{)}
\end{array}\right.$$
}

\begin{proof} The proposition is proved by applying Theorem~\ref{thm:eorule}. 
Let $\uge$ denote the set of all directed unweighted graphs over $\ma$. For any graph $G\in \uge$ and any pair of alternatives $a,b$, if there is no edge between $a$ and $b$ then we say that $a$ and $b$ are {\em tied} in $G$. Let $\ugo\subset \uge$ denote the set of all tournament graphs over $\ma$, i.e., $\ugo$ consists of graphs without ties. 

For any profile $P$, let $\umg(P)$ denote the {\em unweighted majority graph} of $P$, which is a graph in $\uge$ such that there is an edge $a\ra b$ if and only if $w_P(a,b)>0$.  It is not hard to see that Copeland$_\alpha$ can be defined over $\uge$. 

Like Proposition~\ref{prop:eo}, the set of all graphs in $\uge$ resulted from UMGs of $n$-profiles can be characterized by $\ug{n}$ that is defined in the following claim.
\begin{claim}
\label{claim:ug}
For any $n>m^4$, we have $\ug{n} =\{\umg(P):P\in \ml(\ma)^n\} =  \left\{\begin{array}{ll} \uge &\text{if } 2\mid n\\ \ugo & \text{if } 2\nmid n \end{array}\right.$.
\end{claim}
\begin{proof} The claim is proved by directly applying Proposition~\ref{prop:eo} and noticing that the middle tier of a palindromic order is empty if and only if its UMG does not contain any ties.
\end{proof}

\noindent{\bf Applying the $\bm 0$ case in Theorem~\ref{thm:eorule}.} In this part we prove that the $0$ case of the proposition can be obtained from applying the $0$ case in Theorem~\ref{thm:eorule}. By Claim~\ref{claim:ug}, it suffices to prove that for any $n>m^4$, there exists $G\in \ug{n}$ such that $|\copeland(G)|=k$ if and only if $2\mid n$, or $2\nmid k$, or $k\le m-2$. 
\begin{itemize}
\item  {\bf The ``if'' direction.} We first prove the ``if'' direction by construction. When $2\mid n$, by Claim~\ref{claim:ug} we have $\ug{n} = \uge$, and we can choose a graph $G$ such that all alternatives in $\{1,\ldots, k\}$ are tied to each other, and all of them are strictly preferred to other alternatives. It follows that $\copeland(G) = \{1,\ldots,k\}$, which means that $|\copeland(G)|=k$. When $2\nmid k$ or $k\le m-2$, we can use the complete graphs in the $\Theta(1)$ case as illustrated in Figure~\ref{fig:copeland} (a) and (b). 

\item  {\bf The ``only if'' direction.}  Now we prove the ``only if'' direction for all $n>m^4$. When $2\nmid n$, $2\mid k$, and $m=k$, we have $\ug{n} = \ugo$. Suppose for the sake of contradiction that there exists $G\in \ugo$ such that $|\copeland(G)|=k$. Notice that $G$ is a complete graph. Therefore, the total Copeland$_\alpha$ score for all alternatives is $m\choose 2$, which means that each alternative must get $\frac{m+1}{2}\not\in\mathbb N$ points. This is a contradiction because the Copeland$_\alpha$ score of each alternative must be an integer as  $G$ is a complete graph.  When $2\nmid n$, $2\mid k$, and $m=k+1$, again we have $\ug{n} = \ugo$. Suppose for the sake of contradiction that there exists $G\in \ugo$ such that $|\copeland(G)|=k$. Again, the total Copeland$_\alpha$ score for all alternatives is $m\choose 2$. Suppose each winner gets $u\in\mathbb N$ points. Then, the total Copeland$_\alpha$ score must be between $(m-1)u$ and $mu-1$. Notice that the total Copeland$_\alpha$ score for all alternatives is $m\choose 2$. If $u\le \frac {m+1}{2}$, then we have ${m\choose 2}> mu-1$; and if $u\ge \frac {m+1}{2}$, then ${m\choose 2}<(m-1)u$. Either case leads to a contradiction.
\end{itemize}

\noindent{\bf Applying the exponential case in Theorem~\ref{thm:eorule}.} We now prove that the exponential case of Theorem~\ref{thm:eorule} do not occur. 
We note that $\umg(\piuni)$ is the empty graph, which means that $\eo(\piuni)$ only contains the middle tier $T_0$. Therefore, for any palindromic order $O$ with $|\copeland(O)|=k$, we have that $O$ refines $\eo(\piuni)$, which means that $\mo_{\cor,k,n}^{\Pi}\ne\emptyset$.

\noindent{\bf Applying the polynomial case in Theorem~\ref{thm:eorule}.}  The remainder of the proof focuses on characterizing $\ell_{\min}$ in the polynomial case of Theorem~\ref{thm:eorule}. We first prove an alternative characterization of  $\ell_{\min}$ that will be frequently used. It states that $\ell_{\min}$ equals to  the minimum number of tied pairs in graphs in $\ug{n}$ where there are exactly $k$ Copeland$_\alpha$ winners. Recall from Claim~\ref{claim:ug} that when $n$ is sufficiently large ($>m^4$), which is the case we will focus in the rest of the proof, for any even number $n$, $\ug{n}=\uge$, which is the set of all unweighted directed graphs over $\ma$; and for any odd number $n$, $\ug{n}=\ugo$, which is the set of all  unweighted tournament graphs over $\ma$. With a little abuse of notation, for any unweighted directed graph $G$ over $\ma$, we let $\ties(G)$ denote the number of tied pairs in $G$. 
\begin{claim}\label{claim:copelandwinner}For any $\Pi$ with $\piuni\in \conv(\Pi)$, $n$, and $k$ such that $\mo_{\cor,k,n}^{\Pi}\ne \emptyset$ under Copeland$_\alpha$, we have:
$$\ell_{\min} = \min\{\ties(G): G\in \ug{n} \text{ s.t. }|\copeland(G)|=k\}$$
\end{claim}
\begin{proof} Because $\eo(\piuni)$ only contains the middle tier $T_0$, any palindromic order refines it, which means that $\mo_{\cor,k,n}^{\Pi} =\{ O\in \pal{n}: |\copeland(O)|=k\}$.

For any palindromic order $O$, let $\umg(O)\in \uge$ be the graph such that there is an edge $a\ra b$ if and only if $a\ra b$ is ranked before the middle tier. It is not hard to verify that $\copeland(O)=\copeland(\umg(O))$. Also it is not hard to verify that $O\in \pal{n}$ if and only if $\umg(O)\in \ug{n}$. Therefore, for any $O\in \mo_{\cor,k,n}^{\Pi}$ such that  $\ties(O) = \ell_{\min}$, we have $\umg(O)\in \ug{n}$ and $|\copeland(\umg(O))|=k$. We also have $\ties(O)\ge \ties(\umg(O))$ because the latter corresponds to the $T_0$ part in $\ties(O)$. Therefore, we have $\ell_{\min} \ge \min\{\ties(G): G\in \ug{n} \text{ s.t. }|\copeland(G)|=k\}$.

For any $G^*=\arg\min_{G\in \ug{n}: |\copeland(G)|=k} \ties(G)$, we construct a palindromic order $O$ whose $T_0$ consists of all tied edges in $G^*$, and each of the other tiers contains exactly one edge. It is not hard to verify that $\ties(O) = \ties(G^*)$ and $\umg(O) = G^*$, which means that $O\in\pal{n}$ and $|\copeland(O)|=k$. Therefore, $\ties(G^*)\ge \ell_{\min}$. This proves the claim. 
\end{proof}
\paragraph{\bf \boldmath The $\bm{\Theta(1)}$ case: $\bm{2\nmid k \text{ or }k\le m-2}$.} We first prove the {$\Theta(1)$ case} as a warm up. Part of the proof will be reused in the proof of other cases. In fact, the $\Theta(1)$ case covers the most number of combinations of $m$ and $k$: it happens when $k$ is odd or $k\le m-2$. In light of Claim~\ref{claim:copelandwinner}, it suffices to prove that in this case there exists a tournament graph $G$ such that $|\copeland(G)| = k$, because $\ties(G) = 0$ and any tournament graph is in $\ug{n}$. We will explicitly construct such $G$ in the two subcases.
\begin{itemize}
\item {$\bm{2\nmid k}$.} The following graph $G$ has exactly $k$ Copeland$_\alpha$ winners $\{1,\ldots,k\}$, see Figure~\ref{fig:copeland} (a) for an example of $m=6,k=5$. For any $1\le i\le k$ and any $1\le s\le \frac{k-1}{2}$ there is an edge $i\ra (1+ (i+s-1 \mod k))$, for example the black edges in Figure~\ref{fig:copeland} (a).  For any $1\le i\le k$ and any $k+1\le j\le m$, there is an edge $i\ra j$, for example the blue edges in Figure~\ref{fig:copeland} (a). 
\item {\bf $\bm{2\mid k}$ and $\bm{k\le m-2}$.} The following  graph $G$ has exactly $k$ Copeland$_\alpha$ winners $\{1,\ldots,k\}$, see Figure~\ref{fig:copeland} (b) for an example of $m=6, k=4$. The edges are defined in the following steps.
\begin{itemize}
\item []{\bf \boldmath Step 1: edges within $\{1,\ldots,k\}$.} These edges are colored black in Figure~\ref{fig:copeland}  (b) and is the similar to the graph in Figure~\ref{fig:copeland} (a), except that there are $\frac k2$ ``diagonal'' edges whose directions need to be assigned because $k$ is an even number.  Formally, for any $1\le i\le\frac{k}{2}$ and any $1\le s\le \frac{k}{2}$, there is an edge $i\ra i+s$. For any $\frac{k}{2}+1\le i\le k$ and any $1\le s\le \frac{k}{2}-1$, there is an edge $i\ra (1+ (i+s-1 \mod k))$. 
\item []{\bf \boldmath Step 2: edges between $\{1,\ldots,k\}$ and $ \{k+1,\ldots,m\}$.} These edges are colored blue in Figure~\ref{fig:copeland}  (b). For any $1\le i\le \frac{k}{2}$ there is an edge $(k+1)\ra i$; and for any other $(i,j)\in \{1,\ldots,k\}\times  \{k+1,\ldots,m\} $, there is an edge $i\ra j$. 
\item []{\bf \boldmath  Step 3: edges within $ \{k+1,\ldots,m\}$.} There is an edge $(k+2)\ra (k+1)$, colored red in Figure~\ref{fig:copeland}  (b). Directions of other edges are assigned arbitrarily to make $G$ a complete graph. 
\end{itemize}
\end{itemize}
\begin{figure}[htp]
\centering
\begin{tabular}{ccc}
    \includegraphics[width=.33\textwidth]{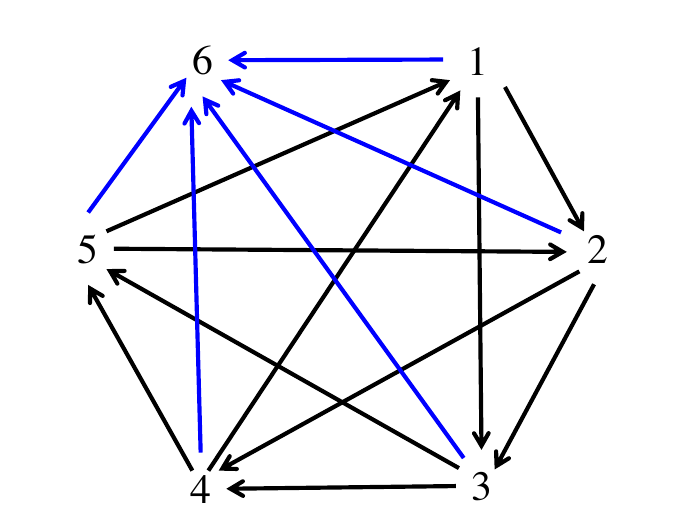} &    \includegraphics[width=.33\textwidth]{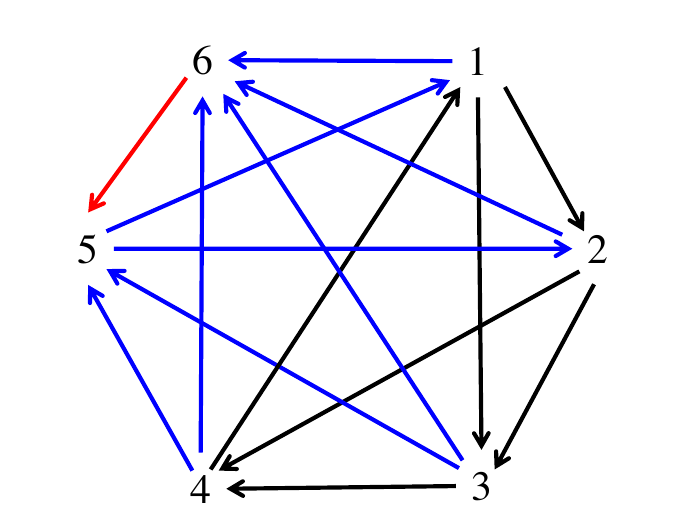} 
    &  \includegraphics[width=.33\textwidth]{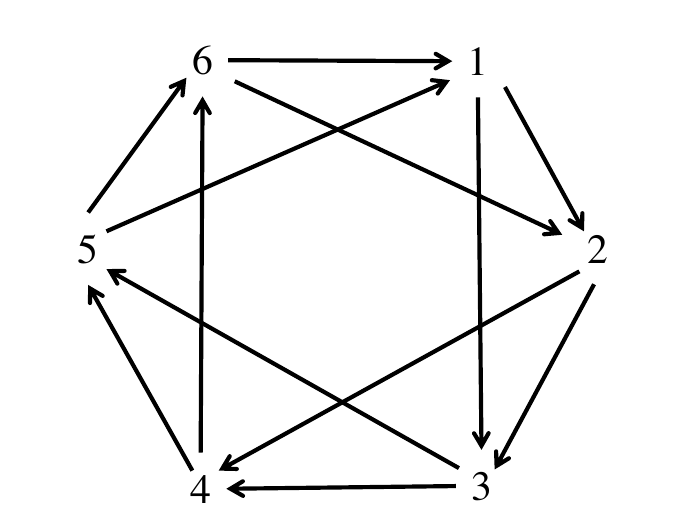} \\
    (a) $m=6$ and $k=5$. &     (b) $m=6$ and $k=4$. & (c) $m=k=6$. 
\end{tabular}
    \caption{Constructions for Copeland$_\alpha$, where the winners are $\{1,\ldots, k\}$.\label{fig:copeland}}
\end{figure}
 
\paragraph{\bf The $\bm{\Theta(n^{-\frac{k}{4}})}$ case:} $\bm{2\mid n,  2\mid  k, } \text{ \bf and }\left\{\begin{array}{l}\bm{(1) k=m,} \text{\bf or} \\ 
\bm{(2) k=m-1 \text{\bf  and } \alpha\ge \frac 12, \text{\bf or}}\\
\bm{(3) k=m-1 \text{\bf and } k\le l_\alpha(l_\alpha+1)}
\end{array}\right.$. We prove that $\ell_{\min} = \frac{k}{2}$ by applying Claim~\ref{claim:copelandwinner} in the three subcases indicated in the statement of the proposition.
\begin{itemize}
\item {\bf (1) $\bm{2\mid n}$ , $\bm{2\mid k}$, and $\bm{k=m}$.}  We first prove $\ell_{\min} \le \frac{k}{2}$. By Claim~\ref{claim:copelandwinner}, it suffices to construct an unweighted graph $G\in\ug{n}=\uge$ such that $|\copeland(G)|=k$ and $\ties(G) = \frac k2$. In fact, the following unweighted directed graph $G$ with $\frac{k}{2}$ ties has  $k$ Copeland winners $\{1,\ldots,k\}$, see Figure~\ref{fig:copeland} (c) for an example. For any $1\le i\le k$ and any $1\le s\le \frac{k}{2}-1$, there is an edge $i\ra (1+ (i+s-1 \mod k))$. There are $\frac{k}{2}$ ties in $G$, namely $\{1,\frac{k}{2}+1\}, \{2,\frac{k}{2}+2\},\ldots, \{\frac{k}{2},k\}$. 

Next, we prove $\ell_{\min}  \ge \frac{k}{2}$. By Claim~\ref{claim:copelandwinner}, it suffices to prove that it is impossible for any unweighted directed graph $G$ with exactly $k=m$ Copeland$_\alpha$ winners to have strictly less than $\frac{k}{2}$ ties. Suppose for the sake of contradiction that such graph exists, denoted by $G'$. Then, there must exist an alternative, w.l.o.g.~alternative $1$, that is not tied to any other alternative in $G'$, otherwise at least $\frac{k}{2}$ ties are necessary to ``cover'' all alternatives. This means that the Copeland$_\alpha$ score of alternative $1$ must be an integer, denoted by $u$, which means that the Copeland$_\alpha$ score of all alternatives must be $u$. Therefore, the total Copeland$_\alpha$ score of all alternatives is $uk$.

Notice that each directed each edge contributes $1$ to the total Copeland$_\alpha$ score and each tie contributes $2\alpha\in [0,2]$ to the total Copeland$_\alpha$ score. Therefore, the total Copeland$_\alpha$ score is between ${m\choose 2}-(\frac{k}{2}-1)$ (which corresponds to the case where there are $\frac{k}{2} -1$ ties and $\alpha =0$) and ${m\choose 2}+(\frac{k}{2}-1)$ (which corresponds to the case where there are $\frac{k}{2} -1$ ties and $\alpha =1$). However, if $u\ge \frac{k}{2}$ then $uk>{m\choose 2}+(\frac{k}{2}-1)$; and if $u\le  \frac{k}{2}-1$ then $uk<{m\choose 2}-(\frac{k}{2}-1)$. In either case $uk\not\in [{m\choose 2}-(\frac{k}{2}-1), {m\choose 2}+(\frac{k}{2}-1)]$, which is a contradiction. Therefore, $\ell_{\min}  = \frac{k}{2}$.


\item {\bf (2) \bm{$2\mid n$}, \bm{$2\mid k$}, $\bm{k=m-1}$, and $\bm{\alpha\ge \frac 12}$.} $\ell_{\min} \le \frac{k}{2}$ is proved by constructing a graph $G\in\ug{n}=\uge$ whose subgraph over $\{1,\ldots, m-1\}$ is the same as the graph constructed in case (1) above, and then we add $m-1$ edges $i\ra m$ for each $i\le m-1$. $\ell_{\min} \ge \frac{k}{2}$ is proved by a similar argument as in case (1) above. More precisely, suppose for the sake of contradiction that there exists a graph $G'$ such that $\copeland(G')=k=m-1$ and $\ties(G')<\frac k2$. Then, there must exist an alternative, w.l.o.g.~alternative $1$, that is not tied to any other alternative in $G'$. This means that the Copeland$_\alpha$ score of alternative $1$ must be an integer, denoted by $u$.  Because $\alpha\ge \frac 12$, the total Copeland$_\alpha$ score is between $m\choose 2$ and ${m\choose 2} + \frac k2 -1$, and is also between $u k$ and $u(k+1)-1$. However, if $u\le \frac k2$, then $u(k+1)-1<{m\choose 2} = \frac{k(k+1)}{2}$; and if $u\ge \frac k2+1$, then $uk>{m\choose 2} + \frac k2 -1$. In either case $[{m\choose 2}, {m\choose 2} + \frac k2 -1]\cap [uk,u(k+1)-1] =\emptyset$, which is a contradiction. This means that $\ell_{\min} = \frac{k}{2}$.

\item {\bf (3) \bm{$2\mid n$}, \bm{$2\mid k$}, \bm{$k=m-1$}, and \bm{$k\le l_\alpha(l_\alpha+1)$}.} $\ell_{\min} \le \frac{k}{2}$ is proved by the same graph as in case (2) above. $\ell_{\min} \ge \frac{k}{2}$ is proved by a similar argument as in case (1) above. More precisely,  for the sake of contradiction suppose there exists a graph $G'$ such that $\copeland(G')=k=m-1$ and $\ties(G')<\frac k2$.  Then, there must exist a winner that is not tied to any other alternative in $G'$. Therefore, the Copeland$_\alpha$ score of the $k$ winners is an integer. However, in order for the Copeland$_\alpha$ score of any alternative $a$ who is tied in $G'$ (with any other alternative) to be an integer, $a$ must be tied with at least $l_\alpha$ alternatives. Among these $l_\alpha$ alternatives, at least $l_\alpha-1$ must be co-winners, each of which must be tied to at least $l_\alpha$ alternatives for their total Copeland$_\alpha$ scores to be integral. It is not hard to verify that the total number of ties in $G'$ is $\frac{l_\alpha(l_\alpha+1)}{2}\ge \frac k2$, which is a contradiction.  This means that $\ell_{\min} = \frac{k}{2}$.
\end{itemize}

\paragraph{\bf \boldmath The $\bm{\Theta(n^{-\frac{l_\alpha(l_\alpha+1)}{4}})}$ case: $\bm{2\mid n,  2\mid  k,  k=m-1, \alpha<\frac12,\text{ and }k> l_\alpha(l_\alpha+1)}$.} Like the $\Theta(n^{-\frac k4})$ case, we apply Claim~\ref{claim:copelandwinner} to prove that $\ell_{\min} = \frac{l_\alpha(l_\alpha+1)}{2}$. We first prove $\ell_{\min} \le  \frac{l_\alpha(l_\alpha+1)}{2}$ by explicitly constructing an unweighted directed graph $G\in\ug{n}=\uge$ such that the Copeland$_\alpha$ winners are $\{1,\ldots,k\}$ whose Copeland$_\alpha$ scores are $\frac k2$. Let $\ma = \ma_1\cup \ma_2\cup\{m\}$, where $\ma_1 = \{1,\ldots, l_{\alpha}\}$ and $\ma_1 = \{l_{\alpha}+1,\ldots, m-1\}$.  The construction of $G$ depends on the parity of $l_\alpha$, discussed in the following two cases.
\begin{itemize}
\item {\bf $\bm{l_\alpha}$ is odd.} We construct  $G\in\ug{n}=\uge$ in the following three steps,  illustrated in Figure~\ref{fig:copeland-alpha-odd}. Note that Figure~\ref{fig:copeland-alpha-odd} is only used for the purpose of illustration but does not corresponds to the $\Theta(n^{-\frac{l_\alpha(l_\alpha+1)}{4}})$ case, because $k = 8< 12 = l_\alpha(l_\alpha +1)$. A real example requires $k\ge 14$, which contains too many edges in $G$ to be easily visible.
\begin{itemize}
\item []{\bf Step 1.~Edges within $\bm{A_1\cup \{m\}}$ and edges within $\bm{A_2}$.} All alternatives in $A_1\cup\{m\}$ are tied, i.e., there is no edge between any pair of alternatives in $A_1\cup \{m\}$. It is not hard to verify that $|A_2|$ is an odd number. Therefore, we let the subgraph of $G$ over $A_2$ to be the completely graph as the one shown in Figure~\ref{fig:copeland} (a). After this step, the Copeland$_\alpha$ score of each alternative in $A_1\cup \{m\}$ becomes $\alpha l_\alpha$, and the Copeland$_\alpha$ score of each alternative in $A_2$ becomes $\frac{k-l_\alpha-1}{2}$.
\item []{\bf Step 2.~Edges between $\bm{A_1}$ and $\bm{A_2}$.} We assign all edges between $A_1$ and $A_2$ an direction such that each alternative in $A_1$ has exactly $\frac k2- \alpha l_\alpha$ outgoing edges, and each alternative in $A_2$ has $\frac{l_\alpha-1}{2}$ or $\frac{l_\alpha+1}{2}$ outgoing edges. This can be done by assigning edges from $A_1$ to $A_2$ evenly across $A_2$. In the graph that combines edges in Step 1 and Step 2, the Copeland$_\alpha$ score of each alternative in $A_1\cup \{m\}$ becomes exactly $\frac{k}{2}$, the Copeland$_\alpha$ score of $\frac{k-l_\alpha(l_\alpha+1)}{2} +\alpha l_\alpha^2$  alternatives in $A_2$ becomes $\frac{k}{2}$, and the Copeland$_\alpha$ score of the remaining $\frac{k+l_\alpha(l_\alpha+1)}{2} -\alpha l_\alpha^2$  alternatives in $A_2$ becomes $\frac{k}{2} -1$.
\item []{\bf Step 3.~Edges between $\bm{m}$ and $\bm{A_2}$.} Finally, the directions of edges in this step are used to guarantee the Copeland$_\alpha$ score of all alternatives in $A_2$ is $\frac k2$. 
\end{itemize}
Let $G$ denote the union of edges defined in Step 1-3. It is not hard to verify that the Copeland$_\alpha$ score of all alternatives in $A_1\cup A_2$ is $\frac k2$, and the Copeland$_\alpha$ score of alternative $m$ is $\alpha l_\alpha + \frac{k-l_\alpha(l_\alpha+1)}{2} +\alpha l_\alpha^2 < \frac k2$, because $\alpha<\frac 12$. This proves $\ell_{\min} \le  \frac{l_\alpha(l_\alpha+1)}{2}$ when $l_\alpha$ is odd.
\begin{figure}[htp]
\centering
\begin{tabular}{ccc}
    \includegraphics[width=.33\textwidth]{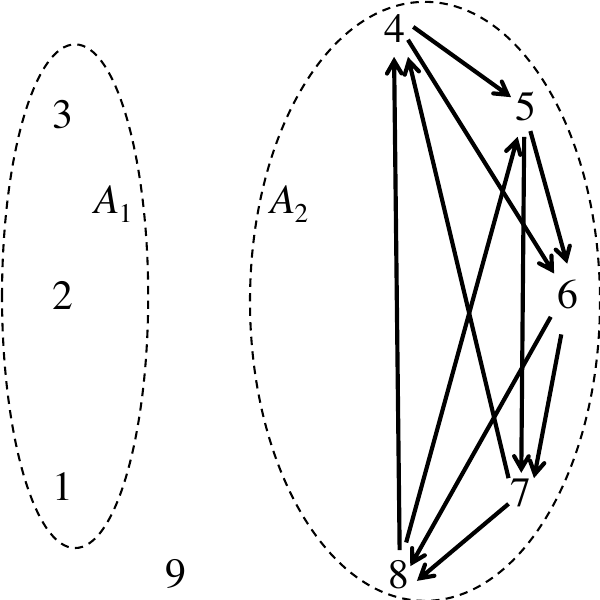} &    \includegraphics[width=.33\textwidth]{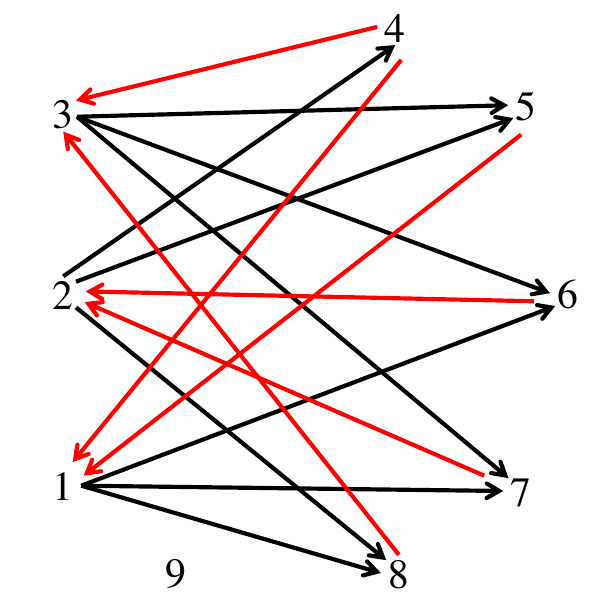} &    \includegraphics[width=.33\textwidth]{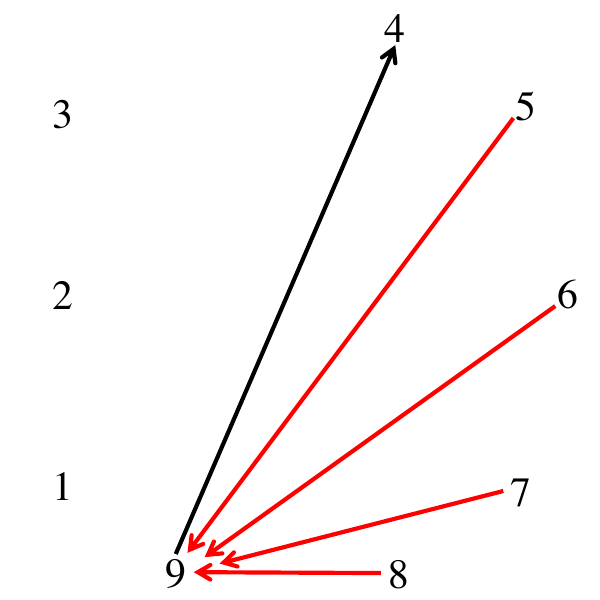}\\
    (a) Step 1. &     (b)   Step 2. &     (b)   Step 3.
\end{tabular}
    \caption{An illustration of the constructions for the $\Theta(n^{-\frac{l_\alpha(l_\alpha+1)}{4}})$ case with odd $l_\alpha$, where $m=9$, $k=8$, $\alpha = \frac 13$, and $l_\alpha =3$. Note that this figure is only used for illustrating the three-step construction. It does not correspond to the $\Theta(n^{-\frac{l_\alpha(l_\alpha+1)}{4}})$ case because $k<l_\alpha(l_\alpha+1)$. In (b) and (c), edges from right to left are colored red. \label{fig:copeland-alpha-odd}}
\end{figure}

\item {\bf $\bm{l_\alpha}$ is even.} Like the $2\nmid l_\alpha$ case, we construct  $G\in\ug{n}=\uge$ in the following three steps. 
\begin{itemize}
\item []{\bf Step 1.~Edges within $\bm{A_1\cup \{m\}}$ and edges within $\bm{A_2}$.} All alternatives in $A_1\cup\{m\}$ are tied, i.e., there is no edge between any pair of alternatives in $A_1\cup \{m\}$. It is not hard to verify that $|A_2|$ is an odd number. Therefore, we let the subgraph of $G$ over $A_2$ to be isomorphic to the (complete) subgraph over $\{1,\ldots,k\}$ as the one in Figure~\ref{fig:copeland} (b). After this step, the Copeland$_\alpha$ score of each alternative in $A_1\cup \{m\}$ becomes $\alpha l_\alpha$, the Copeland$_\alpha$ score of $\frac{k-l_\alpha}{2}$ alternative in $A_2$ becomes $\frac{k-l_\alpha}{2}$, and the Copeland$_\alpha$ score of the remaining $\frac{k-l_\alpha}{2}$ alternative in $A_2$ becomes $\frac{k-l_\alpha}{2}-1$.
\item []{\bf Step 2.~Edges between $\bm{A_1}$ and $\bm{A_2}$.} We assign all edges between $A_1$ and $A_2$ an direction such that each alternative in $A_1$ has exactly $\frac k2- \alpha l_\alpha$ outgoing edges, and each alternative in $A_2$ has $\frac{l_\alpha}{2}-1$ or $\frac{l_\alpha}{2}$ outgoing edges. This can be done by assigning edges from $A_1$ to $A_2$ uniformly across $A_2$, starting with alternatives whose Copeland$_\alpha$ scores are $\frac{k-l_\alpha}{2}$ in Step 1. In the graph that combines edges in Step 1 and Step 2, the Copeland$_\alpha$ score of each alternative in $A_1\cup \{m\}$ becomes exactly $\frac{k}{2}$, the Copeland$_\alpha$ score of $\frac{l_\alpha^2}{2} +\alpha l_\alpha$  alternatives in $A_2$ becomes $\frac{k}{2}$, and the Copeland$_\alpha$ score of the remaining $k-l_\alpha - \frac{l_\alpha^2}{2} -\alpha l_\alpha$ (which is strictly positive because $k>l_\alpha(l_\alpha+1)$ and $l_\alpha\ge 4$ because $\alpha<\frac 12$)  alternatives in $A_2$ becomes $\frac{k}{2} -1$.
\item []{\bf Step 3.~Edges between $\bm{m}$ and $\bm{A_2}$.} Finally, the directions of edges in this step are used to guarantee the Copeland$_\alpha$ score of all alternatives in $A_2$ is $\frac k2$. 
\end{itemize}
Let $G$ denote the union of edges defined in Step 1-3. It is not hard to verify that the Copeland$_\alpha$ score of all alternatives in $A_1\cup A_2$ is $\frac k2$, and the Copeland$_\alpha$ score of alternative $m$ is $\frac{l_\alpha^2}{2} +\alpha l_\alpha + \alpha l_\alpha\le \frac{l_\alpha^2}{2} + 2 \le \frac {l_\alpha^2+l_\alpha}{2} < \frac k2$, because $l_\alpha \ge 4$ and $k>l_\alpha(l_\alpha+1)$. This proves $\ell_{\min} \le  \frac{l_\alpha(l_\alpha+1)}{2}$ when $l_\alpha$ is even.
\end{itemize}

We now prove that $\ell_{\min} \ge  \frac{l_\alpha(l_\alpha+1)}{2}$. The proof is similar to the proof of the $\Theta(n^{-\frac k4})$ case (3). More precisely,  for the sake of contradiction that there exists a graph $G'$ such that $\copeland(G')=k=m-1$ and $\ties(G')<\frac{l_\alpha(l_\alpha+1)}{2}<\frac k2$.  Then, there must exist a winner that is not tied to any other alternative in $G'$. Therefore, the Copeland$_\alpha$ score of  the $k$ winners is an integer. However, in order for the Copeland$_\alpha$ score of any alternative $a$ who is tied in $G'$ (with any other alternative) to be an integer, following the same argument as in the proof of the $\Theta(n^{-\frac k4})$ case (3), we must have $\ties(G')\ge \frac{l_\alpha(l_\alpha+1)}{2}$, which is a contradiction. This means that $\ell_{\min} =\frac{l_\alpha(l_\alpha+1)}{2}$.

This completes the proof for Proposition~\ref{prop:copeland}.
\end{proof}

\subsection{Proof of Proposition~\ref{prop:maximin}}
\label{app:proof-mm}

\appProp{prop:maximin}{Max smoothed likelihood of ties: maximin}{
Let $\mm= (\Theta,\ml(\ma),\Pi)$ be a strictly positive and closed single-agent preference model where $\piuni\in \conv(\Pi)$. For any $2\le k\le m$ and any $n\in\mathbb N$, 
$$\slt{\Pi}{\maximin}{m}{k}{n}=\Theta(n^{-\frac{k-1}{2}})$$
}
\begin{proof} The proposition is proved by applying Theorem~\ref{thm:eorule}. 

\paragraph{\bf The $\bm 0$ case and the exponential case.} We first prove that when $n>m^4$, the $0$ case and the exponential case of Theorem~\ref{thm:eorule} do not occur. By Proposition~\ref{prop:eo} and notice that $\eo(\piuni)$ only contains $T_0$, it suffices to prove that for any $k$, there exists $O\in \palo$ such that $|\mm(O)|=k$. This is done by first defining a directed weighted graph $G$, and then let $O=\eo(G)$. Let $G$ denote the following directed weighted graph over $\ma$, as illustrated in Figure~\ref{fig:mm}.
\begin{figure}[htp]\centering
  \includegraphics[width=.4\textwidth]{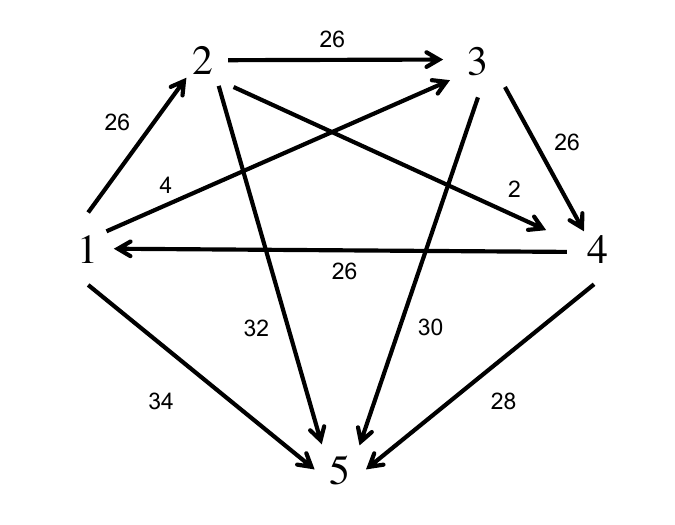}     \caption{Construction for maximin where $m=5$ and $k=4$.}
\label{fig:mm}
\end{figure}
\begin{itemize}
\item $G$ contains the following cycle of $k$ edges, the weight of each of which is $m^2+1$.
$$E=\{1\ra 2, 2\ra 3,\ldots, (k-1)\ra k, k\ra 1\}$$
\item The other edges within $\{1,\ldots k\}$ have different weights, whose parity is the same as $m^2+1$ and whose absolute values are strictly smaller than $m^2+1$ and strictly larger than $0$.
\item The absolute values of any other edge has the same parity as $m^2+1$ and are strictly larger than $m^2+1$. In addition, for any $1\le i\le k$ and any $k+1\le j\le m$, the  weight on $i\ra j$ is positive.
\end{itemize}
It is not hard to check that $\maximin(G)=\{1,\ldots,k\}$. Let $O=\eo(G)$. Because there is no tie in $G$, the middle tier $T_0$ in $\eo(G)$ is empty, which means that $O\in \palo\subseteq \pal{n}$.

\paragraph{\bf The polynomial case.}   We first prove  that $\ell_{\min} \le k-1$ by using $O$ defined above for the $0$ case and exponential case. There are only two tiers with more than one elements, namely the tiers corresponding to $E$ and the inverse $E$, which means that $\ties(O)= k-1$. 

We now prove that $\ell_{\min}\ge k-1$. For any $O\in \mo_{\cor,k,n}^{\Pi}$, suppose w.l.o.g.~that the maximin winners are $\{1,\ldots,k\}$. For each $1\le i\le k$, let $i\ra a_i$ denote an arbitrary edge that corresponds to the min score of $i$. Because $\{1,\ldots,k\}$ have the same min score, $\{i\ra a_i:1\le i\le k\}$ must be in the same tier. If they are in $T_0$, then $T_0$ contains at least $2k$ edges $\{i\ra a_i, a_i\ra i:1\le i\le k\}$, which means that $\ties(O)\ge |T_0|/2\ge k$. If they are not in $T_0$, then $\ties(O)\ge k-1$. In either case we have $\ell_{\min}\ge k-1$. 

Therefore, $\ell_{\min}=k-1$. This completes the proof for Proposition~\ref{prop:maximin}.
\end{proof}

\subsection{Proof of Proposition~\ref{prop:schulze}}
\label{app:proof-schulze}

\appProp{prop:schulze}{Max smoothed likelihood of ties: Schulze}{
Let $\mm= (\Theta,\ml(\ma),\Pi)$ be a strictly positive and closed single-agent preference model  with $\piuni\in \conv(\Pi)$. For any $2\le k\le m$ and any $n\in\mathbb N$,  
$$\slt{\Pi}{\schulze}{m}{k}{n}=\Theta(n^{-\frac{k-1}{2}})$$}
\begin{proof} The proposition is proved by applying Theorem~\ref{thm:eorule}. 

\paragraph{\bf The $\bm 0$ case and the exponential case.} When $n>m^4$, the $0$ case and the exponential case of Theorem~\ref{thm:eorule} do not occur. This is prove by the same graph $G$ as in the proof of Proposition~\ref{prop:maximin} that is illustrated in Figure~\ref{fig:mm}. 

\paragraph{\bf The polynomial case.}   $\ell_{\min} \le k-1$ is proved by using $O=\eo(G)$ for the same $G$ used in Proposition~\ref{prop:maximin}.  

We now prove that $\ell_{\min}\ge k-1$. For any $O\in \mo_{\cor,k,n}^{\Pi}$, suppose w.l.o.g.~that the Schulze winners are $A_1=\{1,\ldots,k\}$. We define the following $t$ unweighted directed graphs. For any $1\le i\le t+1$, let $G_i$ denote the graph over $\ma$ whose edges are $T_1\cup\ldots\cup T_{t+1-i}$. Let $G_{t+1}$ denote the graph without any edges. By definition $G_1 = \umg(O)$. Given any unweighted directed graph $G'$, a node $a$ is {\em undominated}, if for any node $b$ that there exists a directed path from $b$ to $a$, there exists a directed path from $a$ to $b$.

\begin{claim}
\label{claim:undominated}
For each $1\le i\le t$, $\{1,\ldots, k\}$ are undominated in $G_i$.
\end{claim}
\begin{proof}
Suppose for the sake of contradiction that the claim is not true and w.l.o.g.~alternative $1$ is dominated by $j$ in $G_i$ for some $1\le i\le t$. Let $w$ denote the weight of the edges in $T_{t+1-i}$. This means that $s[j,1]\ge w$ and in each path from $1$ to $j$ there is an edge whose weight is strictly less than $w$. Therefore, $s[j,1]>s[1,j]$, which contracts the assumption that $1$ is a winner.
\end{proof}

For each $1\le i\le t$, let $\calS_i\subset 2^{A_1}$ denote the connected components in $G_i$ that  intersect $A_1=\{1,\ldots, k\}$ and let $\eta_i = |\calS_i|$. By Claim~\ref{claim:undominated}, each set $S$ in $\calS_i$ has no incoming edges from $\ma\setminus S$ in $G_i$, otherwise $S\cap A_1$ are dominated in $G_i$. If $\eta_1 >1$, then  we have $|T_0|\ge 2(k-1)$, which means that $\ties(O)\ge k-1$ and therefore $\ell_{\min}\ge k-1$, which proves the proposition. In the rest of the proof we assume that $\eta_1 = 1$.

For each $1\le i\le t$, because $G_{i+1}$ is obtained from $G_i$ by removing some edges, each set in $\calS_{i+1}$ must be a subset of a set in $\calS_i$.  If $\eta_{i+1}>\eta_i$, then there are at least $\eta_{i+1}-\eta_i+1$ sets in $\calS_{i+1}$, denoted by $\calS_{i+1}'$, that are strict subsets of some sets in $G_i$. 
For any $S'_{i+1}\in \calS_{i+1}'$, let $S'_i\in \calS_i'$ be the set such that $S'_{i+1}\subsetneq S'_i$. 
By Claim~\ref{claim:undominated}, $S'_{i+1}$ does not contain any incoming edges from $\ma\setminus S'_{i+1}$ in $G_{i+1}$, and in particular, no edges from $S'_i\setminus S'_{i+1}$ to $S_{i+1}'$ (which are in $G_i$ and there is at least one such edge because $S_i'$ is a connected component in $G_i$) are in $G_{i+1}$. This means that $T_{t+1-i}$ contains at least one incoming edge to each set in $\calS_{i+1}'$, and such sets of edges are non-overlapping because $\calS_{i+1}'$ consists of non-overlapping subsets of $\ma$. Therefore, $|T_{t+1-i}|\ge \eta_{i+1}-\eta_i+1$. If $\eta_{i+1}=\eta_i$, then apparently we have $|T_{t+1-i}|\ge \eta_{i+1}-\eta_i+1 = 1$. Recall that $\eta_1 = 1$ and $\eta_{t+1} = k$, we have:
$$\ties(O) = \sum_{i=1}^t(|T_{t+1-i}|-1) + |T_0|/2 \ge \sum_{i=1}^t(\eta_{i+1}-\eta_i) = \eta_{t+1} - \eta_1=k-1$$
Again, we have $\ell_{\min}\ge k-1$.

This completes the proof for Proposition~\ref{prop:schulze}.
\end{proof}

\subsection{Proof of Proposition~\ref{prop:rp}}
\label{app:proof-rp}

\appProp{prop:rp}{Max smoothed likelihood of ties: ranked pairs}{
Let $\mm= (\Theta,\ml(\ma),\Pi)$ be a strictly positive and closed single-agent preference model  with $\piuni\in \conv(\Pi)$. For any $2\le k\le m$ and any $n\in\mathbb N$, 
$$\Omega(n^{-\frac{k-1}{2}}) \le \slt{\Pi}{\rp}{m}{k}{n} \le n^{-\Omega(\frac{\log k}{\log\log k})}$$
Moreover, when $m\ge  k+5\lceil \log k\rceil$, 
$$\slt{\Pi}{\rp}{m}{k}{n}\le \Omega(n^{-\frac{\lceil \log k\rceil}{2}})$$
When $k=2$, 
$$\slt{\Pi}{\rp}{m}{2}{n}=\Theta(n^{-0.5})$$
}
\begin{proof} 
The proposition is proved by applying Theorem~\ref{thm:eorule}. It is not hard to verify that the $0$ case and the exponential case of Theorem~\ref{thm:eorule} do not occur when $n$ is sufficiently large.  Let $s_k$ denote the largest integer such that $s_k!\le k$. It follows that $s_k = O(\frac{\log k}{\log \log k})$.  To prove the first inequality in the statement of the proposition, it suffices to prove that $s_k-1\le \ell_{\min} \le k-1$. 

We first prove that $\ell_{\min} \le k-1$ by constructing a weighted directed graph $G$ such that $\ties(\eo(G)) = k-1$, $|\rp(G)|=k$, and $\eo(G)\in \palo\subseteq \pale$. In fact, we can use the graph in the proof of Proposition~\ref{prop:maximin} illustrated in Figure~\ref{fig:mm}. 
The winner is determined by the tie-breaking order among edges in the cycle $1\ra 2\ra\cdots\ra k\ra 1$.

We now  prove that $\ell_{\min} \ge s_k-1$.  Let $O\in \mo_\ma^{n}$ denote the palindromic order such that $\ties(O) = \ell_{\min}$. The total number of ways to break ties in ranked pairs is no more than $(\ell_{\min}+1)!$, and the maximum is achieved when a tier before $T_0$ contains all $l+1$ tied edges. Therefore, we must have $(\ell_{\min}+1)!\ge k$, which means that  $\ell_{\min}\ge s_k-1$. 

\paragraph{\bf The ``moreover'' part.} This part is proved by construction. When $m\ge 5\lceil \log k\rceil + k$, we  construct a weighted directed graph $G$ such that (1) $\eo(G)\in\palo$, (2) $|\rp(G)|=k$, and (3) $\ties(\eo(G))=\lceil \log k\rceil$. The construction is illustrated in Figure~\ref{fig:rp} (b).  We let $\ma$ denote the union of the following three sets of alternatives.
\begin{itemize}
\item $\{0,\ldots, k-1\}$. These are the ranked pairs winners. Notice that the index starts at $0$ instead of $1$.
\item For each $i\le \lceil \log k\rceil$, there are five alternatives $\{a_i,b_i,c_i,d_i,f_i\}$.
\item Let $A$ denote the remaining $m-k-5\lceil \log k\rceil$ alternatives.
\end{itemize}
There are $\lceil \log k\rceil$ pairs of edges with same weights in $G$: for any $i\le \lceil \log k\rceil$, $w_G(a_i\ra c_i)=w_G(b_i\ra d_i)$.
The edges are divided into the following groups such that weights in an earlier group are higher than weights in an later group. We further require that no pair of edges have the same weight except $a_i\ra c_i$ and $b_i\ra d_i$, and in particular, no edge has weight $0$.
\begin{itemize}
\item Group 1. For each $i\le \lceil \log k\rceil$, we have the following edges: $c_i\ra b_i$, $d_i\ra a_i$, $c_i\ra f_i$, $d_i\ra f_i$.  For each $0\le j\le k-1$, there is an edge from $j$ to each alternative in $A$. For each $0\le j\le k-1$ and each $i\le \lceil \log k\rceil$, if the $i$-th digit from right in $j$'s binary representation is $0$, then there is an edge $j\ra a_i$; otherwise there is an edge $j\ra b_i$. 
\item Group 2.  For each $i\le \lceil \log k\rceil$ and each $0\le j\le k-1$, there is an edge $f_i\ra j$.
\item Group 3. All remaining edges are in Group 3.
\end{itemize} 

\begin{figure}[htp]
\centering
    \includegraphics[width=.4\textwidth]{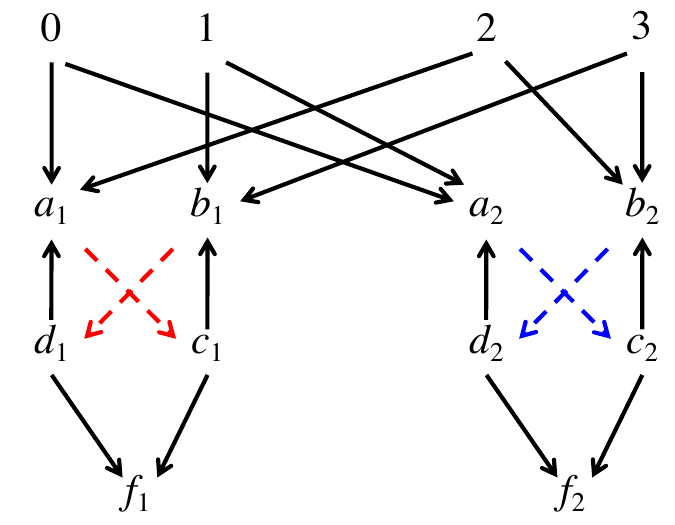} \\
    \caption{Construction for ranked pairs. $a_i\ra c_i$ and $b_i\ra d_i$ have the same weights. The figure only shows Group 1 edges.}
\label{fig:rp}
\end{figure}
For each $i\le \lceil \log k\rceil$, if $a_i\ra c_i$ is fixed before (respectively, after) $b_i\ra d_i$, then there is a path from each $0\le j\le k-1$ whose $i$-th digit from right in binary is $0$ (respectively, $1$) to $f_i$. Therefore, after edges in Group 1 are fixed by breaking ties between $\{a_i\ra c_i, b_i\ra d_i\}$ for all $i\le \lceil \log k\rceil$, there is a unique alternative $0\le j\le k-1$ that has a path to all $f_i$'s. Then, each alternative in $\{0,\ldots, k-1\}$ except $j$ will have an incoming edge from at least one $f_i$. It follows that $j$ will be the ranked pairs winner. We note that each alternative in $\{0,\ldots, k-1\}$ can be made to win by breaking ties in some way. 

Therefore, in $\eo(G)$ we have $T_0=\emptyset$ and there are $\lceil \log k\rceil$ tiers containing more than one edge, each of which contain exactly two edges. This means that $\ties(\eo(G))=\lceil \log k\rceil$, which means that  $\ell_{\min}\le \lceil \log k\rceil$. 

\paragraph{\bf \boldmath When $k=2$.} The $\Omega(n^{-0.5})$ lower bound has already been proved for general $k$. Therefore, it suffices to prove $\ell_{\min} \ge 1$. Suppose for the sake of contradiction that this is not true. Then, there exists a palindromic order $O$ such that $\ties(O)<1$ and $|\rp(O)|=2$. However, $\ties(O)=0$ means that no edges are tied in $\wmg(P)$. Therefore, $|\rp(O)|=1\ne 2$, which is a contradiction.  

This completes the proof of Proposition~\ref{prop:rp}.
\end{proof}

\section{Appendix for Section~\ref{sec:mrerules}: STV and Coombs}
\label{app:mrse}

\subsection{Formal  Definitions and Statements of Results}

We first formally define multi-round score-based elimination (MRSE) rules, which are generalizations of STV and Coombs.
\begin{dfn}
\label{dfn:mrerule}
A {\em multi-round score-based elimination (MRSE)}  rule $\cor$ for $m$ alternatives is defined as a vector of  $m-1$ rules $(\cor_2,\ldots, \cor_{m})$ such that for each $2\le i\le m$, $\cor_i$ is an integer positional scoring voting rule over $i$ alternatives that outputs a {\em total preorder} over them according to their scores. Given a profile $P$ over $m$ alternatives, $\cor(P)$ is chosen in $m-1$ steps. For each $1\le i\le m-1$, in round $i$ a loser (alternative in the lowest tier) under $\cor_{m+1-i}$ is eliminated from the election. An alternative $b$ is a winner if there is a way to break ties so that $b$ is the remaining alternative after $m-1$ rounds of elimination.
\end{dfn}
We now define the counterpart to palindromic orders in Section~\ref{sec:eorules}, called {\em PUT structures}, that contain all information needed to determine the winners for MRSE rules.
\begin{dfn}[\bf PUT structure]
\label{dfn:PUT}
A {\em PUT structure} over $\ma$ is a mapping $\puts$ that maps each $B\subsetneq \ma$ to a total pre-order over $\ma\setminus B$. Let $\calW_{\ma}$ denote the set of all PUT structures over $\ma$.

 For any MRSE rule $\cor$ and any profile $P$, let $\pc{\cor}{P}$ denote the PUT structure $\puts$ such that for any $B\subsetneq \ma$, $\puts(B)=\cor_{m-|B|}(P|_{\ma\setminus B})$ is a total preorder over $\ma\setminus B$, where $P|_{\ma\setminus B}$ is the profile over $\ma\setminus B$ that is obtained from $P$ by removing all alternatives in $B$.
\end{dfn}
We note that $\pc{\cor}{P}$ depends on both $\cor$ and $P$. A PUT structure $\puts$ can be equivalently represented by a {\em PUT graph}, where the nodes are non-empty subsets of $\ma$, each node $B$ is labeled by $\puts(B)$, and there is an edge between node $B$ and $B'$ if $B' = B\cup\{a\}$ and $a$ is in the lowest tier of $\puts(B)$. With a little abuse of notation, we will also use $W$ to denote its PUT graph.

The {\em major component} of $\puts$ is its subgraph whose nodes are reachable from $\emptyset$, denoted by $\mc(\puts)$. For example, a profile $P$ and the major component of $\pc{\stv}{P}$ are shown in Figure~\ref{fig:PUTtree}. We note that the winners only depend on the major  component of the PUT structure of the profile.

\begin{dfn}[\bf Tier representation and refinement of PUT structures]
For any PUT structure $\puts\in\calW_\ma$ and any $B\subsetneq \ma$, let $T_1\succ \cdots\succ T_s$ denote the {\em tier representation} of $\puts(B)$, where for each $i\le s$, alternatives in $T_i$ are tied. Let $\ties(\puts(B)) = \sum_{i\le s}(|T_i|-1)$ and let $\ties(\puts)=\sum_{B\subsetneq \ma} \ties(\puts(B))$. 

A PUT structure $\puts_1$ {\em refines} another PUT structure $\puts_2$, if for all $B\subsetneq \ma$, $\puts_1(B)$ refines $\puts_2(B)$ in the total preorder sense.
\end{dfn}

\begin{ex}Figure~\ref{fig:PUTtree} shows a profile $P$ and the PUT graph of $\pc{\stv}{P}$. $\mc(\pc{\stv}{P})$ is the subgraph whose nodes are $\{\emptyset, \{1\}, \{2\}, \{1,2\}, \{1,3\}, \{2,3\}, \{1,2,4\}, \{1,3,4\}, \{2,3,4\}\}$. Let $\puts=\pc{\stv}{P}$. We have $\ties(\puts(\emptyset))=\ties(\puts(\{1\}))=\ties(\puts(\{2\}))=1$, and $\ties(\puts(B))=0$ for any other $B\subsetneq \ma$,  which means that $\ties(\puts)=3$. 

\begin{figure}[htp]
\centering
\begin{tabular}{cc}
\begin{minipage}[t]{0.4\linewidth}
\centering
\begin{tabular}{|c|c|}
\hline  preferences & multiplicity \\
\hline $1\succ 3\succ 2\succ 4$ & $1$ \\
\hline $1\succ 2\succ 3\succ 4$ & $2$ \\
\hline $2\succ 1\succ 3\succ 4$& $3$ \\
\hline $3\succ 2\succ 1\succ 4$& $4$ \\
\hline $4\succ 1\succ 2\succ 3$& $6$ \\
\hline
\end{tabular}
\end{minipage}
&
\begin{minipage}{0.55\linewidth}
\centering
\includegraphics[width = \textwidth]{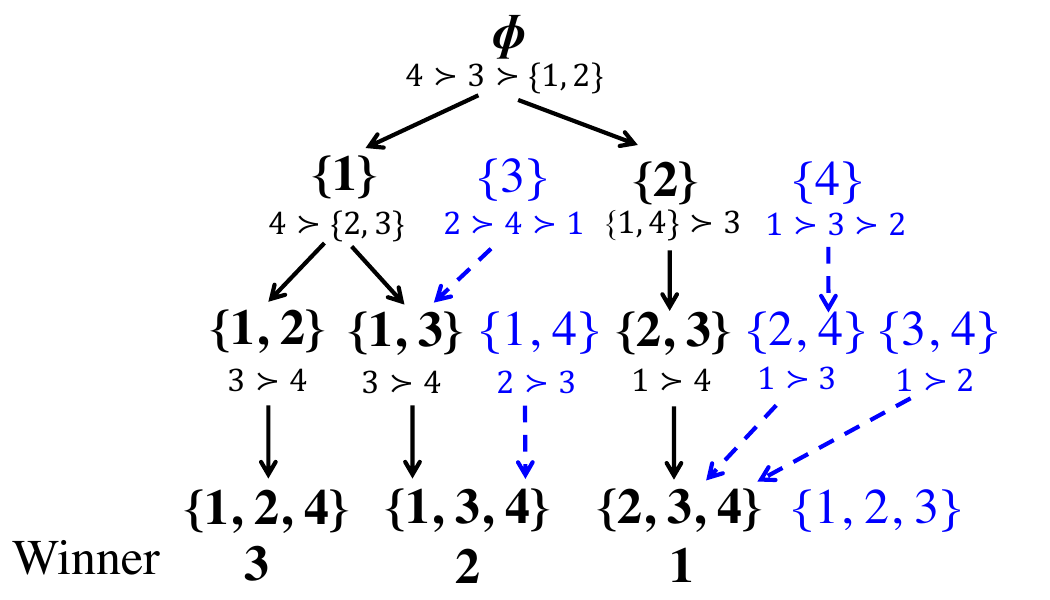}
\end{minipage}

\\
(a) A profile $P$. & (b) PUT graph of $\pc{\stv}{P}$ and its major component.
\end{tabular}
\caption{A preference profile $P$, its PUT structure $\pc{\stv}{P}$ under STV and the corresponding PUT graph. 
The major component consists of the subgraph with the nine nodes in black. 
\label{fig:PUTtree}}
\end{figure}
\end{ex}
Next, we define notation that will be used to present the theorem on the smoothed likelihood of ties under STV.

For any weak order $O$ with tier representation $T_1\succ \cdots \succ T_t$, where elements in any $T_i$ are tied, we let $\GCD{O}$ denote the greatest common divisor of the sizes of the $t$ tiers in $O$. That is,
$$\GCD{O}= \gcd (|T_1|,\ldots,|T_t|)$$
\begin{ex}
Let $W$ denote the PUT structure in Figure~\ref{fig:PUTtree} (b). We have $\GCD{W(\emptyset)} = \GCD{3\succ 4\succ \{1,2\}} =\gcd(1,1,2)=1$. As another example, $\GCD{\{3,4\}\succ \{1,2\}}=\gcd(2,2)=2$.
\end{ex}
The next lemma states that for any sufficiently large $n$, $W$ can be represented by the PUT structure under STV or Coombs of some $n$-profile if and only if $n$ is divisible by $\GCD{W(B)}$ for all $B\subsetneq \ma$.

\begin{lem}[\bf PUT structures under STV and Coombs]
\label{lem:construction-STV-Coombs} 
Let $\cor\in\{\stv,\coombs\}$. For any $m\ge 3$, there exists $N\in\mathbb N$ such that for any PUT structure $W$ over $\ma$ and any $n\ge N$, 
$$\left(\exists P\in\ml(\ma)^n\text{ s.t. } \pc{\cor}{P}=W \right) \Longleftrightarrow \left(\forall B\subsetneq \ma, \GCD{W(B)}\mid n\right)$$
\end{lem}
The proof can be found in Appendix~\ref{app:stv-construction-proof}. In light of Lemma~\ref{lem:construction-STV-Coombs}, for any $\ma$ and any $n\in\mathbb N$, we define  $\calW_n$ to be the PUT structures over $\ma$ such that for every $B\subsetneq \ma$, $\gcd(W(B))\mid n$. 
\begin{dfn}
Let $\cor\in\{\stv,\coombs\}$.  Given any $\ma$, any $n\in\mathbb N$,   any $2\le k\le m$, and any distribution $\pi$ over $\ml(\ma)$, let
$$\calW_{\cor,k,n}^{\pi} = \{\puts\in\calW_n: |\cor(\puts)|=k\text{ and }\puts\text{ refines }\pc{\cor}{\pi}\}$$
 For any set of distributions $\Pi$ over $\ml(\ma)$, let $\calW_{\cor,k,n}^{\Pi} = \bigcup_{\pi\in\conv(\Pi)}\calW_{\cor,k,n}^{\pi}$.  When $\calW_{\cor,k,n}^{\Pi}\ne \emptyset$, we let 
 $$w_{\min} = \min\nolimits_{\puts\in  \calW_{\cor,k,n}^{\Pi}}\ties(\puts).$$ 
 When $\calW_{\cor,k,n}^{\pi}\ne \emptyset$ for all $\pi\in\conv(\Pi)$, we let 
 $$w_{\text{mm}} = \max\nolimits_{\pi\in\conv(\Pi)}\min\nolimits_{\puts\in  \calW_{\cor,k,n}^{\pi}}\ties(\puts)$$
\end{dfn}

In words, $\calW_n$ consists of all PUT structures resulted from all $n$-profiles under $\cor\in \{\stv,\coombs\}$.  $\calW_{\cor,k,n}^{\pi}$ is the set of PUT structures $\puts$ that satisfies three conditions:  (1) $\puts$ is the PUT structure of an $n$-profile under $\cor$; (2) there are exactly $k$ winners in $\puts$ according to $\cor$, and (3) $\puts$ refines $\pc{\cor}{\pi}$. $\calW_{\cor,k,n}^{\Pi}$ is the union of all PUT structures corresponding to all $\pi\in \conv(\Pi)$. $w_{\min}$ is the minimum number of ties in PUT structures in $\calW_{\cor,k,n}^{\Pi}$. $w_{\text{mm}}$ is the maximin number of ties, where the maximum is taken for all $\pi\in \conv(\Pi)$, and for any given $\pi$, the minimum is taken for all PUT structures in $\calW_{\cor,k,n}^{\pi}$.

We note that $w_{\min}$ and $w_{\text{mm}}$ depend on  $\cor\in \{\stv,\coombs\}$, $\Pi$,  $k$, and $n$, which are clear from the context. In fact, $w_{\min}$ and $w_{\text{mm}}$ correspond to $m!-\alpha_n$ and $m!-\beta_n$ in Theorem~\ref{thm:union-poly}, respectively, and we have the following theorem for STV and Coombs.


\appThm{thm:STV-Coombs}{\bf Smoothed likelihood of ties: STV and Coombs}{Let $\cor\in\{\stv,\coombs\}$ and let $\mm= (\Theta,\ml(\ma),\Pi)$ be a strictly positive and closed single-agent preference model. For any $2\le k\le m$ and any $n\in\mathbb N$,  
\begin{align*}
&\slt{\Pi}{\cor}{m}{k}{n}=\left\{\begin{array}{ll}0 &\text{if } \forall W\in \calW_n, |\cor(W)|\ne k
\\
\exp(-\Theta(n)) &\text{otherwise if } \calW_{\cor,k,n}^{\Pi} = \emptyset\\
\Theta\left(n^{-\frac{w_{\min}}{2}}\right) &\text{otherwise}
\end{array}\right.\\
&\ilt{\Pi}{\cor}{m}{k}{n}=\left\{\begin{array}{ll}0 &\text{if } \forall W\in \calW_n, |\cor(W)|\ne k\\
\exp(-\Theta(n)) &\text{otherwise if } \exists \pi\in\conv(\Pi) \text{ s.t. }\calW_{\cor,k,n}^{\pi} = \emptyset\\
\Theta\left(n^{-\frac{w_{\text{mm}}}{2}}\right) &\text{otherwise}
\end{array}\right.
\end{align*}
}
The proof can be found in Appendix~\ref{appendix:proof-mre}. 
Like in Section~\ref{sec:eorules}, we next  characterize max smoothed likelihood of ties for STV and Coombs  for   distributions $\Pi$  where $\piuni\in \conv(\Pi)$. 



\appProp{prop:stv-Coombs}{\bf Max smoothed likelihood of ties: STV and Coombs}
{Let $\cor\in\{\stv,\coombs\}$ and  let $\mm= (\Theta,\ml(\ma),\Pi)$ be a strictly positive and closed single-agent preference model  with $\piuni\in \conv(\Pi)$. For any $2\le k\le m$ and any $n\in\mathbb N$,   
$$\slt{\Pi}{\cor}{m}{k}{n}=\left\{\begin{array}{ll}\Theta(n^{-\frac{k-1}{2}})&\text{if }\left\{\begin{array}{l} (1) m\ge 4 \text{, or } \\
(2) m=3\text{ and }k=2, \text{ or }\\
(3) m=k=3\text{ and }(2\mid n \text{ or } 3\mid n)\end{array}\right.\\
0 &\text{otherwise (i.e., }m=k=3\text{,  }2\nmid n,\text{ and }3\nmid n\text{)}\end{array}\right.
$$
}

The proof of Proposition~\ref{prop:stv-Coombs}  can be found in Appendix~\ref{app:proof-stv}.

\subsection{Proof of Lemma~\ref{lem:construction-STV-Coombs}}
\label{app:stv-construction-proof}

\appLem{lem:construction-STV-Coombs}{PUT structures under STV and Combs}{
Let $\cor\in\{\stv,\coombs\}$. For any $m\ge 3$, there exists $N\in\mathbb N$ such that for any PUT structure $W$ over $\ma$ and any $n\ge N$, 
$$\left(\exists P\in\ml(\ma)^n\text{ s.t. } \pc{\cor}{P}=W \right) \Longleftrightarrow \left(\forall B\subsetneq \ma, \GCD{W(B)}\mid n\right)$$
}
\begin{proof} We first prove the Lemma for $\cor = \stv$, then comment on how to modify the proof for $\coombs$.

\paragraph{\bf Proof for STV.} Notice The $\Rightarrow$ direction for STV follows after noticing that the total plurality score of all alternatives after $B$ is removed, which is $n$, is divisible by $\GCD{W(B)}$. More precisely, for any $n$-profile $P$ such that $\pc{\stv}{P}=W$ and any $B\subsetneq \ma$, let the tier representation of $W(B)$ be: 
$$W(B)=T_1\succ T_2\succ \cdots \succ T_t$$ 
For every $s\le t$, let $\gamma_s$ denote the plurality score of any alternative in $T_s$ after $B$ is removed. That is,
$$\gamma_s = \score_{\plu}(P|_{\ma\setminus B}, a_s)  \text{ for any }a_s\in T_s$$
Therefore, we have 
$$n = \sum_{a\in (\ma\setminus B)}\score_{\plu}(P|_{\ma\setminus B}, a)= \sum_{1\le s\le t}\gamma_s|T_s|,$$
which is divisible by $\GCD{W(B)}$.

The $\Leftarrow$ direction is proved in the following steps.

\paragraph{\bf \boldmath Step 1 for $\Leftarrow$.} For any $B\subseteq \ma$  and any pair of alternatives $a,b\in \ma\setminus B$, in this step we define an $m!$-profile $P_{B,a,b}$ as the building block.

\begin{dfn}[\bf\boldmath  $P_{B,a,b}$]
\label{dfn:P-a-b}
For any pair of different alternatives $a,b$ and any $B\subsetneq (\ma\setminus\{a,b\})$, $P_{B,a,b}$ is obtained from $\ml(\ma)$ as follows. For each $B^*$ such that $B\subseteq B^*\subseteq (\ma\setminus\{a,b\})$,
\begin{itemize}
\item If $|B^*\setminus B|$ is an even number, then we replace $[B^*\succ b\succ a\succ \others]$  by $[B^*\succ a\succ b\succ \others]$, where ``$\others$'' represents the other alternatives, and alternatives in $B^*$ and ``$\others$'' are ranked w.r.t.~the lexicographic order.
\item If  $|B^*\setminus B|$ is an odd number, then we replace $[B^*\succ a\succ b\succ \others]$ by $[B^*\succ b\succ a\succ \others]$.
\end{itemize}
\end{dfn}

An example is shown in Table~\ref{tab:stv-construction}, where only the differences between $\ml(\ma)$ and $P_{\{1\},2,3}$ are shown.

\begin{table}[htp]
\centering
\begin{tabular}{|c|c|c|c|c|}
\hline
$B^*$ & $\{1\}$ & $\{1,4\}$ & $\{1,5\}$ & $\{1,4,5\}$\\
\hline $\ml(\ma)$& $1\succ 3\succ 2\succ 4\succ 5$& $1\succ 4\succ 2\succ 3\succ 5$&  $1\succ 5\succ 2\succ3 \succ 4$&  $1\succ 4\succ 5\succ 3\succ 2$\\
\hline $P_{\{1\},2,3}$& $1\succ 2\succ 3\succ 4\succ 5$& $1\succ 4\succ 3\succ 2\succ 5$&  $1\succ 5\succ 3\succ2 \succ 4$&  $1\succ 4\succ 5\succ 2\succ 3$\\
\hline 
\end{tabular}
\caption{Differences between $\ml(\ma)$ and $P_{\{1\},2,3}$ for $m=4$, $B=\{1\}$, $a=2$, and $b=3$. \label{tab:stv-construction}}
\end{table}

We now prove that $P_{B,a,b}$  satisfies two properties described in the following claim. The first property states that for any $B'\ne B$, after $B'$ is removed, the plurality scores of the remaining alternatives in $P_{B,a,b}|_{\ma\setminus B'}$ are equal. The second property states that after $B$ is removed, the plurality score of $a$ is increased by $1$, the plurality score of $b$ is decreased by $1$, and the plurality scores of  other alternatives are the same,  after $B$ is removed from $P_{B,a,b}$.

\begin{claim}[\bf\boldmath Properties of $P_{B,a,b}$]
\label{claim:P-a-b}
Let $P_{B,a,b}$ denote the profile defined in Definition~\ref{dfn:P-a-b}, 
\begin{itemize}
\item[]{(i)} for any $B'\subsetneq \ma$ with $B'\ne B$ and any $c\in \ma\setminus (B\cup\{a,b\})$, we have 
$$\score_{\plu}(P_{B,a,b}|_{\ma\setminus B'},c) = \score_{\plu}(P_{B,a,b}|_{\ma\setminus B'},a) = \score_{\plu}(P_{B,a,b}|_{\ma\setminus B'},b)$$
\item[]{(ii)} for any $c\in \ma\setminus (B\cup\{a,b\})$, we have 
$$\score_{\plu}(P_{B,a,b}|_{\ma\setminus B},c) = \score_{\plu}(P_{B,a,b}|_{\ma\setminus B},a)-1 = \score_{\plu}(P_{B,a,b}|_{\ma\setminus B},b)+1$$
\end{itemize}
\end{claim}
\begin{proof}
We first verify that (i) holds after $B' \subsetneq \ma$ is removed, for the following cases of $B'$.
\begin{itemize}
\item If $B\not\subseteq B'$, then after $B'$ is removed, at least one alternative in $B$ is not removed. According to the definition of $P_{B,a,b}$ (Definition~\ref{dfn:P-a-b}), the plurality scores of all remaining alternatives are the same, which verifies (i).
\item If $a\in B'$, then notice that $P_{B,a,b}$ only switch positions of $a$ and $b$ in some rankings where they are adjacent. Therefore, after $B'$ is removed, which means that $a$ is removed, (i) holds.
\item Similarly, we can show that if  $b\in B'$, then (i) holds.
\end{itemize}
The only remaining case is $B\subseteq B'\subseteq(\ma\setminus\{a,b\})$. Next, we prove that assuming $B\subseteq B'\subseteq(\ma\setminus\{a,b\})$, (ii) holds if and only if $B'=B$, and otherwise (i) holds (i.e., when $B\subsetneq B'\subseteq(\ma\setminus\{a,b\})$).

First, for any $c\in\ma\setminus (B'\cup\{a,b\}) $, notice that the difference between $P_{B,a,b}$ and $\ml(\ma)$ only affects the plurality scores of $a$ and $b$ after $B'$ is removed. Therefore, for any $B'$, we have
$$\score_{\plu}(P_{B,a,b}|_{\ma\setminus B'},c) = \score_{\plu}(\ml(\ma)|_{\ma\setminus B'},c)$$
Next, we calculate the difference between the plurality score of $a$ in $P_{B,a,b}$ after $B'$ is removed and the plurality score of $a$ in $\ml(\ma)$ after $B'$ is removed. This is done by examining the effect of switching the (adjacent) locations of $a$ and $b$ in some votes as described in  Definition~\ref{dfn:P-a-b}. More precisely, for every $B^*$ with $B\subseteq B^*\subseteq(\ma\setminus\{a,b\})$, we have the following observations.
\begin{itemize}
\item If $B^*\nsubseteq B'$, then  $a$ is not ranked in the top positions of $[B^*\succ a\succ b\succ \others]$ or  $[B^*\succ b\succ a\succ \others]$ after $B'$ is removed. Therefore, the difference between $P_{B,a,b}$ and $\ml(\ma)$ in the vote that corresponds to $B^*$ does not affect the difference in the plurality scores of $a$ in $P_{B,a,b}$ and $\ml(\ma)$ after $B'$ is removed.
\item If $B^*\subseteq B'$ and $|B^*\setminus B|$ is an even number, then according to the definition of $P_{B,a,b}$ (Definition~\ref{dfn:P-a-b}), $P_{B,a,b}$  replaces  $[B^*\succ b\succ a\succ \others]$ in $\ml(\ma)$ by $[B^*\succ a\succ b\succ \others]$. Notice that after $B'$ is removed, $b$ is ranked at the top position in the former ranking and $a$ is ranked at the top position in the latter ranking. Therefore, the plurality score of $a$ (respectively, $b$) is increased (respectively, decreased) by $1$ in 
$P_{B,a,b}$ after $B'$ is removed, compared to the plurality score of $a$ (respectively, $b$) in $\ml(\ma)$.
\item Similarly, if $B^*\subseteq B'$ and $|B^*\setminus B|$ is an odd number, then the plurality score of $a$ (respectively, $b$) is decreased (respectively, increased) by $1$ in 
$P_{B,a,b}$ after $B'$ is removed, compared to the plurality score of $a$ (respectively, $b$) in $\ml(\ma)$.
\end{itemize}
Therefore, we have the following calculation.
\begin{align*}
&\score_{\plu}(P_{B,a,b}|_{\ma\setminus B'},a) - \score_{\plu}(\ml(\ma)|_{\ma\setminus B'},a) \\
=&  \sum\nolimits_{B^*:B\subseteq B^*\subseteq B'} (-1)^{|B^*\setminus B|} = \sum\nolimits_{i=0}^{|B'|-|B|}(-1)^i{|B'|-|B| \choose i} =\left\{\begin{array}{ll}1&\text{if }B' =B\\
0&\text{otherwise}\end{array}\right.
\end{align*}
Similarly, we have the following calculation for $b$.
\begin{align*}
&\score_{\plu}(P_{B,a,b}|_{\ma\setminus B'},b) - \score_{\plu}(\ml(\ma)|_{\ma\setminus B'},b) = -\sum_{B^*:B\subseteq B^*\subseteq B'} (-1)^{|B^*\setminus B|} =\left\{\begin{array}{ll}-1&\text{if }B' =B\\
0&\text{otherwise}\end{array}\right.
\end{align*}
It follows that (ii) holds when $B'=B$, otherwise (i) holds. This proves Claim~\ref{claim:P-a-b}.
\end{proof}

\paragraph{\bf\boldmath  Step 2 for $\Leftarrow$.} For every $B\subsetneq \ma$, we use $P_{B,a,b}$  to fine-tune the plurality scores of alternatives when $B$ is removed, so that the weak order among them   becomes $W(B)$. For every $0\le n'\le m!-1$ such that for every $B\subsetneq \ma$, $\GCD{W(B)}\mid n'$, we construct a profile $P_{n'}$ such that $|P_{n'}|\equiv n' \pmod{m!}$ and $\pc{\stv}{P_{n'}}=W$ in the following steps.

\begin{itemize}
\item {\bf Step 2.1.} Let $n' = \lfloor \frac{n}{m!}\rfloor \times m!$. Initialize $P_{n'}$ to be $n'$ copies of $[1\succ2\succ\cdots\succ m]$.
\item {\bf Step 2.2.}  For every $B\subsetneq \ma$, let $a_B\in (\ma\setminus B)$ denote an arbitrary alternative. For every $b\in \ma\setminus (B\cup\{a_B\})$, we add $\score_{\plu}(P_{B,a_B,b}|_{\ma\setminus B},b)$ copies of $P_{B,a_B,b}$ to $P_{n'}$, which can be seen as transferring the plurality score of $b$ after to $a_B$. After this step, for every $B\subsetneq \ma$ and every $b\in \ma\setminus (B\cup\{a_B\})$, we have:
$$\score_{\plu}(P_{n'}|_{\ma\setminus B},a_B) = \score_{\plu}(P_{n'}|_{\ma\setminus B},b) + n'$$
\item {\bf Step 2.3.}  For every $B\subsetneq \ma$, we add multiple copies of $P_{B,a,b}$ for certain combinations of $(a,b)$ such that the plurality scores for the same tier are the same, while the order between tiers may not be consistent with $W(B)$. Notice that before Step 2.3, the plurality scores of alternatives in each tier are already the same, except the tier that contains the distinguished alternative $a_B$  defined in Step 2.3. Given $B\subsetneq \ma$, let $T_1\succ T_2\succ\cdots\succ T_t$ denote the tier representation of $W(B)$. Following the definition of $\GCD{W(B)}$, there exists $t$ integers $\eta_1,\ldots,\eta_t$ such that 
$$\GCD{W(B)} = \sum\nolimits_{s=1}^t \eta_s\times |T_s|$$
For every $1\le s \le t$, and every $b\in T_s$, we add the following profiles to $P_{n'}$, where we let  $P_{B,a_B,a_B}=\emptyset$ for convenience. 
\begin{itemize}
\item if $\eta_t\ge 0$, then we add $\frac{\eta_t n'}{\GCD{W(B)}}$ copies of $P_{B,a_B,b}$;
\item if $\eta_t<0$, then we add $\frac{\eta_t n'}{\GCD{W(B)}}$ copies of $P_{B,b,a_B}$.
\end{itemize}

\item {\bf Step 2.4.}  Finally, we will add multiple copies of $P_{B,a,b}$ for certain combinations of $(a,b)$ to $P_{n'}$ so that the order between tiers are the same as in $W$. Formally, for every $B\subsetneq \ma$ such that the tier representation of $W(B)$ is $T_1\succ \cdots\succ T_t$, and every $1\le s\le t-1$, we add $C_s$ copies of $\bigcup\nolimits_{a\in T_s, b\in T_t} P_{B,a,b}$ to $P_{n'}$, where $C_1,\ldots,C_t-1$ are constants to guarantee $W(B)$. This can be achieved because the effect of $\bigcup\nolimits_{a\in T_s, b\in T_t} P_{B,a,b}$ to $P_{n'}$ increases the plurality score of every alternative in $T_s$ by $|T_t|$ and reduces the plurality score of every alternative in $T_t$ by $|T_s|$, after $B$ is removed. 
\end{itemize}
It follows the construction that $|P_{n'}|\equiv n' \pmod{m!}$ and $\pc{\stv}{P_{n'}}=W$. 

\paragraph{\bf\boldmath  Step 3 for $\Leftarrow$.} Let $N=\max_{1\le n'\le m!-1}P_{n'}$. For any $n>N$, let $n'=n \pmod{m!}$, and 
$$P=P_{n'}+ \frac{n-n'}{m!}\times \ml(\ma)$$
It is not hard to verify that $P$ is an $n$-profile and $\pc{\stv}{P}=W$, which proves the $\Leftarrow$ part of Lemma~\ref{lem:construction-STV-Coombs}, and therefore completes the proof of the STV part of Lemma~\ref{lem:construction-STV-Coombs}.

\paragraph{\bf Proof for Coombs.} Notice that Coombs uses veto in each round, and the veto score of a ranking is the plurality score of the reverse ranking. In light of this connection, the proof  is similar to the STV part and the main difference is that  for Coombs, we use $P_{B,a,b}'$ that is defined from $P_{B,a,b}$ (Definition~\ref{dfn:P-a-b}) by reversing all rankings.  This completes the proof   of   Lemma~\ref{lem:construction-STV-Coombs}.
\end{proof}

\subsection{Proof of Theorem~\ref{thm:STV-Coombs}}
\label{appendix:proof-mre}
\appThm{thm:STV-Coombs}{Smoothed likelihood of ties: STV}{
Let $\mm= (\Theta,\ml(\ma),\Pi)$ be a strictly positive and closed single-agent preference model. There exists $N\in\mathbb N$ such that for any $2\le k\le m$ and any $n> N$,  
\begin{align*}
&\slt{\Pi}{\stv}{m}{k}{n}=\left\{\begin{array}{ll}0 &\text{if } \forall W\in \calW_n, |\stv(W)|\ne k
\\
\exp(-\Theta(n)) &\text{otherwise if } \calW_{\stv,k,n}^{\Pi} = \emptyset\\
\Theta\left(n^{-\frac{w_{\min}}{2}}\right) &\text{otherwise}
\end{array}\right.\\
&\ilt{\Pi}{\stv}{m}{k}{n}=\left\{\begin{array}{ll}0 &\text{if } \forall W\in \calW_n, |\stv(W)|\ne k\\
\exp(-\Theta(n)) &\text{otherwise if } \exists \pi\in\conv(\Pi) \text{ s.t. }\calW_{\stv,k,n}^{\pi} = \emptyset\\
\Theta\left(n^{-\frac{w_{\text{mm}}}{2}}\right) &\text{otherwise}
\end{array}\right.
\end{align*}
}

\begin{proof}  Like in the proof of Theorem~\ref{thm:score}, the theorem is proved by modeling the set of profiles with $k$ winners as the union of constantly many polyhedra, then applying Theorem~\ref{thm:union-poly}. As in the proof of Lemma~\ref{lem:construction-STV-Coombs}, we first present the proof for $\cor=\stv$, then comment on how to modify the proof for $\coombs$.

\paragraph{\bf Proof for STV.}  The proof proceeds in the following three steps.   In Step 1, for each PUT structure $\puts$, we define a polyhedron $\ppoly{\puts}$ that characterizes the profiles whose PUT structures are $\puts$. In Step 2, we prove properties about $\ppoly{\puts}$, in particular $\dim(\ppolyz{\puts}) = m! - \ties(\puts)$. In Step 3 we formally apply Theorem~\ref{thm:union-poly} to $\upoly = \bigcup_{W\in \calW_n: |\stv(W)|=k}\ppoly{W}$.

\paragraph{\bf \boldmath  Step 1 for STV:  Define $\bm{\ppoly{\puts}}$.} We first define the vectors that will be used in the construction.

\begin{dfn}[\bf Pairwise difference vector for STV]
\label{dfn:pairwise-difference-STV}
For any pair of different alternatives $a,b$ and any  $B\in\ma\setminus\{a,b\}$, we let $\pair_{B,a,b}$ denote the vector such that for any $R\in\ml(\ma)$, the $R$-th component of $\pair_{B,a,b}$ is $\score_{\plu}(R|_{\ma\setminus B}, a)-\score_{\plu}(R|_{\ma\setminus B}, b)$.
\end{dfn}
In words, $\pair_{B,a,b}$ represents the difference between the plurality score of $a$ and the plurality score of $b$ after $B$ is removed. We now use these vectors to construct a polyhedron that corresponds to a given PUT structure $W$.

\begin{dfn}[\boldmath $\ppoly{W}$] Given a PUT structure $W$, for any $B\subsetneq \ma$, $\ppoly{W}$ is characterized by the following constraints $\pba{W}\cdot\invert{\vec x_\ma}\le \invert{\vec b}$: for any pair of alternatives $a\succ_{\puts(B)} b$, there is an inequality $\pair_{B,b,a}\cdot\vec x_\ma\le -1$; and for any pair of alternatives $a\equiv_{W_{B}} b$, there is an inequality $\pair_{B,b,a}\cdot\vec x_\ma\le 0$.
\end{dfn}

\paragraph{\bf \boldmath  Step 2 for STV: Prove properties about $\ppoly{\puts}$.} We prove the following claim about $\ppoly{W}$. 
\begin{claim}[\boldmath \bf Properties of $\ppoly{W}$]\label{claim:solPUTconfiguration}
For any PUT structure $\puts$, we have:
\begin{enumerate}[label=(\roman*)]
\item for any profile $P$, $\hist(P)\in \ppoly{W}$ if and only if $\pc{\stv}{P} = \puts$;
\item for any $\vec x\in\mathbb R^{m!}$, $\vec x\in {\ppolyz{W}}$ if and only if $\puts$ refines $\pc{\stv}{\vec x}$;
\item $\dim(\ppolyz{W})= m! - \ties(W)$.
\end{enumerate}
\end{claim}
\begin{proof}
(i) and (ii) follow after the definition of $\ppoly{\puts}$. To prove part (iii), again, according to equation (9) on page 100 in~\citep{Schrijver1998:Theory}, it suffices to prove that the rank of the essential equalities of $\pba{\puts}$, denoted by $\ba^{=}$, is $\ties(W)$. We first prove that 
\begin{equation}
\label{equ:mrse-essential-equalities}\ba^{=} =\{\pair_{B,a,b}:\forall B\subsetneq \ma\text{ and }a\equiv_{\puts(B)} b\},
\end{equation}
where $a\equiv_{\puts(B)} b$ means that $a$ and $b$ are tied in $\puts(B)$. Clearly $\ba^{=}$ contains the right hand side of Equation (\ref{equ:mrse-essential-equalities}). By Lemma~\ref{lem:construction-STV-Coombs}, there exists an $(m!)$-profile $P^*$ such that $\pc{\stv}{P^*}=W$. Therefore,   $\hist(P^*)$ is an inner point of $\ppoly{W}$ in the sense that all inequalities not mentioned in the right hand side of Equation (\ref{equ:mrse-essential-equalities}) are strict under $\hist(P^*)$, which means that Equation (\ref{equ:mrse-essential-equalities}) holds.

We now prove that $\rank(\ba^=) = \ties(\puts)$.  For any $B\subsetneq \ma$, let $\puts(B)  = T_1\succ \cdots\succ T_s$ denote the tier presentation of $\puts(B)$. Let $\ba_B$  be the sub-matrix of $\ba^=$ that consists of $\pair_{B,b,a}$ for all pairs of alternatives $a\equiv_{\puts(B)}b$. 
It is not hard to check that $\rank(\ba_B) \le  \ties(\puts)$, because for any $T_i$, each vector in $\{\pair_{B,a,b}:a,b\in T_i\}$ can be represented as the linear combination of a subset of $|T_i|-1$ linear orders. This means that $\rank(\ba^=) \le \sum_{B\subsetneq \ma}\ties(\puts(B)) =\ties(\puts)$.

$\rank(\ba^=) \ge \ties(\puts)$ follows after Lemma~\ref{lem:construction-STV-Coombs}. More precisely, suppose for the sake of contradiction that $\rank(\ba^=) < \ties(\puts)$ and let $A$ denote an arbitrary row basis of $\ba^=$. Due to the pigeon hole principle, there exists $B\subsetneq \ma$ such that $A$ contains no more than $\ties(\puts(B))$ rows in $\ba_B$, which means that there exists a tier $T_i$ such that $A$ contains no more than $|T_i|-2$ rows in $\{\pair_{B,a,b}:a,b\in T_i\}$. This means that $T_i$ can be partitioned into two sets $\{a_1,\ldots,a_{t_1}\}$ and $\{b_1,\ldots,b_{t_2}\}$ such that $A$ does not contain $\pair_{B,a_{i_1},b_{i_2}}$ for any combination  of $1\le i_1\le t_1$ and $1\le i_2\le t_2$.   By Lemma~\ref{lem:construction-STV-Coombs}, there exists an $(m!)$-profile $P$ such that 
\begin{itemize}
\item for all $B'\ne B$, all alternatives are tied after $B'$ is removed, and
\item $\pc{\cor}{P}(B)$ consists of two tiers: the first tier is $T_1\cup\cdots\cup T_{i-1}\{a_1,\ldots,a_{t_1}\}$ and the second tier is $ \{b_1,\ldots,b_{t_2}\}\cup T_{i+1}\cup\cdots\cup T_{t}$. 
\end{itemize}
It is not hard to check that 
$$ A\cdot\invert{\hist(P)} = \invert{\vec 0} \text{ but }\pair_{B,a_1,b_1}\cdot \hist(P)\ne 0.$$ 
This means that $\pair_{B,a_1,b_1}\in \ba^=$ is not a linear combination of rows in $A$, which contradicts the assumption that $A$ is a basis of $\ba^=$.

Therefore, $\rank(\ba^=) = \ties(W)$. This completes the proof of Claim~\ref{claim:solPUTconfiguration}.\end{proof}

\paragraph{\bf \boldmath Step 3 for STV: Apply Theorem~\ref{thm:union-poly}.}  Let $\upoly = \bigcup_{\puts\in \calW_n: |\stv(\puts)|=k}\ppoly{\puts}$. It follows that for any profile $P$, $|\stv(P)| = k$ if and only if $\hist(P)\in \upoly$. Therefore, for any $n$ and any $\vec\pi\in\Pi^n$, we have:
$$\Pr\nolimits_{P\sim\vec\pi} (|\stv(P)| = k) = \Pr(\vXp\in \upoly)$$
Recall that for all PUT structure $\puts$ over $\ma$, by Claim~\ref{claim:solPUTconfiguration} (iii) we have $\dim(\ppoly{\puts}) = m!-\ties(P)+1$. This means that $\alpha_n =m!-\ell_{\min}$ and $\beta_n = m!-\ell_{\text{mm}}$ when $\puts\in \calW_n$. 

Therefore, to prove Theorem~\ref{thm:STV-Coombs}, it suffices to prove that the conditions for the $0$, exponential, and polynomial cases in Theorem~\ref{thm:STV-Coombs} are equivalent to the conditions for the $0$, exponential, and polynomial cases in Theorem~\ref{thm:union-poly} (applied to $\upoly$ and $\Pi$), respectively.  This follows a similar reasoning as in Step 3 of the proof of Theorem~\ref{thm:score} combined with Claim~\ref{claim:solPUTconfiguration} (ii).
This completes the proof of the STV part of Theorem~\ref{thm:STV-Coombs}.

\paragraph{\bf Proof for Coombs.} As in the proof of the Coombs part of Lemma~\ref{lem:construction-STV-Coombs}, notice that Coombs uses veto in each round, and the veto score of a ranking is the plurality score of the reverse ranking. Therefore, the proof is similar to the STV part and the main difference is that  for Coombs, when defining the pairwise score difference vectors (Definition~\ref{dfn:pairwise-difference-STV}), we use the veto rule instead of the plurality rule.  This completes the proof   of   Lemma~\ref{lem:construction-STV-Coombs}.
\end{proof}

\subsection{Proof of Proposition~\ref{prop:stv-Coombs}}
\label{app:proof-stv}

\appProp{prop:stv-Coombs}{Max smoothed likelihood of ties: STV and Coombs}{
Let $\cor\in\{\stv,\coombs\}$ and  let $\mm= (\Theta,\ml(\ma),\Pi)$ be a strictly positive and closed single-agent preference model  with $\piuni\in \conv(\Pi)$. For any $2\le k\le m$ and any $n\in\mathbb N$,   
$$\slt{\Pi}{\cor}{m}{k}{n}=\left\{\begin{array}{ll}\Theta(n^{-\frac{k-1}{2}})&\text{if }\left\{\begin{array}{l} (1) m\ge 4 \text{, or } \\
(2) m=3\text{ and }k=2, \text{ or }\\
(3) m=k=3\text{ and }(2\mid n \text{ or } 3\mid n)\end{array}\right.\\
0 &\text{otherwise (i.e., }m=k=3\text{,  }2\nmid n,\text{ and }3\nmid n\text{)}\end{array}\right.
$$
}
\begin{proof} According to  Theorem~\ref{thm:STV-Coombs}, we will prove that (1) the exponential case does not happen,  and (2) $w_{\min} = k-1$ in the polynomial cases.

To see (1), notice that if the $0$ case of Theorem~\ref{thm:STV-Coombs} does not hold, then there exists $W\in \calW_n$ such that $|\stv(W)|=k$. Also notice that $\piuni\in \conv(\Pi)$ and $\pc{\stv}{\piuni}$ is the PUT structure where all alternatives are tied after any $B\subsetneq \ma$ is removed, which means that any PUT structure is a refinement of $\pc{\stv}{\piuni}$. Therefore,  $\calW_{\cor,k,n}^{\Pi}\ne \emptyset$, which means that the exponential case of Theorem~\ref{thm:STV-Coombs} does not hold.

To see (2), notice that for any PUT structure $W$, the outgoing edges of any node labeled $B\subsetneq \ma$ with $|B|\le m-2$, the number of outgoing edges of $B$ is $\ties(W(B))+1$. Therefore,  $\ties(W)$ is at least the total number of  nodes in $W$ that do not have outgoing edges minus one. It follows that for any PUT structure $W$ with $|\stv(W)|=k$, we must have $\ties(W)\ge k-1$, which means that $w_{\text{mm}}\ge k-1$. 

In the sequel, we will prove $w_{\text{mm}}=k-1$ by construction for the polynomial cases described  in  Proposition~\ref{prop:stv-Coombs}. More precisely, we will prove that there exists an $n$-profile $P$ such that $\ties(\pc{\stv}{P})=k-1$ and $|\stv(\pc{\stv}{P})|=k$. In light of Lemma~\ref{lem:construction-STV-Coombs}, it suffices to construct a PUT structure $W$ such that 
$$\text{(1) }\ties(W)=k-1\text{, (2) }|\stv(W)|=k\text{,  and (3) for every }B\subsetneq \ma, \GCD{W(B)}\mid n$$

\paragraph{\bf \boldmath  $ \Theta(n^{-\frac{k-1}{2}})$ subcase (1): $m\ge 4$.} When $m\ge 4$, we let $W$ denote the PUT structure such that 
\begin{align*}
&W(\emptyset) = k\succ k+1\succ\cdots\succ m\succ \{1,\ldots,k-1\}\\
&W(\{1\}) = 2\succ \cdots \succ k-2\succ k+1\succ \cdots \succ m\succ \{k-1,k\},
\end{align*}
and $\puts(B)$ for any other $B$ is a linear order so that in the graph representation of $\puts$, there are $k$ nodes without outgoing edges, whose winners are $\{1,\ldots,k\}$. See Figure~\ref{fig:STVPUTgraph} (a) for an example of $m=k=4$.

\begin{figure}[htp]
\centering
\begin{tabular}{cc}
\begin{minipage}[b]{0.58\linewidth}
\centering
\includegraphics[width = \textwidth]{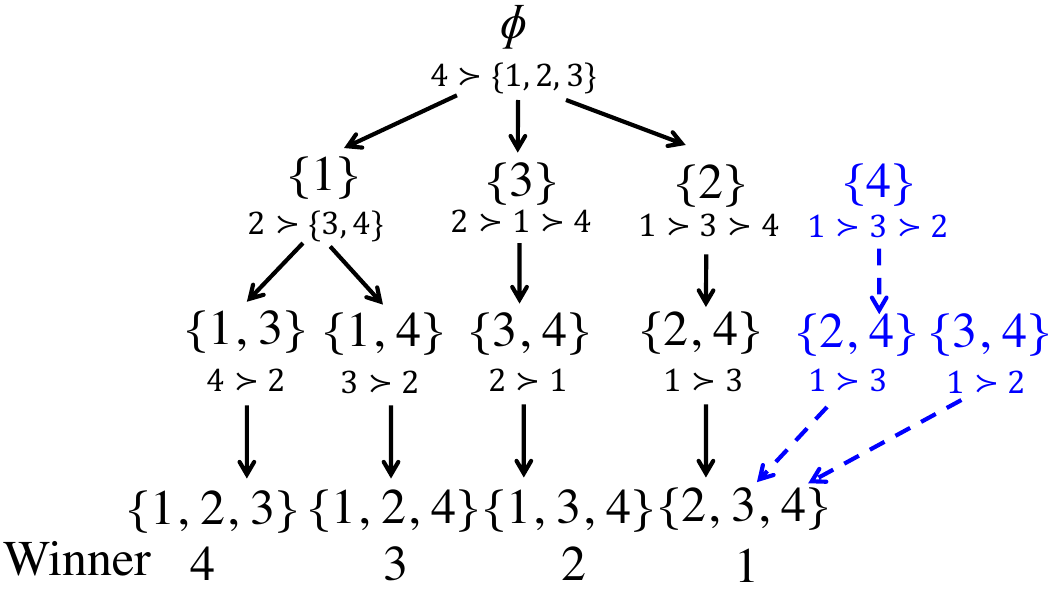}
\end{minipage}
&
\begin{minipage}[b]{0.42\linewidth}
\centering
\includegraphics[width = \textwidth]{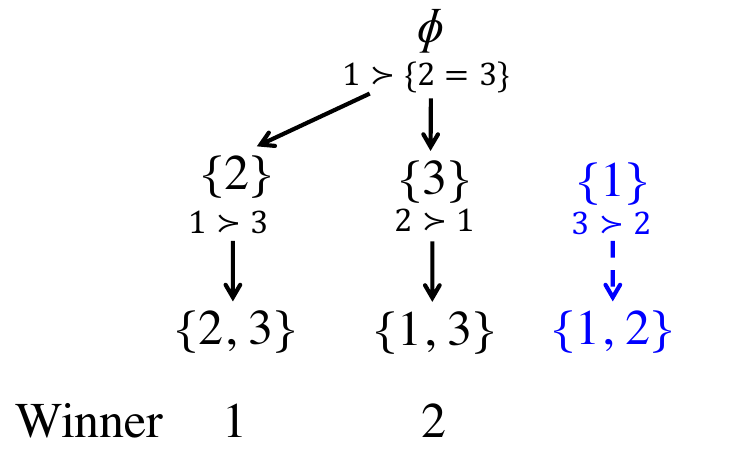}
\end{minipage}
\\
(a) $m=k=4$. & (b) $m=3$ and $k=2$.
\end{tabular}
\caption{PUT structures for STV.
\label{fig:STVPUTgraph}}
\end{figure}

It is not hard to verify that $\ties(W)=k-1$ and $|\stv(W)|=k$. Notice that for every $B\subsetneq \ma$, there exists a tier in $W(B)$ that consists of a single alternatives, which means that $\GCD{W(B)}=1$ and therefore, $W\in\calW_n$.  This proves the proposition for the $m\ge 4$ case.

\paragraph{\bf \boldmath  $ \Theta(n^{-\frac{k-1}{2}})$ subcase (2): $m =3$ and $k=2$.} The  $m =3$ and $k=2$ case is proved by the PUT structure illustrated in Figure~\ref{fig:STVPUTgraph} (b).

\paragraph{\bf \boldmath  $ \Theta(n^{-\frac{k-1}{2}})$ subcase (3): $m =k=3$ and ($2\mid n$ or $3\mid n$).}  The case where $m = k=3$ and $2\mid n$ is proved by the PUT structure $W_1$ illustrated in Figure~\ref{fig:STVPUTgraphm=3} (a). 
The case where $m = k=3$ and $3\mid n$ is proved by the PUT structure $W_2$ illustrated in Figure~\ref{fig:STVPUTgraphm=3} (b). 

\begin{figure}[H]
\centering
\begin{tabular}{cc}
\begin{minipage}[b]{0.4\linewidth}
\centering
\includegraphics[width = \textwidth]{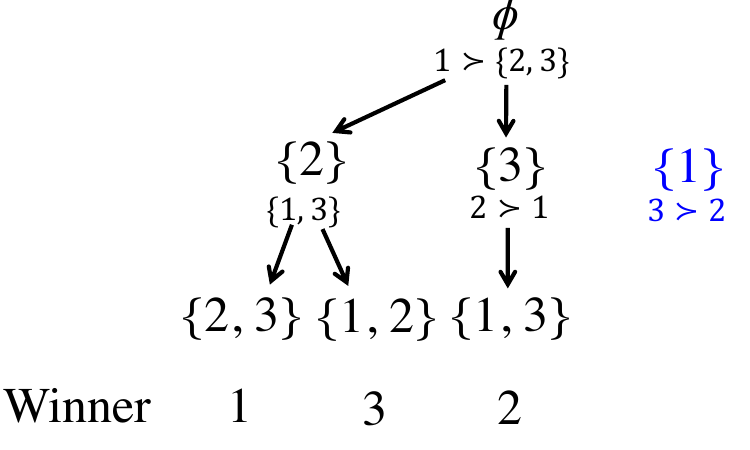}
\end{minipage}
&
\begin{minipage}[b]{0.4\linewidth}
\centering
\includegraphics[width = \textwidth]{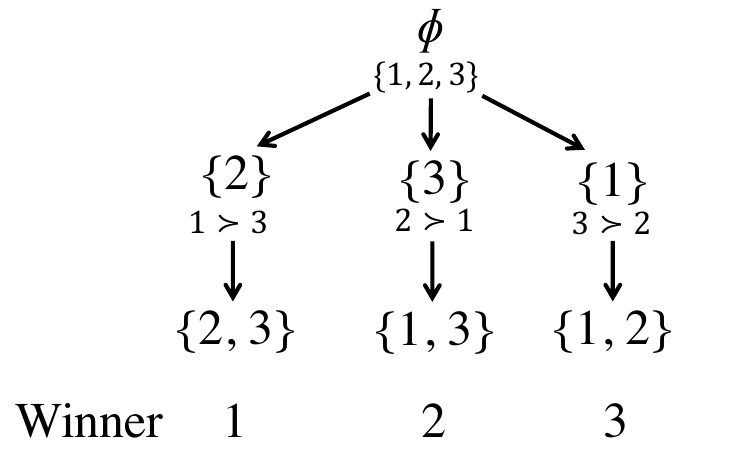}
\end{minipage}

\\
(a) $2\mid n$. & (b) $3\mid n$.
\end{tabular}
\caption{PUT structures for STV.
\label{fig:STVPUTgraphm=3}}
\end{figure}

\paragraph{\bf \boldmath  The $0$ case.} When $m = k=3$, $2\nmid n$, and $3\nmid n$, suppose for the sake of contradiction there exists an $n$-profile $P$ such that $|\stv(P)|=3$. Let $W = \pc{\stv}{P}$. Then, either $\puts(B)$  is a tie among the three alternatives for some $B\subsetneq \ma$ (which means that $\GCD(W(B))=3$, and therefore $3\mid n$), or $\puts(B)$ is a tie among two alternatives for some $|B|=1$ (which means that $\GCD(W(B))=2$, and therefore $2\mid n$). Either case leads to a contradiction. Therefore, $\calW_{\cor,k,n}^{\Pi}=\emptyset$ in this case.
\end{proof}

\section{Experimental Results for All Rules}
\label{sec:exp-full}
\begin{figure}[htp]
\centering
  \includegraphics[width = 0.8\linewidth]{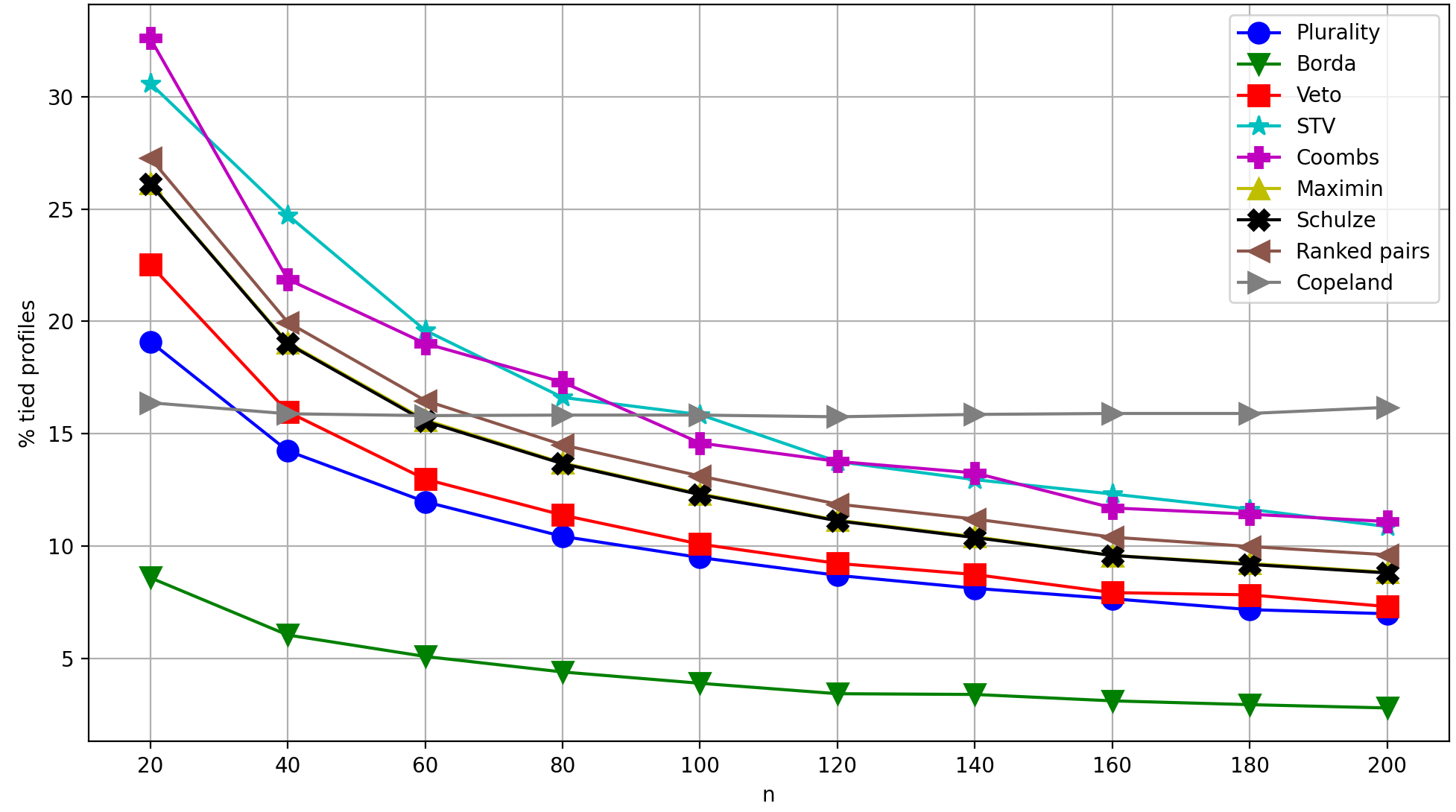}
\caption{\small Percentage of tied profiles under IC. \label{fig:exp-full}
}
\end{figure}

\begin{figure}[htp]
\centering
\begin{tabular}{cc}
  \includegraphics[width = 0.5\linewidth]{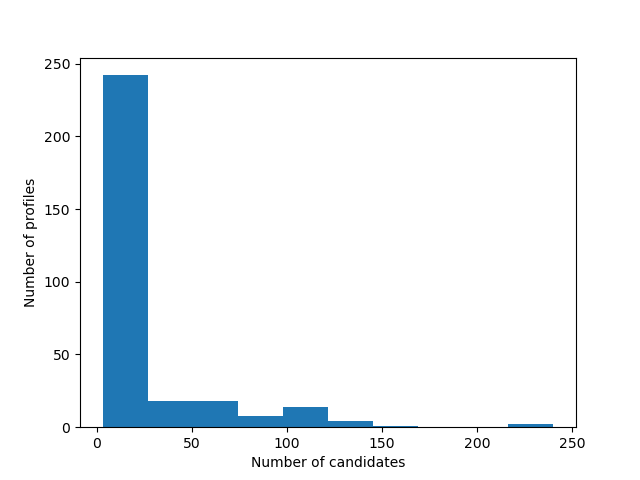} &
    \includegraphics[width = 0.5\linewidth]{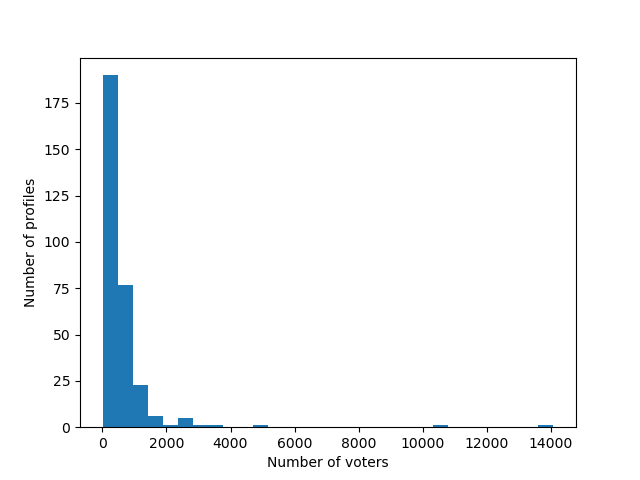} 
\end{tabular}
\caption{\small Histograms of number of candidates and number of voters in the 307 Preflib SOC data studied in this paper. \label{fig:ex-histograms}}
\end{figure}

\end{document}